\documentclass[a4paper, 11pt, leqno]{amsart} 

\usepackage[utf8]{inputenc}
\usepackage[T1]{fontenc}	

\usepackage{lmodern}			

\usepackage{graphicx}	

\usepackage[plainpages=false,colorlinks,linkcolor=bleuFonce,citecolor=rougeFonce,urlcolor=vertFonce,breaklinks,unicode,bookmarks=true]{hyperref}

\usepackage{placeins}
\usepackage{float} 

\usepackage{tikz}
\usepackage{makecell}
\usepackage{dsfont}		
\usepackage{mathtools}
\usetikzlibrary{patterns}
\usetikzlibrary{decorations.pathreplacing}

\usepackage{etoolbox}
\usepackage{hyperref} 

\newcounter{Hcounter}

\newenvironment{HypothesisList}{%
  \setcounter{Hcounter}{0} 
  \begin{enumerate}
}{%
  \end{enumerate}
}

\newcommand{\Hitem}{\stepcounter{Hcounter}\item}

\usepackage[normalem]{ulem}

\usepackage{color}
\definecolor{vertFonce}	{rgb}{0,0.5,0}
\definecolor{numLignes}	{rgb}{0.17,0.57,0.7}	
\definecolor{gris}		{rgb}{0.5,0.5,0.5}
\definecolor{grisFonce}	{rgb}{0.2,0.2,0.2}
\definecolor{orange}	{rgb}{1,0.65,0.31}		
\definecolor{orangeFonce}{rgb}{1,0.4,0}
\definecolor{bleuFonce}	{rgb}{0,0,0.4}
\definecolor{rougeFonce}{rgb}{0.3,0,0}
\definecolor{rougeWord}	{rgb}{0.5,0,0}
\definecolor{vertClair}	{rgb}{0.8,1,0.8}
\definecolor{rougeClair}{rgb}{1,0.5,0.5}
\definecolor{violet}	{rgb}{0.5,0,0.5}

\usepackage{pict2e}
\usepackage{multido}
\usepackage{upgreek}

\usepackage{amsfonts,amssymb,amsthm,amsmath}
\usepackage{dsfont}				
\usepackage{mathrsfs}
\usepackage[mathscr]{euscript}
\usepackage{yfonts}
\usepackage{cancel}
\usepackage{enumerate}

\theoremstyle{plain}
\newtheorem{theorem}{Theorem}[section]
\newtheorem{lem}[theorem]{Lemma}

\newtheorem{proposition}[theorem]{Proposition}
\newtheorem{cor}[theorem]{Corollary}
\newtheorem{prop}[theorem]{Proposition}

\theoremstyle{definition}

\theoremstyle{remark}
\newtheorem{remark}{Remark}[section]
\newtheorem{remarks}{Remarks}[section]

\newenvironment{thm}[1][]{
\begin{theorem}[#1]
}{
\end{theorem}
}

\usepackage{stmaryrd}
\SetSymbolFont{stmry}{bold}{U}{stmry}{m}{n}
\newcommand		{\subsetArrow}	{\mathrel{\ooalign{$\subset$\cr
		\hidewidth\raise-.087ex\hbox{$_\shortrightarrow\mkern-1.5mu$}\cr}}}
\newcommand		{\subsetarrow}	{\mathrel{\ooalign{$\subset$\cr
		\hidewidth\raise-1.45ex\hbox{$\vec{}\mkern6mu$}\cr}}}

\newcommand		{\N}		{\mathbb N}			
\newcommand		{\RR}		{\mathbb R}			
\newcommand		{\R}		{\RR}

\newcommand		{\Rd}		{\R^3}

\newcommand		{\CC}		{\mathbb C}			
			
\newcommand		{\cH}		{\mathcal H}		
		
\newcommand		{\cM}		{\mathcal M}		
\newcommand		{\cN}		{\mathcal N}		
\newcommand		{\cP}		{\mathcal P}

\newcommand		{\cK}		{\mathcal K}		
		
\newcommand		{\cW}		{\mathcal W}		

\newcommand		{\cA}		{\mathcal A}
\newcommand		{\cB}		{\mathcal B}
\newcommand		{\cE}		{\mathcal E}
\newcommand		{\cG}		{\mathcal G}
\newcommand		{\cI}		{\mathcal I}

\newcommand		{\cO}		{\mathcal O}
\newcommand		{\cQ}		{\mathcal Q}

\newcommand		{\cT}		{\mathcal T}
\newcommand		{\cU}		{\mathcal U}
\newcommand		{\cV}		{\mathcal V}
\newcommand		{\cL}		{\mathcal L}

\newcommand		\sfG		{\mathsf G}

\newcommand		\sfJ		{\mathsf J}

\newcommand		\sfL		{\mathsf L}			
			
\newcommand		\sfR		{\mathsf R}

\newcommand		{\lt}			{\left}				
\newcommand		{\rt}			{\right}			
\renewcommand	{\(}			{\lt(}
\renewcommand	{\)}			{\rt)}
\newcommand		{\bangle}[1]	{\lt\langle #1\rt\rangle}
\newcommand		{\weight}[1]	{\bangle{#1}}	

\newcommand		{\inprod}[2]	{\bangle{#1, #2}}
\newcommand		{\com}[1]		{\lt[{#1}\rt]}

\newcommand		{\n}[1]			{\lt\lvert #1 \rt\rvert}
\newcommand		{\norm}[1]		{\big\lVert #1 \big\rVert}		
\newcommand		{\nrm}[1]		{\lt\lVert #1\rt\rVert}
\newcommand		{\Nrm}[2]		{\nrm{#1}_{#2}}

\renewcommand		{\d}		{\mathrm{d}}		
\newcommand			{\dd}		{\,\d}				
\newcommand			{\bd}		{\partial}

\newcommand			{\grad}		{\nabla}
\newcommand			{\lapl}		{\Delta}

\newcommand			{\conj}[1]	{\overline{#1}}

\DeclareMathOperator{\cF}		{\mathcal{F}}

\DeclareMathOperator{\re}		{Re}				
\DeclareMathOperator{\im}		{Im}				
				
\DeclareMathOperator{\tr}		{Tr}

\DeclareMathOperator{\Nor}		{Nor}

\newcommand			{\ch}		{\operatorname{ch}}
\newcommand			{\sh}		{\operatorname{sh}}

\newcommand		{\nor}[1]		{\Nor\!\lt\{ #1 \rt\}}

\newcommand		{\intd}			{\int_{\Rd}}

\newcommand		{\iintd}		{\iint_{\Rd\times\Rd}}

\newcommand		{\ii}			{\mathrm{i}}	
\newcommand		{\jj}			{\mathrm{j}}	
\newcommand		{\init}			{\mathrm{in}}

\newcommand		{\eps}			{\varepsilon}

\usepackage{braket}

\newcommand		{\wt}		{\widetilde}

\newcommand{\vect}[1]{\boldsymbol{\mathbf{#1}}}

\newcommand		{\h}		{\mathfrak{h}}		

\newcommand		{\fH}		{\mathfrak{H}}

\newcommand		\sfK		{\mathsf K}

\newcommand		\sfM		{\mathsf M}
	
\newcommand		\sfS		{\mathsf S}	
\newcommand		\sfW		{\mathsf W}

\newcommand     \sfn        {\mathsf n}

\newcommand{\id}{\mathds{1}}

\newcommand{\dGamma}{\d\Gamma}

\newcommand {\asc}      {\mathfrak{a}}
\newcommand {\bsc}      {\mathfrak{b}}
\newcommand {\cc}      {\mathfrak{c}}

\newcommand		{\wtbphi}	{\boldsymbol{\wt\phi}}
\newcommand		{\bphi}	{\boldsymbol{\phi}}

\numberwithin{equation}{section}

\begin{document}

\title[Quantitative Derivation of the Two-Component GP Equation]{Quantitative Derivation of the Two-Component Gross--Pitaevskii Equation in the Hard-core Limit with Uniform-in-Time Convergence Rate}

\author{Jacky Chong}
\address[J. Chong]{School of Mathematical Sciences, Peking University, Beijing, China}
\email{jwchong@math.pku.edu.cn}    

\author{Jinyeop Lee}
\address[J. Lee]{Department of Mathematics, University of British Columbia, 1984 Mathematics Rd., Vancouver, BC, Canada V6T 1Z2}
\email{lee@math.ubc.ca}
\thanks{Corresponding author: Jinyeop Lee, 
\texttt{lee@math.ubc.ca}}

\author{Zhiwei Sun}
\address[Z. Sun]{Institute of Analysis and Scientific Computing, TU Wien, Wiedner Hauptstraße 8--10, 1040 Wien, Austria}
\email{zhiwei.sun@asc.tuwien.ac.at}


\begin{abstract}
We derive the time-dependent two-component Gross--Pitaevskii (GP) equation as an effective description of the dynamics of a dilute two-component Bose gas near its ground state, which exhibits a two-component Bose--Einstein condensate, in the GP limit. Our main result establishes a uniform-in-time bound on the convergence rate between the many-body dynamics and the effective description, explicitly quantified in terms of the particle number $N$, and also implies a uniform-in-time bound for the one-component case. This improves upon the works of Michelangeli and Olgliati \cite{MichelangeliOlgliatiBEC, olgiati2017effective} by providing a sharper, $N$-dependent, time-independent convergence rate. Our approach further extends the framework of Benedikter, de Oliveira, and Schlein \cite{Benedikter2015quantitative} to the multi-component Bose gas in the hard-core limit setting. More specifically, we develop the necessary Bogoliubov theory to analyze the dynamics of multi-component Bose gases in the GP regime.  
\end{abstract}


\maketitle

\section{Background and Main Result}

\subsection{Setting} Following the first successful experimental realization of Bose--Einstein condensate (BEC) in a gas of $^{87}$Rb \cite{WC1995BEC}, BECs have garnered significant attention from both theoretical and experimental researchers. Subsequently, two-component mixture BECs have also been successfully observed in gases of atoms of the
same element, typically $^{87}$Rb, which occupy two hyperfine states \cite{myatt1997production, hall1998dynamics},  and in heteronuclear mixture systems such as $^{41}$K–$^{87}$Rb \cite{modugno2001bose}, $^{41}$K–$^{85}$Rb \cite{modugno2002two}, $^{39}$K–$^{85}$Rb \cite{mancini2004observation} , and $^{85}$Rb–$^{87}$Rb \cite{papp2006observation}. For a review of the physical
properties of quantum mixtures we refer to \cite{pitaevskii2016bose}. 

We consider a two-component Bose gas consisting of $N_1$ bosons of the first species and $N_2$ bosons of the second species, with a total particle number $N = N_1 + N_2$, and assume that the ratios $N_\ii / N$ converge to $\sfn_\ii \in (0, 1)$ as $N \to \infty$.
The particles are confined by a trapping potential $W_{\rm trap}(x)$ to a volume of order one and interact through repulsive, short-range potentials $V_{\ii}$ with effective range of order $1/N$. Denoting $V^{\lambda}_{\ii} := \lambda V_{\ii}$ for some coupling constant $\lambda \ge 1$ (possibly $N$‑dependent), we model the system by the Hamiltonian
\begin{equation}\label{eq:2species_trapped_Hamiltonian} 
	\begin{aligned}
		H_{N_1, N_2}^{\rm trap} =&\, \sum_{j=1}^{N_1} \(-\lapl_{x_j}+W_{\rm trap}(x_j)\)+ \sum_{j<k}^{N_1} N^2 
         V^\lambda_{1}(N (x_k - x_j))\\
		&\, +\sum_{j=1}^{N_2} \(-\lapl_{y_j}+W_{\rm trap}(y_j)\)+ \sum_{j<k}^{N_2} N^2 V^\lambda_{2}(N (y_k - y_j))  \\
		&\, + \sum^{N_2}_{j=1}\sum_{k=1}^{N_1} N^2 V^\lambda_{12}(N (x_k - y_j)),
	\end{aligned}
\end{equation}
which acts on the Hilbert space $\fH_{N_1, N_2}=L^2_s (\RR^{3N_1})\otimes L^2_s(\RR^{3N_2})$, the subspace of $L^2 (\RR^{3N})$ consisting of wave functions that are symmetric with respect to intra-species particle permutations, that is, 
\begin{align*}
	\psi_{N}(x_1, \ldots, x_{N_1}; y_1, \ldots, y_{N_2})= \psi_{N}(x_{\sigma(1)}, \ldots, x_{\sigma(N_1)}; y_{\pi(1)}, \ldots, y_{\pi(N_2)})
\end{align*}
for every $\sigma \in \mathfrak{S}_{N_1}$ and $\pi \in \mathfrak{S}_{N_2}$, where $\mathfrak{S}_n$ is the symmetric group of $n$ elements. Moreover, in the rest of the paper, we make the following assumption on $V_\ii$ for $\ii\in\{1,2,12\}$.
\begin{HypothesisList}
    \Hitem \label{assume: L3} The interaction potential $V_\ii$ satisfies $V_\ii \in L^1(\mathbb{R}^3) \cap L^3(\mathbb{R}^3)$ 
    and is a positive radial function with compact support $ \{ x \in \mathbb{R}^3 :\n{x} \leq \bsc_\ii \}$.
    \Hitem \label{assume: boundedness} The potential $V_\ii$ is essentially bounded, i.e., $V_\ii \in L^\infty(\mathbb{R}^3)$.
\end{HypothesisList}
Our analysis employs Assumption~\ref{assume: L3} for the case of $\lambda = 1$, 
and both Assumptions~\ref{assume: L3} and \ref{assume: boundedness} in the limiting regime $\lambda \to \infty$.
Hence, Assumption~\ref{assume: L3} is always satisfied throughout this work.

An important quantity to consider is the scattering length $\asc \ge 0$ of an interaction potential $V$ defined through the solution $f$ of the zero-energy scattering problem 
\begin{align}\label{eq:zero-energy_scattering}
	\(-\lapl+ \tfrac12V \) f = 0, 
    \quad \text{ with }  \quad
    \lim_{\n{x}\rightarrow \infty} f(x) = 1,
\end{align}
via the expression
\begin{equation}
    8\pi \asc:= \intd V(x) f(x)\dd x\ ,
\end{equation}
so that $f(x) = 1 - \asc/\n{x}$ outside the support of $V$. 
The scattering lengths $\asc_\ii$ of $V_\ii$ play important roles in determining of the ground state energy of the system $E_{N_1, N_2}^{\rm trap}$ defined by
\begin{equation*}
    E_{N_1, N_2}^{\rm trap}
    :=
    \min_{\substack{\psi_{N} \in \fH_{N_1, N_2} \\  \Nrm{\psi_{N}}{L^2}^2=1 }} 
    \langle \psi_{N},\, H_{N_1, N_2}^{\rm trap} \psi_{N}  \rangle.
\end{equation*}
In \cite{michelangeli2019ground}, it was proven, to leading order in $N$, the ground state energy with $\lambda=1$ fixed, satisfies the property
\begin{equation}\label{eq:conv1} 
	\lim_{N \to \infty} \frac{E_{N_1, N_2}^{\rm trap}}{N} = \min_{\substack{u, v \in H^1(\RR^3) \\  \Nrm{u}{L^2}^2=1,\, \Nrm{v}{L^2}^2=1}} \cE_{\rm GP}[u, v] =: e_{\rm GP}
\end{equation} 
with the two-component Gross--Pitaevskii (GP) energy functional 
\begin{equation}\label{eq:GP_energy_functional}
	\begin{aligned}
		\cE_{\rm GP}[u, v]  =&\, \intd \sfn_1\n{\nabla u(x)}^2 +   W_{\rm trap} \sfn_1\n{u(x)}^2 + 4 \pi  \asc_{1} \sfn_1^2\n{u(x)}^4 \dd x\\
		&\, +\intd \sfn_2\n{\nabla v(x)}^2 + W_{\rm trap} \sfn_2\n{v(x)}^2 + 4 \pi  \asc_{2}\sfn_2^2\n{v(x)}^4 \dd x\\
		&\, +\intd 8\pi \asc_{12}\sfn_1\sfn_2 \n{u(x)}^2\n{v(x)}^2\dd x \, .
	\end{aligned}
\end{equation} 
provided $\asc_{1}\asc_{2}-\asc_{12}^2\ge 0$, which is called the \emph{miscibility condition}.
Note that the miscibility condition guarantees the existence and uniqueness of the minimizer to the two-component GP functional. 

Let $( u_{\rm GP},  v_{\rm GP}) \in L^2 (\RR^3; \CC^2)$ denote the normalized minimizer (unique, up to a phase) of the GP functional~\eqref{eq:GP_energy_functional}. It turns out that the ground state of the Hamiltonian~\eqref{eq:2species_trapped_Hamiltonian} and, in fact, every sequence of approximate ground states exhibits a complete two-component BEC in the states $u_{\rm GP}$ and $v_{\rm GP}$. More precisely, let us consider a normalized sequence $\psi_{N} \in \h_{N_1, N_2}$ satisfying 
\begin{align*}
	\frac{1}{N} \inprod{\psi_{N}}{H_{N_1, N_2}^\text{trap} \psi_{N}} \to e_{\rm GP}
\end{align*}
as $N \to \infty$ (i.e., $\psi_{N}$ is a sequence of approximate ground states). Let $\gamma_{N}^{(k, \ell)}$ denote the $(k,\ell)$-particle reduced density matrix associated with $\psi_{N}$ and $k, \ell \in \N_{0}$, which is defined as the non-negative trace class operator on $L^2 (\RR^{3k})\otimes L^2(\RR^{3\ell})$ with the integral kernel 
\begin{align}\label{def:reduced_density}
	\gamma^{(k,\ell)}_{N}(X,Y; X',Y')
	:=\, \int_{\mathbb{R}^{3(N_1-k)}}\int_{\mathbb{R}^{3(N_2-\ell)}}\psi_{N}(X,Z;Y,Z')\, \overline{\psi_{N}(X',Z;Y',Z')} \dd Z\d Z'
\end{align}
where $X=(x_1, \ldots, x_{k})\in \RR^{3k}$ and $Y=(y_1, \ldots, y_{\ell})\in \RR^{3\ell}$. 
By letting $\norm{\psi_N}_{L^2}=1$, we normalize $\gamma_{N}^{(k, \ell)}$, that is, $\tr_{L^2 (\RR^{3k})\otimes L^2(\RR^{3\ell})}(\gamma_{N}^{(k, \ell)}) = 1$. Then, it was first proved in \cite{michelangeli2019ground} that 
\begin{equation}\label{eq:BEC} 
	\lim_{N \to \infty} \inprod{u_{\rm GP}^{\otimes k}\otimes v_{\rm GP}^{\otimes\ell}}{\gamma_{N}^{(k, \ell)} u_{\rm GP}^{\otimes k}\otimes v_{\rm GP}^{\otimes\ell}}= 1.
\end{equation}

When the trapping potential in the Hamiltonian~\eqref{eq:2species_trapped_Hamiltonian} is switched off (so that the system is no longer in equilibrium), the time evolution of system is governed by the $N$-particle Schr\"{o}dinger equation of two-species:
\begin{equation}\label{eq:linear_schrodinger}  
	i \partial_t \psi_{N, t} = H_{N_1, N_2} \psi_{N, t}
\end{equation} 
with the corresponding translation-invariant Hamiltonian
\begin{equation}\label{eq:2species_Hamiltonian} 
	\begin{aligned}
		H_{N_1, N_2} =&\, \sum_{j=1}^{N_1} -\lapl_{x_j} +\sum_{j=1}^{N_2} -\lapl_{y_j}
		+ \sum_{k<j}^{N_1} N^2 V^{\lambda}_{1}(N (x_k - x_j)) \\ &\, + \sum_{k<j}^{N_2} N^2 V^{\lambda}_{2}(N (y_k - y_j)) + \sum^{N_2}_{j=1}\sum_{k=1}^{N_1} N^2 V^{\lambda}_{12}(N (x_k - y_j))
	\end{aligned}
\end{equation}
for initial data $\psi_{N}^\init$ approximating the ground state of Hamiltonian~\eqref{eq:2species_trapped_Hamiltonian}.

\noindent\textbf{(a) GP limit}. For fixed coupling constant $\lambda = 1$, suppose the time-evolution $\psi_{N, t}$ exhibits BEC in the factorized form $u_t^{\otimes N_1}\otimes v_t^{\otimes N_2} \in L^2_s (\RR^{3N_1})\otimes L^2_s(\RR^{3N_2})$ at time $t$. Then the rescaled condensate profiles $(\phi_{1,t},\, \phi_{2,t})$ defined through the density-scaling
\begin{equation}
\phi_{1,t} = \sqrt{\sfn_1}\, u_t 
\quad\text{and}\quad 
\phi_{2,t} = \sqrt{\sfn_2}\, v_t
\end{equation}
satisfy the coupled GP system derived from the Hamiltonian \eqref{eq:GP_energy_functional} in the absence of trapping potential:
\begin{equation}\label{eq:GP_system}
	\left\{
	\begin{aligned}
		i\bd_t\phi_{1,t} =&\, -\lapl \phi_{1,t} + 8\pi \asc_{1}\n{\phi_{1,t}}^2\phi_{1,t} + 8\pi \asc_{12}\n{\phi_{2,t}}^2\phi_{1,t} \, ,\\
		i\bd_t\phi_{2,t} =&\, -\lapl \phi_{2,t} + 8\pi \asc_{2}\n{\phi_{2,t}}^2\phi_{2,t} + 8\pi \asc_{12}\n{\phi_{1,t}}^2\phi_{2,t} \, ,\\
		&\Nrm{\phi_{1}}{L^2}^2 = \sfn_1, \qquad \Nrm{\phi_2}{L^2}^2 = \sfn_2 \, .
	\end{aligned}
	\right.
\end{equation}

\noindent\textbf{(b) Hard-core limit}.
Moreover, we are also interested in the case $\lambda\rightarrow \infty$, which we called the \emph{hard-core limit}, since in this case we have 
\begin{align}\label{hard core limit}
    \lim_{\lambda \rightarrow \infty} \lambda V_\ii(x) =
    \begin{cases}
        \infty & \text{ if } \n{x}< \bsc_\ii;\\
        0 & \text{ otherwise},
    \end{cases}
\end{align}
pointwise almost everywhere. In the hard-core limit, we expect to derive the two-component GP equation \begin{equation}\label{eq:GP_system hard core}
	\left\{
	\begin{aligned}
		i\bd_t\phi_{1,t} =&\, -\lapl \phi_{1,t} + 8\pi \bsc_{1}\n{\phi_{1,t}}^2\phi_{1,t} + 8\pi \bsc_{12}\n{\phi_{2,t}}^2\phi_{1,t} \, ,\\
		i\bd_t\phi_{2,t} =&\, -\lapl \phi_{2,t} + 8\pi \bsc_{2}\n{\phi_{2,t}}^2\phi_{2,t} + 8\pi \bsc_{12}\n{\phi_{1,t}}^2\phi_{2,t} \, ,\\
		&\Nrm{\phi_{1}}{L^2}^2 = \sfn_1, \qquad \Nrm{\phi_2}{L^2}^2 = \sfn_2 \, .
	\end{aligned}
	\right.
\end{equation}

In the one-component setting, rigorous quantitative methods for deriving the GP equation have been developed, as shown in \cite{Benedikter2015quantitative, boccato2017quantum, caraci2024quantum}. 
For the two-component quantum mixtures, \cite{olgiati2017effective} shows that the time-evolution \eqref{eq:linear_schrodinger} of an initial data $\psi_{N}^{\init}$ exhibiting BEC  still exhibits BEC at time $t$, described by \eqref{eq:GP_system}.
However, several fundamental questions remain open.
First, while \cite{olgiati2017effective} established qualitative convergence of the many-body Schrödinger dynamics to the GP system, explicit $N$-dependent convergence rate of the error remains undetermined. Second, results in the GP regime guarantee approximation accuracy only with double-exponential-in-time bound (see \cite{Benedikter2015quantitative, caraci2024quantum}), with uniform‑in‑time error bounds remain unexplored. 

This work bridges these gaps by deriving the GP systems \eqref{eq:GP_system} and \eqref{eq:GP_system hard core} from two-component quantum mixtures in the GP scaling regime, addressing both the $\lambda = 1$ and $\lambda\rightarrow +\infty$ (hard-core) limits. 
Our contributions are as follows.
\begin{itemize}
    \item 
    \textbf{Quantitative Convergence Rate}.
We extend the quantitative Bogoliubov transformation framework of \cite{Benedikter2015quantitative} to the two-component setting, deriving an explicit $O(N^{-1/2})$ error bound between the one‑particle reduced density operator and GP dynamics. 
Unlike previous qualitative results for two‑component condensates \cite{olgiati2017effective}, our approach yields a particle‑number–dependent convergence rate that is robust even in the hard‑core limit of strong interactions \eqref{hard core limit} .
    \item \textbf{Uniform‑in‑Time Control}. To obtain a uniform-in-time error estimate,
we introduce a modified GP system with a solution denoted by $\wtbphi^{(N)}_t$, which arises from a localizing the zero-energy scattering problem to the corresponding Neumann problem. This localization, along with the more refined estimates of Proposition~\ref{prop:estimate of fluctuation}, allows us to bound the fluctuation generator in terms of
\begin{equation}\label{integral of modified solution}
    \int_0^T \|\wtbphi^{(N)}_t\|_{L^{\infty}_x} \d t
\end{equation}
and—by using Morawetz estimates with Strichartz estimates—prove that this time integral remains bounded for all $T<\infty$. As a result, the convergence error remains uniformly controlled for arbitrarily large times.
\end{itemize}

\subsection{Modified Gross--Pitaevskii System}
The goal is to approximate the many-body wave function $\psi_{N,t}$ in the sense of reduced density matrices for a class of initial data exhibiting a complete BEC mixture.
Following the approaches in \cite{boccato2017quantum, Brennecke2019, caraci2024quantum}, it is convenient to consider a slightly modified, $N$-dependent, self-consistent equations instead of \eqref{eq:GP_system}. 
In the modified system, the interaction potential that appears in \eqref{eq:2species_Hamiltonian} is corrected, to take into account the correlations among the particles.
Fix $\ell>0$, we localize the solution of \eqref{eq:zero-energy_scattering} to the ball $\n{x} \leq N\ell$, namely, we consider solution $f_{\ii, \ell}$ to the Neumann problem
\begin{equation}\label{eq:scatl} 
\begin{cases}
	\(-\lapl + \tfrac12 V^{\lambda}_{\ii}\) f_{\ii, \ell} = \nu_{\ii,\ell} f_{\ii, \ell}
    & \text{for } \n{x} \leq N\ell ,\\[1ex]
    (x\cdot \nabla)f_{\ii, \ell} (x) = 0,
    \quad
    f_{\ii, \ell} (x) = 1 
    & \text{on } \n{x} = N\ell,
\end{cases}
\end{equation}
where $f_{\ii, \ell}(x)$ is extended by 1 outside the ball.
We omit here the $\lambda$ dependence in the notation for $f_{\ii, \ell}$ and for $\nu_{\ii,\ell}$; notice that $\nu_{\ii,\ell}$ scales as $N^{-3}$. 
We introduce the notation $w_{\ii, \ell} (x) := 1 - f_{\ii, \ell} (x)$. To describe the correlations created by the rescaled interaction appearing in Hamiltonian~\eqref{eq:2species_trapped_Hamiltonian} and in \eqref{eq:2species_Hamiltonian}, we use the functions $f^{(N)}_{\ii, \ell}(x) = f_{\ii, \ell} (Nx)$, $w^{(N)}_{\ii, \ell} (x) = w_{\ii, \ell} (Nx) = 1- f^{(N)}_{\ii, \ell}(x)$. By scaling, we observe that 
\begin{align}\label{eq:scatlN} 
	\(-\lapl + \tfrac12 N^2V^{\lambda}_{\ii} (N \cdot ) \) f^{(N)}_{\ii, \ell} = N^2 \nu_{\ii,\ell} f^{(N)}_{\ii, \ell} =: \nu_{\ii,\ell}^{(N)} f^{(N)}_{\ii, \ell}
\end{align}
on the ball $\n{x} \leq \ell$. With this choice, we expect that $f_{\ii, \ell}$ will be close, in the large $N$ limit, to the solution of the zero-energy scattering equation \eqref{eq:zero-energy_scattering}. 

With $f_{\ii,\ell}$ satisfying \eqref{eq:scatl}, we now define the two-component condensate wave function $\wt{\vect{\phi}}^{(N)}_t=(\wt \phi_{1, t}^{(N)}, \wt \phi_{2, t}^{(N)})^\top$ satisfying the modified Gross--Pitaevskii (GP) system  
\begin{equation}\label{eq:modified_GP_system}
	\left\{
	\begin{aligned}
		i\bd_t\wt \phi_1^{(N)} =&\, -\lapl \wt \phi_1^{(N)} +\Big(N^{3}V^{\lambda}_{1}(N\cdot)f_{1, \ell}(N\cdot)\ast |\wt \phi_1^{(N)}|^2\Big)\wt \phi_1^{(N)} \\
		&+ \Big(N^{3}V^{\lambda}_{12}(N\cdot)f_{12, \ell}(N\cdot)\ast |\wt \phi_2^{(N)}|^2\Big)\wt \phi_1^{(N)} \, ,\\
		i\bd_t\wt \phi_2^{(N)} =&\, -\lapl \wt \phi_2^{(N)} + \Big(N^{3}V^{\lambda}_{2}(N\cdot)f_{2, \ell}(N\cdot)\ast |\wt \phi_2^{(N)}|^2\Big)\wt \phi_2^{(N)} \\
		&+ \Big(N^{3}V^{\lambda}_{12}(N\cdot)f_{12, \ell}(N\cdot)\ast |\wt \phi_1^{(N)}|^2\Big)\wt \phi_2^{(N)} \, ,\\
		&\norm{\wt \phi_{1, 0}^{(N)}}_{L^2_x}^2   = \frac{N_1}{N}, \qquad 
		\norm{\wt \phi_{2, 0}^{(N)}}_{L^2_x}^2  = \frac{N_2}{N} \, .
	\end{aligned}
	\right.
\end{equation}

\subsection{Some Fock space notations}
To derive an explicit rate of convergence, we employ the Bogoliubov state approximation in the grand canonical picture, which reduces the problem of comparing the one‑particle reduced density operator of $\Psi_{N,t}$ with the orthogonal projection operator onto $\widetilde{\boldsymbol\phi}^{(N)}_t$ to the analysis of a fluctuation dynamics.
To explain this approach, we start by reviewing some rudimentary Fock space (grand canonical) formalism. 

The two-component bosonic Fock space is defined by
\begin{equation*}
	\cF:=
    \bigoplus_{n,m\in\mathbb{N}_0 } 
    L_s^2(\mathbb{R}^{3n}) \otimes L_s^2(\mathbb{R}^{3m}).
\end{equation*}
Normalized two-component Fock state $\Psi = \{\psi_{n,m}\}_{n,m\in\N_0} \in \cF$ describes a system that has a probability $\nrm{\psi_{n,m}}^2$ to contain $n$ particles of first species and $m$ particles of second species in the state $\psi_{n, m}/\nrm{\psi_{n, m}}$.
For any $f\oplus g \in L^2 (\RR^{3})\oplus L^2(\RR^{3})$, we define the corresponding creation and annihilation operators by $z^\ast(f\oplus g)$ and $z(f\oplus g)$. We also write $ a^*(f):=z^\ast(f\oplus 0)$ and $a(f):=z(f\oplus 0)$, which are the creation and annihilation operators for the first species. They act on $\cF$ by creating and annihilating a particle of the first species with a wave function $f$, respectively.  Similarly, we write $b^*(g):=z^\ast(0\oplus g)$ and $b(g):=z(0\oplus g)$ for the second species.
It is also convenient to introduce the spatially localized operator-valued distributions $a^\ast_x, a_x, b^\ast_x$ and $b_x$
(definitions and basic properties are given in Section \ref{sec: Fock}).

In terms of these distributions, we define the Hamiltonian
\begin{equation}\label{def:two-component_Fock_Hamiltonian}
	\begin{aligned}
		\cH_{N} =&\, \cH_{1} + \cH_{2} + \cV_{12} \, ,
	\end{aligned}
\end{equation}
where the intra-species Hamiltonian $\cH_{\ii}$ and the inter-species interaction potential $\cV_{12}$ are given by
\begin{equation}\label{intra and inter Hamiltonian}
	\left\{\begin{aligned}
		\cH_{1}=&\, \intd 
        a^\ast_x(-\lapl_x) a_x\dd x 
        +\frac12\iintd N^2 V^{\lambda}_{1}(N(x-y))\, a^{\ast}_x a^{\ast}_y a_y a_x\dd x\d y , \\
		\cH_{2}=&\, \intd 
         b^\ast_x(-\lapl_x)b_x\dd x 
        +\frac12\iintd N^2 V^{\lambda}_{2}(N(x-y))\, b^{\ast}_x b^{\ast}_y b_y b_x\dd x\d y, \\
		\cV_{12}=&\, \iintd N^2 V^{\lambda}_{12}(N(x-y))\,a^{\ast}_x a_x b^{\ast}_y b_y\dd x\d y \, .
	\end{aligned}\right.
\end{equation}
Another useful operator to consider is the total number operator $\cN=\cN_1 + \cN_2$, whose action on $\psi = \{\psi_{n,m}\}_{n,m\in\mathbb{N}_0}\in \cF$ yields $(\cN_1 \psi)_{n,m} = n \psi_{n,m}$, and $(\cN_2 \psi)_{n,m} = m \psi_{n,m}$. 

Following \cite{Benedikter2015quantitative}, we also define the two-component coherent state. We start by defining the two-component Weyl operator  
\begin{equation}
    \cW(\mathbf{f}) = \exp\(a^*(f_1)+ b^*(f_2) - a(f_1)-b(f_2)\).
\end{equation}
Denote the Fock vacuum in $\cF$ by $\Omega= (1, 0, 0, \ldots)\otimes (1, 0, 0, \ldots)$, which describes a state without particles. A two-component coherent state of $\mathbf{f}$ is defined by
\begin{align}
	\cW(\mathbf{f}) \Omega
	= e^{-\frac{\|\mathbf{f}\|_2^2}{2}}
    \sum_{n,m\ge0}
    \frac{a^\ast(f_1)^n\;
    b^\ast(f_2)^m
    \Omega}{n!\;m!} \, .
\end{align}


\subsection{Bogoliubov state approximation}
In the one‑component GP regime, it was already known from \cite{erdos2010derivation} that the many‑body evolution of a coherent state develops significant two-body correlation structures and therefore cannot remain approximately a coherent state. 
To capture these correlations, \cite{Benedikter2015quantitative} introduce the following Bogoliubov operator on Fock space
\begin{equation}\label{operator T for one-species}
T(k)=\exp \left(\frac{1}{2} \iintd \lt\{k(x, y)\, a_x^* a_y^*-\conj{k(x, y)}\, a_x a_y\right\}\d x\d y\right),
\end{equation}
where the kernel $k\in L^2(\mathbb{R}^3\times \mathbb{R}^3)$ is chosen to encode the short‐range correlation structure.

In the following, we extend this construction to the two-component setting.
Introducing the following matrix-valued kernel called the pair excitation matrix: 
\begin{equation}
	\boldsymbol{k}(x, y)=
	\begin{pmatrix}
		k_{1}(x, y) & k_{12}(x, y) \\
		k_{21}(x, y) & k_{2}(x, y) 
	\end{pmatrix}    ,
\end{equation}
and the operator-valued distribution vectors 
\begin{equation}\label{distribution vector}
    z_x = \begin{pmatrix}a_x\\b_x\end{pmatrix},
\qquad
z_x^* = \begin{pmatrix}a_x^*\\b_x^*\end{pmatrix}.
\end{equation}
We formally define the two‑component Bogoliubov operator by
\begin{equation}\label{eq: generalized T}
	\cT(\boldsymbol{k}) 
    = 
    \exp \left(
    \frac 12 \iintd 
    \lt\{ (z^\ast_x)^\top \boldsymbol{k}(x, y) z_y^*
    -
    z_x^\top\, \conj{\boldsymbol{k}(x, y) } z_y\rt\}\d x \d y\right),
\end{equation}
where $z^\top$ denotes the vector transpose. 
This $
\cT(\boldsymbol{k})$ reduces to the one‑component operator when only one species is present and extends naturally to the general multi‑component case. 
As in the one‑component setting, conjugation by $
\cT(\boldsymbol{k})$ implements a two‑component Bogoliubov transformation of
$(z_x, z^\ast_x)$; see Section \ref{subsect:bogoliubov_transf} for definitions and basic properties.

Similar to the one-component case, where correlations were successfully described by the solution of \eqref{eq:scatl}, we define the (time-dependent) kernel $\boldsymbol{k} = \boldsymbol{k}_t$ by
\begin{subequations}
	\begin{equation}\label{def:k_i}
		\left\{\begin{aligned}
			&k_{1,t}(x, y)=-N w_{1, \ell}\(N(x-y)\) \wt\phi_{1,t}^{(N)}(x) \wt\phi_{1,t}^{(N)}(y)\, ,\\
			&k_{2,t}(x, y)=-N w_{2, \ell}\( N(x-y)\) \wt\phi_{2,t}^{(N)}(x) \wt\phi_{2,t}^{(N)}(y) \, ,
		\end{aligned}\right.
	\end{equation}
    for the intraspecies components, and
	\begin{equation}\label{def:k_ij}
		\left\{\begin{aligned}
			&k_{12,t}(x, y)=-N w_{12, \ell}\(N(x-y)\) \wt\phi_{1,t}^{(N)}(x) \wt\phi_{2,t}^{(N)}(y)\, ,\\
			&k_{21,t}(x, y)=-N w_{12, \ell}\(N(x-y)\) \wt\phi_{2,t}^{(N)}(x) \wt\phi_{1,t}^{(N)}(y) \, ,
		\end{aligned}\right.
	\end{equation}
\end{subequations}
for the cross terms,
where $(\wt\phi_{1,t}^{(N)},\wt\phi_{2,t}^{(N)})$ is the solution to the modified GP system \eqref{eq:modified_GP_system}, and $w_{\ii, \ell} (x) = 1 - f_{\ii, \ell} (x)$
with $f_{\ii, \ell}$ being the solution of \eqref{eq:scatl} for $\ii \in \{1, 2, 12\}$.
Note that one has the symmetry $k_{12,t}(x,y) = k_{21,t}(y, x)$. 
Using these kernels, the two‑component Bogoliubov operator \eqref{eq: generalized T} can be written as $\cT_t:= \cT (\boldsymbol{k}_t)= e^{-\cB(\boldsymbol{k}_t)}$, where
\begin{equation}\label{def: operator B}
	\begin{aligned}
    \cB(\boldsymbol{k}_t):=
        &\, \iintd \lt\{\conj{k_{12, t}(x, y)}\,a_xb_y-k_{12, t}(x, y)\, a^\ast_x b^\ast_y\rt\}\d x\d y\\
		&\, + \frac12 \iintd \lt\{\conj{k_{1, t}(x, y)}\,a_xa_y-k_{1, t}(x, y)\, a^\ast_x a^\ast_y\rt\}\d x\d y\\
		&\, +  \frac12 \iintd \lt\{\conj{k_{2, t}(x, y)}\,b_xb_y-k_{2, t}(x, y)\, b^\ast_x b^\ast_y\rt\}\d x\d y.
	\end{aligned}
\end{equation}
Similarly to the one-component case, $\cT_t$ provides a way to generate pair correlation structures in the GP dynamics. Hence, we shall study the time evolution of initial data close to the squeezed coherent state 
\begin{align}\label{squeezed coherent state}
    \Psi_0 = \cW(\sqrt{N}\wt{\vect{\phi}}^{(N)}_0) \cT_0 \Xi_{N} \, , 
\end{align}
where $\Xi_{N} \in \cF$ is close to the vacuum. 
Under the two‑component Fock Hamiltonian~\eqref{def:two-component_Fock_Hamiltonian}
  the exact solution is 
\begin{align}
	\Psi_{N, t}= e^{-it\cH_N}\cW(\sqrt{N}\wt{\vect{\phi}}^{(N)}_0) \cT_0 \Xi_{N} \,.
\end{align}
Equivalently, we can rewrite the solution by
\begin{equation}\label{eq:full-fock} 
	\Psi_{N, t}
    =
    \cW(\sqrt{N} \wt{\boldsymbol{\phi}}^{(N)}_t) \cT_t\, \cU_N (t; 0) \Xi_{N}
    =:
    \cW(\sqrt{N} \wt{\boldsymbol{\phi}}^{(N)}_t) \cT_t \Xi_{N, t} 
\end{equation}
where  $\cU_N(t; s)$ is the fluctuation dynamics given by 
\begin{equation}\label{eq:fluc-fock} 
	\cU_N (t; s) = \cT_t^* \cW^* (\sqrt{N} \wt{\vect{\phi}}^{(N)}_t) e^{-i (t-s)\cH_N } \cW(\sqrt{N} \wt{\vect{\phi}}^{(N)}_s) \cT_s\,.
\end{equation}
One can readily check it satisfies the following Schr\"odinger-type equation
\begin{equation}
	i \bd_t \cU_N(t ; s)=\cL_N(t) \cU_N(t ; s)
\end{equation}
with the time-dependent generator
\begin{multline*}
	\cL_N(t)
	= (i \bd_t \cT_t^*) 
	\cT_t +
	\cT_t^\ast
	\Big(i \bd_t \cW^\ast(\sqrt{N} \wt{\boldsymbol{\phi}}^{(N)}_t)\Big)
	\cW(\sqrt{N} \wt{\boldsymbol{\phi}}^{(N)}_t)
	\cT_t  \\
	+
	\cT_t^*
	\cW^\ast(\sqrt{N} \wt{\boldsymbol{\phi}}^{(N)}_t)
	\cH_N 
	\cW(\sqrt{N} \wt{\boldsymbol{\phi}}^{(N)}_t)
	\cT_t\,.
\end{multline*}

\subsection{Main Result} 
In order to state our main result, we need to define the two-component one-particle reduced density operator. 
Let $\Psi \in \cF$ be a normalized vector with finite particle expectation, then we define the two-component one-particle reduced density operator to be the positive semi-definite operator $\Gamma_{\Psi}^{(1)}$ defined on  $\fH$ by 
\begin{align*}
	\inprod{u\oplus v}{\Gamma_{\Psi}^{(1)} f\oplus g}_{\fH} =&\, \frac{\inprod{\Psi}{z^\ast(f\oplus g)z(u\oplus v)\Psi}}{\inprod{\Psi}{\cN\, \Psi}}\,, 
\end{align*}
for all $f\oplus g, u\oplus v \in \fH$, that is, the integral kernel of $\Gamma^{(1)}_{\Psi}$ is given by
\begin{align*}
	\Gamma_{\Psi}^{(1)}(x, y)= 
	\frac{\langle{\Psi},{z^\ast_y z_x^\top\, \Psi}\rangle}{\inprod{\Psi}{\cN\, \Psi}}
	=
	\frac{1}{\inprod{\Psi}{\cN\, \Psi}}
	\begin{pmatrix}
		\langle{\Psi},{a^\ast_y a_x\,\Psi}\rangle & \langle{\Psi},{a^\ast_y b_x\,\Psi}\rangle\\[1ex]
		\langle{\Psi},{b^\ast_y a_x\, \Psi}\rangle & \langle{\Psi},{b^\ast_y b_x\,\Psi}\rangle
	\end{pmatrix}.
\end{align*}
Furthermore, we denote by $|\bphi_t \rangle\!\langle\bphi_t |$ the orthogonal projection onto the vector $\bphi_t$, whose integral kernel is given by
\begin{equation}\label{orthogonal projection}
    |\bphi_t \rangle\!\langle\bphi_t | (x,y)
    =
    \begin{pmatrix}
		\phi_{1, t}(x)\conj{\phi_{1, t}(y)} & \phi_{2, t}(x)\conj{\phi_{1, t}(y)}\\[0.5ex]
		\phi_{1, t}(x)\conj{\phi_{2, t}(y)} & \phi_{2, t}(x)\conj{\phi_{2, t}(y)}
	\end{pmatrix}.
\end{equation}
To show that the one-particle reduced density operator $\Gamma_{N,t}^{(1)}$, associated with the LHS of \eqref{eq:full-fock},
remains close to the orthogonal projection \eqref{orthogonal projection} onto the solution of the GP system~\eqref{eq:GP_system}, it is enough to prove that the expectation of the number of particles in $\Xi_{N, t}$ is small, compared to the total number of particles $N$ (assuming that this is true for $\Xi_{N}$, at time $t=0$).

\begin{thm}\label{thm:main}
	Let $\wtbphi^{(N)}_0, \bphi_0 \in H^4\cap W^{4, 1}(\RR^3, \CC^2)$ uniformly in $N$, with $\|\wt\phi^{(N)}_{1, 0} \|_{L^2_x} =N_1/N, \, \|\wt\phi^{(N)}_{2, 0} \|_{L^2_x} =N_2/N$ and $\|\phi_{1, 0} \|_{L^2_x} =\sfn_1, \, \|\phi_{2, 0} \|_{L^2_x} =\sfn_2$. Moreover, assume 
	\begin{align}
		\|\wtbphi^{(N)}_0- \bphi_0\|_{H^1_x}\le \frac{1}{N}\,.
	\end{align}
	Let $\cH_N$ be the Hamiltonian operator defined in \eqref{def:two-component_Fock_Hamiltonian} where $V_\ii$ satisfies the Assumption~\ref{assume: L3}. Suppose $\Xi_N \in \cF$ (possibly depending on $N$) be such that
	\begin{equation}\label{eq:ass-thm1} 
		\langle \Xi_N , \cN \Xi_N \rangle,\,  \frac{1}{N} \langle \Xi_N , \cN^2 \Xi_N \rangle,\, \langle \Xi_N, \cH_N \Xi_N \rangle \leq D \end{equation}
	for a constant $D>0$. 
	Let $\Gamma_{N,t}^{(1)}$ denote the two-component one-particle reduced density operator associated with the evolved Bogoliubov state as defined above expression in \eqref{eq:full-fock}, i.e., $e^{-it\cH_N} \cW(\sqrt{N} \wtbphi^{(N)}_0) \cT_0 \Xi_{N}$. 
	Then, there exist constants $C>0$, depending only on $V_\ii, \, \|\wtbphi^{(N)}_0\|_{H^4_x}, \, \|\bphi_0\|_{H^4_x}, \, \|\wtbphi^{(N)}_0\|_{W^{4, 1}_x}, \, \|\bphi_0\|_{W^{4, 1}_x}$,  and on the constant $D$ appearing in (\ref{eq:ass-thm1}), such that the following holds for all $t \geq 0$ and $N \in \N$.
    \begin{enumerate}
    \item When $\lambda = 1$, then we have
	\begin{equation}\label{eq:mt} 
		\tr_{\fH}\left| \Gamma_{N,t}^{(1)} - |\bphi_t \rangle\!\langle\bphi_t | \right| \leq \frac{C}{\sqrt{N}} .
	\end{equation}
	 Here, $\bphi_t=(\phi_{1, t}, \phi_{2, t})^\top$ denotes the solution of the time-dependent two-component GP system \eqref{eq:GP_system}
	with the initial condition $\left.\bphi_{t}\right|_{t=0}=\bphi_0$.
    \item When $\lambda = \ln N^\gamma$ for some $0<\gamma<\frac{1}{2}$, we further impose Assumption~\ref{assume: boundedness}, then we have
    \begin{equation}
		\tr_{\fH}\left| \Gamma_{N,t}^{(1)} - |\bphi_t \rangle\!\langle\bphi_t | \right| \leq C\(N^{\gamma-\frac{1}{2}} +\upepsilon(\lambda)\).
	\end{equation}
    Here, $\bphi_t=(\phi_{1, t}, \phi_{2, t})^\top$ is the solution of the GP system \eqref{eq:GP_system hard core}
	with the initial condition $\left.\bphi_{t}\right|_{t=0}=\bphi_0$, and $\upepsilon(\lambda)$ is as defined in Lemma~\ref{lem:neumann_scattering_function}.
    \end{enumerate}
\end{thm}

\begin{remarks}\phantom{ }
	\begin{enumerate}[1.]
        \item The above uniform-in-time bound also implies a uniform-in-time bound result for the case $p=1$. More precisely, Proposition~\ref{prop:estimate of fluctuation} and the dispersive estimate in Section \ref{sec:limiting-eq-estimates} also hold in the $p=1$ case improving upon the result of \cite[Theorem 3.5]{Benedikter2015quantitative}. 
		\item Theorem \ref{thm:main} can be extended to any finite $p$-component system. For simplicity of notation, we present the proof only for the case $p = 2$. However, this work is written with a view toward the general $p$-component case. By generalizing the Fock space to $\cF = \cF_1 \otimes \cF_2 \otimes \dots \otimes \cF_p$ and appropriately defining the particle creation and annihilation operators, as well as the $p$-component particle number operators, the proof can be readily extended. This generalization has already been addressed for the mean-field regime in \cite{Lee2021mixture}. Furthermore, a key contribution of our work is the systematic extension of the Bogoliubov transformation to the $p$-component setting, which enables us to handle the GP regime (see Subsection~\ref{subsect:bogoliubov_transf}).
		\item For a normalized $\Psi \in \cF$ with finite expected-valued with respect to $\cN^k$,  we define the two-component $k$-particle reduced density operator $\Gamma^{(k)}_{\Psi}$ acting on $\fH^{\otimes_s k}$ by  
		\begin{align*}
			\inprod{\vect{f}_1\otimes_s \cdots \otimes_s \vect{f}_k}{\Gamma^{(k)}_{\Psi} \vect{g}_1\otimes_s\cdots \otimes_s \vect{g}_k} =  \frac{\inprod{\Psi}{z^\ast(\vect{g}_1)\cdots z^\ast (\vect{g}_k)z(\vect{f}_1)\cdots z(\vect{f}_k)\Psi}}{\inprod{\Psi}{\cN(\cN-1)\cdots (\cN-k+1)\, \Psi}},
		\end{align*}
		for every $\vect{f}_i, \vect{g}_i \in \fH$, or equivalently, $\Gamma^{(k)}_{\Psi}$ has the tensor kernel 
		\begin{align*}
			\Gamma^{(k)}_{\Psi}(x_1, \ldots, x_k; x_1', \ldots, x_k') = \frac{\langle{\Psi},{z_{x_1'}^\ast\otimes \cdots \otimes z^\ast_{x_k'}\otimes z_{x_1}\otimes\cdots \otimes z_{x_k} \Psi}\rangle}{\inprod{\Psi}{\cN(\cN-1)\cdots (\cN-k+1)\,\Psi}}.
		\end{align*}
		Moreover, we define the $(k, \ell)$-particle reduced density operator by 
		\begin{align*}
			\upgamma^{(k, \ell)}_{\Psi}(X_k, Y_\ell; X_k', Y_\ell') = \frac{\langle{\Psi},{a^\ast_{x_1}\cdots a^\ast_{x_k} a_{x_k}\cdots a_{x_1} b^\ast_{y_1}\cdots b^\ast_{y_\ell} b_{y_\ell}\cdots b_{y_1}\Psi}\rangle}{\inprod{\Psi}{\cN(\cN-1)\cdots (\cN-k-\ell+1)\,\Psi}}
		\end{align*}
		where $X_{k} = (x_1, \ldots, x_k)$ and $Y_\ell =(y_1, \ldots, y_\ell)$. Notice that, for given $k,\ell\in \mathbb{N}_0$, the operator $\upgamma^{(k, \ell)}_{\Psi}$ is one of the the operator entries of $\Gamma^{(k+\ell)}_{N}$ that commute with the number operator. 
		Then, following the argument in \cite[Section 2]{knowles2010mean}, we could also obtain convergence of the higher reduced density operators as follows: For any $k\in \mathbb{N}$, we have the inequality
		\[
		\tr\left|\Gamma^{(k)}_{N, t}-|\bphi_t^{\otimes k}\rangle\!\langle \bphi_t^{\otimes k}|\right|
		\leq 
		C_k \sqrt{\tr\left|\Gamma^{(1)}_{N, t}-|\bphi_t\rangle\!\langle \bphi_t|\right|}
		\leq \frac{C_k'}{N^{\frac14}}.
		\]
		A convergence rate of $CN^{-\frac12}$ can be obtain if we directly compare $\Gamma^{(k)}_{N, t}$ with $|\bphi_t^{\otimes k}\rangle\!\langle \bphi_t^{\otimes k}|$, but, of course, this entails controlling the expectation value $\cN^k$ with respect to the fluctation dynamics. 
		\item 
		
		For the GP limit, convergence without a specific rate was provided in \cite{olgiati2017effective}. In fact, the result was stated with the existence of a constant $\gamma$ such that the left-hand side of \eqref{eq:mt} is bounded by $C_t N^{-\gamma}$. In this paper, we provide an explicit convergence rate.  More precisely, Theorem \ref{thm:main} and the above inequality implies 
		\[
		\tr\left|\upgamma^{(1,1)}_{N, t}-|\phi_{1, t}\otimes \phi_{2, t}\rangle\!\langle \phi_{1, t}\otimes \phi_{2, t}|\right|
		\leq 
		\tr\left|\Gamma^{(2)}_{N, t}-|\bphi_t^{\otimes 2}\rangle\!\langle \bphi_t^{\otimes2}|\right|
		\leq \frac{C}{N^{\frac14}}\,,
		\]
		which can be compared with \cite[Theorem 2]{olgiati2017effective}.
		\item Despite our result is written in the grand canonical language, Theorem~\ref{thm:main} also implies a convergence result for the canonical picture for a certain class of $N$-particle initial data. If we consider an $N$-particle state given by $\psi_{N}=P_{N_1, N_2}\cW(\sqrt{N} \wtbphi^{(N)}_0) \cT_0 \Xi_{N}$, where $P_{N_1, N_2}:\cF\rightarrow \fH_{N_1, N_2}$ is the projection operator onto the $(N_1, N_2)$ sector of $\cF$, with $\|P_{N_1, N_2}\cW(\sqrt{N} \wtbphi^{(N)}_0) \cT_0 \Xi_{N}\|_{L^2(\R^{3N})}\ge CN^{-\frac14}$ and $N$ sufficiently large, then we see that  
		\[
		\tr\left|\gamma^{(1, 0)}_{N, t}-|u_{t}\rangle\!\langle u_{t}|\right|
		\leq \frac{C}{N^{\frac14}} \quad \text{ and } \quad 
		\tr\left|\gamma^{(0, 1)}_{N, t}-|v_{t}\rangle\!\langle v_{t}|\right|
		\leq \frac{C}{N^{\frac14}},
		\]
		where $u_t = \phi_{1, t}/\sqrt{\sfn_1}$ and $v_t = \phi_{2, t}/\sqrt{\sfn_2}$.
		The proof is similar to the one given in \cite[Appendix C]{Benedikter2015quantitative}. The rate of convergence could be improved if we adapt the method used in \cite{lewin2015bogoliubov, Brennecke2019}, but, for simplicity, we decided not to pursuit this route. 
	\end{enumerate}
\end{remarks}

\subsection{Literature Overview}

\subsubsection*{One-Component BEC: Equilibrium Properties and Dynamics} The pioneering work of Lieb and Yngvason in \cite{lieb1998ground} on the ground state energy inspired numerous subsequent studies, including the rigorous proof that there is $100\%$ BEC in the ground state in the GP limit. This work also demonstrated that the GP theory correctly describes the ground-state properties of an interacting Bose gas \cite{lieb2002proof,lieb2000bosons,nam2016ground}. These results have since been refined and extended via Bogoliubov theory in \cite{boccato2019bogoliubov, boccato2020optimal, brennecke2022bogoliubov, brennecke2022bose, nam2022optimal, nam2023bogoliubov}. For recent developments concerning the equilibrium properties of Bose gases in the (translation-invariant) GP regime, beyond the GP regime, and in the thermodynamic limit, see \cite{adhikari2021bose, basti2024upper, basti2023second, basti2021new, boccato2023bose, brennecke2022excitation, caraci2021bose, caraci2023excitation, Fournais2020LengthSF, fournais2020energy, fournais2023energy, haberberger2023free, hainzl2022bogoliubov}.

The validity of GP theory in the time-dependent setting was first established in a series of foundational papers by Erdős, Schlein, and Yau \cite{erdos2006derivation, erdos2007derivation, erdos2009rigorous, erdos2010derivation}. These works rigorously derived the 3D time-dependent GP equation from the dynamics of a repulsive quantum many-body system in the GP limit, motivating a large body of research on quantum BBGKY and GP hierarchies \cite{CHPS15, CP11, CH16correlation, CH16on, CH19, HS16, KSS11, klainerman2008uniqueness, Soh15}. In parallel, using a slightly different and more quantitative approach based on projection operators in the $N$-particle Hilbert space, Pickl \cite{Pickl2010, Pickl2011, Pickl2015} also derived the GP equation and related nonlinear Schrödinger-type equations. For lower dimensions, in 1D, the nonlinear Schrödinger equation (NLSE)  was derived using the quantum BBGKY hierarchy in \cite{adami2004towards, adami2007rigorous}. In 2D, within the GP regime, the NLSE was derived in \cite{JLP2019} by generalizing the approach in \cite{Pickl2010}.

By studying quantum fluctuations about the effective dynamics, robust quantitative approaches to derive the GP equation in the canonical and grand canonical formalisms have been developed, as shown in \cite{Benedikter2015quantitative, boccato2017quantum, caraci2024quantum}. These methods build on earlier works \cite{ginibre1979classical, grillakis2010second, grillakis2011second, hepp1974classical, lewin2015bogoliubov, rodnianski2009quantum}. Notably, Rodnianski and Schlein \cite{rodnianski2009quantum} introduced a coherent state approach to derive effective dynamics for BECs in the mean-field regime, showing that the trace norm difference between the many-body state $\Psi_N$ and the limiting Hartree state decays with $N$. Specifically, they established $\tr \big| \gamma^{(1)}_{N, t} - |\phi_t\rangle\!\langle \phi_t| \big| \leq Ce^{Kt} N^{-1/2},$
where $\phi_t$ is the solution of the Hartree equation, and $C, K>0$ are constants independent of $N$. Subsequently, Grillakis, Machedon, and Margetis \cite{grillakis2010second, grillakis2011second} introduced the Bogoliubov transformation to account for the second-order correction to the mean-field evolution of weakly interacting bosons, offering a more precise description of behavior of the system beyond the mean-field limit.

These approaches provide a powerful tool for obtaining quantitative norm estimates of the error between the many-body and effective mean-field dynamics in the subcritical regime (e.g., \cite{brennecke2019fluctuations, Brennecke2019, XChen2012ARMA, chong2022dynamical, grillakis2013pair, grillakis2017pair, lewin2015fluctuations, nam2017bogoliubov, nam2017note}).
In certain regimes, approximations with optimal or arbitrarily precise decay in $N$ error bounds have been constructed (see \cite{bossmann2020higher, bossmann2022beyond, chen2011rate, erdos2009quantum}). Moreover, the usage of quantum fluctuations to study the hard-core limit was also done in \cite{chong2025derivation} to derive the compressible Euler equations from dilute Bose gas systems.
See also \cite{Napiorkowski2023} for an overview of the current state of research in the analysis of Bose gas. 

\subsubsection{Time Dependence of the Convergence Rate}
Since the development of quantitative approaches, the time dependence of the convergence rate has been an intensive studied topic. 
In the mean-field case, it was noted in \cite{Pickl2011} that a time-independent bound for the trace norm is possible if the solution to the Hartree equation satisfies a scattering condition. 
Using a different approach, Kuz obtained a $\sqrt{t}$ error bound in \cite{kuz2015rate}, and the time dependence of the convergence rate has been studied for various classes of mean-field interactions in \cite{Lee2019Time}. 
In the intermediate regimes between the mean-field and GP regimes, the time dependence of the convergence rate was partially analyzed in \cite{chong2021dynamics, DietzeLee2022Uniform}, leading to polynomial or time-independent bounds. 
In the GP regime, a double-exponential-in-time bound was obtained in \cite{Benedikter2015quantitative}, which could be improved to an exponential-in-time bound if the solution to the GP equation satisfies a scattering condition. 
In the norm approximation setting, significant work has been done to improve the time dependence of the convergence rate (see \cite{grillakis2011second, XChen2012ARMA, grillakis2013pair, kuz2015rate, kuz2017exact, chong2022dynamical}).

\subsubsection*{Multi-Component Case}  
For a multi-component Bose gas, we have already mentioned the results of \cite{michelangeli2019ground} regarding the GP regime above. Moreover, that work also demonstrated the validity of Bogoliubov’s approximation in the mean-field regime and provided the second-order expansion of the ground-state energy. Recently, an exact solution for the ground state of a general $p$-component harmonically trapped Bose gas interacting through harmonic forces was presented in \cite{AC2024}.  

For the dynamics, the trace norm convergence results have been extended to the case of multi-component Bose gases. Notably, the corresponding results for two-component BECs are well-established; see, for example, \cite{AnapolitanosHottHundertmark, MichelangeliOlgliatiBEC, MichelangeliOlgliatiPSEUDO, olgiati2017effective}.  
In \cite{olgiati2017effective}, Olgiati rigorously derived the two-component GP system from the dynamics of a two-component Bose gas using the counting method developed in \cite{Pickl2015}.  
In the mean-field regime, a similar result was obtained, with convergence in a trace Sobolev-type norm, in \cite{AnapolitanosHottHundertmark}.  
An explicit bound of order $O(N^{-1/2})$ on the convergence rate, with a Coulomb-type interaction, was derived by de Oliveira and Michelangeli \cite{MichelangeliOliveiraCoulomb} using the Fock space method.  
Furthermore, in \cite{Lee2021mixture}, an optimal convergence rate of $O(N^{-1})$ in the trace norm was established for a general $p$-component mixture.  

For fragmented BECs involving identical particles with spin, the convergence of the trace norm and its rate have been discussed in \cite{Dimonte_2021, LeeMichelangeli2023}. These works address systems in which the condensate occupies multiple quantum states, requiring a more nuanced analysis of the convergence behavior.

\section{Estimates for the Modified Gross--Pitaevskii System}\label{sec:limiting-eq-estimates}

\subsection{Preliminaries} Let $\wtbphi^{(N)}_t=(\wt \phi_{1, t}^{(N)}, \wt \phi_{2, t}^{(N)})^\top$ denotes the solution to System~\eqref{eq:modified_GP_system}, then we could recast the system in the following vector form
\begin{align}\label{eq:modified_GP_system_vector_form}
	i\bd_t\wtbphi^{(N)} = \vect{H}_0 \wtbphi^{(N)}+F_N(\wtbphi^{(N)})\wtbphi^{(N)}
\end{align}
where 
\begin{align*}
	\vect{H}_0 := 
	\begin{pmatrix}
		-\lapl & 0\\
		0 & -\lapl
	\end{pmatrix}
\end{align*}
and $F_N(\bphi):\fH\rightarrow\fH$ with kernel
\begin{align*}
	F_N(\bphi)(x, y)
	:=
	\begin{pmatrix}
		U_1(x-y)\phi_{1}(x)\conj{\phi_1(y)}& U_{12}(x-y)\phi_{1}(x)\conj{\phi_2(y)} \\
		U_{12}(x-y)\phi_{2}(x)\conj{\phi_1(y)} & U_{2}(x-y)\phi_{2}(x)\conj{\phi_2(y)}
	\end{pmatrix}
\end{align*}
with $U_\ii(\cdot) = N^3V^{\lambda}_\ii(N\cdot)f^\lambda_{\ii, \ell}(N\cdot) \in L^1(\R^3)$ uniformly in $N$ and $\lambda$ for $\ii \in \{1, 2, 12\}$. The purpose of this section is to obtain uniform-in-$N$ global estimates for solutions to system~\eqref{eq:modified_GP_system_vector_form} by comparing its solutions with solutions to 
\begin{align}\label{eq:GP_system_vector_form}
	i\bd_t\bphi= \vect{H}_0 \bphi+F_{\cc}(\bphi)\bphi
\end{align}
where $\cc = \{ \cc_{1}, \cc_{2}, \cc_{12}\}$ and $F_{\cc}(\bphi)$ is a multiplication operator defined by
\begin{align*}
	F_{\cc}(\bphi)(x)
	:=
	\begin{pmatrix}
		8\pi \cc_{1}\,\phi_{1}(x)\conj{\phi_1(x)} & 8\pi \cc_{12}\, \phi_{1}(x)\conj{\phi_2(x)} \\
		8\pi \cc_{12}\,\phi_{2}(x)\conj{\phi_1(x)} & 8\pi \cc_{2}\, \phi_{2}(x)\conj{\phi_2(x)} 
	\end{pmatrix}.
\end{align*}
The constant $\cc_{\ii}$ will later take the value of the scattering length $\asc_{\ii}$ or the radius $\bsc_{\ii}$.

We write $A\lesssim B$ to denote $A \le CB$ for some $C>0$, independent of $N$ and $\lambda$. Let $\bphi:I\times\R^3\rightarrow \CC^p$ with $\bphi = (\phi_1, \ldots, \phi_p)$, we write the Lebesgue and Sobolev norms of $\bphi$ by 
\begin{align*}
	\norm{\bphi}_{L^r_x}:=\Big(\sum^p_{\ii=1}\norm{\phi_{\ii}}_{L^r_x}^2\Big)^{\frac12} \quad \text{ and } \quad \norm{\bphi}_{W^{s, r}_x}:=\Big(\sum^p_{\ii=1}\norm{\phi_{\ii}}_{W^{s, r}_x}^2\Big)^{\frac12}.
\end{align*}
for $q\in [1, \infty]$ and $s\ge 0$, with the obvious modification when $q=\infty$. 

When $r=2$ and $s>0$, we use the standard notation $\norm{\bphi}_{H^{s}_x}:=\norm{\bphi}_{W^{s, 2}_x}$. We also define the corresponding norms 
\begin{align*}
	\norm{\bphi}_{L^q_tW^{s, r}_x(I\times \R^3)}:= \(\int_I\norm{\bphi}_{W^{s, r}_x}^q\dd t\)^{\frac1q} 
\end{align*}
for $q, r \in [1, \infty]$ and $s\ge 0$. Moreover, we said that a pair $(q, r)$ is 3D admissible provided $q\ge 2$ and $\frac{2}{q}+\frac{3}{r}=\frac{3}{2}$.

We start by summarizing some important properties of $f_{\ii, \ell}$ in the following lemma, which are taken from \cite[Lemma A.2]{chong2025derivation} (see also \cite[Lemma A.1]{erdos2006derivation}).

\begin{lem} \label{lem:neumann_scattering_function}
Let $V^{\lambda}=\lambda V$ where $V$ satisfies Assumption~\ref{assume: L3}. Let $\asc^\lambda$ be the corresponding scattering length and $\bsc$ be the radius of the support of $V^\lambda$. Fix $\ell > 0$ and let $f_{\ell}$ denote the solution \eqref{eq:scatl} with $V^\lambda$. For $N$ sufficiently large, the following properties hold true:
	\begin{enumerate}[(i)]
            \item Let $\upepsilon(\lambda) := \bsc -\asc^\lambda$ and employ Assumption~\ref{assume: boundedness} for $V$. Then it follows that $\upepsilon(\lambda)\rightarrow 0$ as $\lambda \rightarrow \infty$, where the convergence rate depends on the profile of $V$.
		\item  We have the  asymptotic upper bound 
		\begin{equation}\label{eq:lambdaell} 
			\nu_{\ell} \le \frac{3\asc^\lambda}{(\ell N)^3} \left(1 +\cO \big(\bsc  / \ell N\big) \right)\,.
		\end{equation}
		\item We have $0\leq f^\lambda_{\ell}, w^\lambda_{\ell}\leq1$. Moreover, there exists $C > 0$, independent of $N$, such that     
		\begin{equation} \label{eq:Vfa0} 
			\n{\int_{\n{x}\le N\ell}  V^{\lambda}(x) f^\lambda_{\ell} (x) \dd x - 8\pi \asc^\lambda} \leq \frac{C\bsc}{\ell N} \, 	
		\end{equation}
		for all $\ell$ fixed and $N \in \N$.
		\item There exists a constant $C>0 $ such that 
		\begin{equation}\label{est:scbounds1} 
			w_{\ell}^\lambda(x)\leq \frac{C\bsc}{|x|} \quad\text{ and }\quad |\nabla w_{\ell}^\lambda (x)|\leq \frac{C\bsc}{\n{x}^2}
		\end{equation}
		for all $x \in \R^3$ and all $N\in \N$ sufficiently large.   
	\end{enumerate}        
\end{lem}

\begin{lem} \label{lem:localtheory}
    Under Assumption~\ref{assume: L3}, system~\eqref{eq:modified_GP_system_vector_form} is uniform-in-$N$ globally well-posed in $H^1:=H^1(\R^3; \CC^2)$. 
\end{lem} 

\begin{proof}
The local $H^1$ well-posedness of \eqref{eq:modified_GP_system_vector_form} is similar to that of the NLSE; see \cite[Proposition 3.19]{tao2006nonlinear}. Since the system is repulsive, the conservation of energy implies that the $H^1$ norm remains bounded, which in turn yields global existence of solutions.
\end{proof}

\subsection{Interaction Morawetz Estimate} 
In this subsection, we prove the uniform-in-$N$ interaction Morawetz estimate for \eqref{eq:modified_GP_system_vector_form}. We start by proving a couple of lemmas regarding the perturbation of \eqref{eq:modified_GP_system_vector_form}.
\begin{lem}[\textbf{Short time Pertubation}] \label{lem:Short time Pertubation}
	Let $I\subset \R$ be a compact interval and $\bphi_t$ be a solution of
	\begin{equation}
		i\partial_t\bphi = \vect{H}_0\bphi + F_N(\bphi)\bphi + \vect{e}
	\end{equation}
	on $I\times\R^3$
	for some time-dependent vector-valued function $\vect{e}$. Assume that we have
	\begin{align}
		&\norm{\bphi_t}_{L^\infty_t H^1_x(I\times\R^3)} \le E, \\
		\label{eq:tiny u}	
		&\norm{\bphi_t}_{L^{4}_t W^{1, 3}_x(I\times\R^3)}	\le \eps_0, \\
		\label{tiny error}
		&\|\vect{e}\|_{L^2_t W^{1,6/5}_x (I \times \R^3)}  \le \eps,
	\end{align}
	for some real $E>0$ and constants $\eps_0, \eps >0$ small enough.
	For $t_0\in I$ let $\wtbphi(t_0)$  be close to $\bphi(t_0)$ in the following sense:
	\begin{align}
		\label{eq:closesness}
		\|e^{-i(t-t_0)\vect{H}_0}(\wtbphi(t_0)-\bphi(t_0))\|_{L^{4}_t W^{1, 3}_x(I\times \R^3)} \le \eps.
	\end{align}
	We then conclude that there exists a solution $\wtbphi_t$ of \eqref{eq:modified_GP_system_vector_form} on $I\times \R^3$ with initial state $\wtbphi(t_0)$ at $t_0$,
	which fulfills the following spacetime bounds:
	\begin{align}
		\label{eq:concludetiny1}
		\|\wtbphi_t-\bphi_t\|_{L^4_{t, x}(I\times \R^3)} \lesssim \|\wtbphi_t-\bphi_t\|_{L^4_tW^{1, 3}_x(I\times \R^3)} \lesssim \eps, \\ 
		\label{eq:concludetiny2}
		\|(i\partial_t -\vect{H}_0) (\wtbphi_t-\bphi_t) \|_{L^2_t W^{1,6/5}_x(I\times \R^3)}	\lesssim \eps.
	\end{align} 
\end{lem}

\begin{proof}
	By the local theory given in Lemma~\ref{lem:localtheory}, one gets estimates~\eqref{eq:concludetiny1} and \eqref{eq:concludetiny2} as a priori estimates, meaning that we assume that $\wtbphi_t$ already exists on $I$.
	
	Let $\vect{\psi}_t = \wtbphi_t -\bphi_t$ and notice that it solves the following equation:
	\begin{align}	\label{eq:dynamicsv}
		i\partial_t \vect{\psi} =\, \vect{H}_0 \vect{\psi}+
		F_N(\bphi+\vect{\psi})(\bphi+\vect{\psi}) -F_N(\bphi)\bphi- \vect{e}.
	\end{align}
	Define
	\begin{align*}
		S(I) := \|(i\partial_t-\vect{H}_0) \vect{\psi}_t\|_{L^2_t W^{1, 6/5}_x(I\times\R^3)}.
	\end{align*}
	Using \eqref{eq:closesness}, Strichartz estimates (See \cite[Chapter 2]{tao2006nonlinear}), and Duhamel formula we estimate the $L^{4}_t W^{1,3}_x$ norm of $\vect{\psi}_t$ by $S(I)$:
	\begin{multline}\label{eq:estimate v}
		\|\vect{\psi}_t\|_{L^{4}_t W^{1, 3}_x(I\times \R^3)} 
		\le \|e^{-i(t-t_0)\vect{H}_0}(\wtbphi(t_0)-\bphi(t_0))\|_{L^{4}_t W^{1, 3}_x(I\times \R^3)} \\
		+
		\|\vect{\psi}_t - e^{-i(t-t_0)\vect{H}_0}(\wtbphi(t_0)-\bphi(t_0))\|_{L^{4}_t W^{1, 3}_x(I\times\R^3)}
		\lesssim  \varepsilon+ S(I) .
	\end{multline}
	Hence it remains to estimate $S(I)$. 
	
	According to the hypotheses, it suffices to estimate $F_N(\bphi+\vect{\psi})(\bphi+\vect{\psi})-F_N(\bphi)\bphi$ in $L^2_tW^{1, 6/5}_x$. Observe we have the estimates  
	\begin{align*}
		\Nrm{\grad\(\(U\ast(\phi_\ii\conj{\psi_\ii})\)\psi_\jj\)}{L^2_tL^{6/5}_x}
		\lesssim&\, \Nrm{U}{L^1}\Nrm{\psi_\ii}{L^\infty_tH^{1}_x}\Nrm{\phi_\ii}{L^4_tW^{1, 3}_x}\Nrm{\psi_\jj}{L^4_tW^{1, 3}_x}\\
		&\,+ \Nrm{U}{L^1}\Nrm{\psi_\jj}{L^\infty_tH^{1}_x}\Nrm{\phi_\ii}{L^4_tW^{1, 3}_x}\Nrm{\psi_\ii}{L^4_tW^{1, 3}_x}
	\end{align*}
	and other similar estimates. Hence, using \eqref{eq:Vfa0}, this means that 
	\begin{align}\label{eq:St<=Stpower}
		S(I)\lesssim 
        \varepsilon
        +(\varepsilon+S(I))\varepsilon_0
        + (\varepsilon+S(I))^2
	\end{align}
	where we have used the fact that the energies and masses are conserved and uniformly bounded. 
	If we choose $\eps,\eps_0$ small enough, a standard continuity argument gives $S(I) \lesssim \eps$. To be precise, let us take a small, fixed constant $c>0$ independent of $\eps,\eps_0$  and assume that 
	\begin{align}
		\label{eq:St<=Stpower-a}S(I)\le c \, . 
	\end{align}
	Then from \eqref{eq:St<=Stpower} we deduce that, for $\eps,\eps_0,c$ sufficiently small but independent of $\lambda$, 
	\[
	S(I) \le C_0\eps 
    + C_0\left( \eps_0+ \eps+c  \right) (S(I)+\eps) 
    \le C_0 \eps + \tfrac{1}{2} (S(I)+\eps) \, ,
	\]
	and hence
	\begin{align} \label{eq:St<=Stpower-b}
		S(I)\le (2C_0+1)\eps
	\end{align}
	which completes the proof of the lemma.
\end{proof}

We now prove a version of Lemma~\ref{lem:Short time Pertubation} without the smallness condition \eqref{eq:tiny u} by using Lemma~\ref{lem:Short time Pertubation} iteratively.
\begin{lem} [\textbf{Long time Pertubation}]
	\label{Long time Pertubation}
	Let $I\subset \R$ be a compact interval and $\bphi_t$ be a solution of
	\begin{equation}
		i\partial_t\bphi = \vect{H}_0\bphi_t + F_N(\bphi_t)\bphi_t + \vect{e}
	\end{equation}
	on $I\times\R^3$
	for some time-dependent vector-valued function $\vect{e}$. Moreover, assume that $\bphi$ fulfills the following spacetime bounds:
	\begin{align}
		\label{est:4spacetime}
		\|\bphi\|_{L^4_{t, x}(I\times \R^3)} &\le M, \\
		\|\bphi\|_{L^{\infty}_t \dot H^1_x (I \times \R^3)} &\le E, \\
        \label{est:small of e}
		\|\vect{e}\|_{L^2_t W^{1, 6/5}_x (I \times \R^3)}	&\le\eps
	\end{align}
	for some constants $M,E \ge 0$ and some small enough $\eps > 0$.
	For $t_0\in I$, let $\wtbphi(t_0)$  be close to $\bphi(t_0)$ in the following sense:
	\begin{align}
		\label{eq:closesness2}
		\|e^{-i(t-t_0)\vect{H}_0}(\wtbphi(t_0)-\bphi(t_0))\|_{L^{4}_t W^{1, 3}_x(I\times \R^3)} \le\eps.
	\end{align}
	We then conclude that there exists a solution $\wtbphi_t$ of \eqref{eq:modified_GP_system_vector_form} on $I\times \R^3$ with initial state $\wtbphi(t_0)$ at $t_0$,
	which fulfills the following spacetime bounds:
	\begin{align}
		\label{concludetinylong1}
		\|\wtbphi_t-\bphi_t\|_{L^4_{t, x}(I\times \R^3)} \lesssim \|\wtbphi_t-\bphi_t\|_{L^4_tW^{1, 3}_x(I\times \R^3)}  \lesssim C(M,E)\eps, \\ 
		\label{concludetinylong2}
		\|(i\partial_t -\vect{H}_0) (\wtbphi_t-\bphi_t) \|_{L^2_t W^{1,6/5}_x(I\times \R^3)}	\lesssim C(M,E)\eps.
	\end{align} 
\end{lem}

\begin{proof} 
	Let $\delta >0$. Using \eqref{est:4spacetime}, we split $I$ into finitely many subintervals $I_1 , ... ,  I_{C(M)}$ such that 
	\begin{align*}
		\|\bphi_t\|_{L^4_{t, x}(I_j \times \R^3)} \le \delta
	\end{align*}
	for each $j\in \{1, \ldots, C(M)\}$. Using Strichartz estimate, H\"older inequality, and Sobolev inequality, we obtain
	\begin{align*}
		\|\bphi_t\|_{L^{4}_t 	W^{1, 3}_x (I_j \times \R^3)} 
		&\lesssim
		\|\bphi_t\|_{L^\infty_tH^1_x(I_j\times\R^3)} + \left\|F_N(\bphi_t)\bphi_t + \vect{e}\right\|_{L^2_t W^{1, 6/5}_x (I_j \times \R^3)}  \\
		&\lesssim E + \|\bphi\|_{L^4_{t, x}(I_j \times \R^3)}\|\bphi\|_{L^{\infty}_t H^1_x(I_j \times\R^3)}\|\bphi\|_{L^{4}_{t}L^{12}_x(I_j \times \R^3)} + \eps \\
		&\lesssim E + \delta E\|\bphi\|_{L^{4}_t W^{1, 3}_x(I_j \times\R^3)} + \eps.
	\end{align*}
	Hence, when $\delta$ is sufficiently small, we have $\Nrm{\bphi_t}{L^4_t W^{1, 3}_x(I_j\times \R^3)}\lesssim E+\varepsilon$. Summing up the intervals yields 
	\begin{align}\label{eq:4-3normbound}
		\Nrm{\bphi_t}{L^4_t W^{1, 3}_x(I\times \R^3)}\le C(M, E).
	\end{align}
	
	We choose $\eps_0$ as in Lemma~\ref{lem:Short time Pertubation} and use \eqref{eq:4-3normbound} to split $I$ again into finitely many subintervals $I_1 , ... ,  I_{C(M,E,\eps_0,\lambda)}$ such that
	\begin{align}
		\Nrm{\bphi_t}{L^4_t W^{1, 3}_x(I_j\times \R^3)}\le \eps_0.
	\end{align} 
	We now apply Lemma~\ref{lem:Short time Pertubation} inductively. On the first subinterval $I_1$, we have
	\begin{align*}
		&\|\wtbphi_t-\bphi_t\|_{L^4_{t, x}(I_1\times \R^3)} \lesssim \|\wtbphi_t-\bphi_t\|_{L^4_tW^{1, 3}_x(I_1\times \R^3)} \lesssim \eps_1, \\ 
		&\|(i\partial_t -\vect{H}_0) (\wtbphi_t-\bphi_t) \|_{L^2_t W^{1,6/5}_x(I_1\times \R^3)}	\lesssim \eps_1.
	\end{align*}
	Now proceeding iteratively,  by Strichartz estimate, we see that
	\begin{align*}
		&\|e^{-i(t-t_1)\vect{H}_0}(\wtbphi(t_1)-\bphi(t_1))\|_{L^{4}_t W^{1, 3}_x (I_1 \times \R^3)} \\ &\le \|e^{-i(t-t_1)\vect{H}_0}(\wtbphi(t_0)-\bphi(t_0))\|_{L^{4}_t W^{1, 3}_x (I_1 \times \R^3)} \\
		&\phantom{\le}+ \left\| e^{-i(t-t_1)\vect{H}_0} \int_{I_1} e^{-i(t_1-s)\vect{H}_0} (i\partial_t - \vect{H}_0) (\wtbphi-\bphi)(s) \dd s\right\|_{L^{4}_t 	W^{1, 3}_x(I_1\times\R^3)} \\
		&\lesssim \eps_1 +
		\|(i\partial_t - \vect{H}_0) (\wtbphi_t-\bphi_t) \|_{L^2_t W^{1, 6/5}_x(I_1\times\R^3)}
		\lesssim \eps_1.
	\end{align*}
	For $\eps>0$ small enough, we can iterate this procedure finitely many times and obtain
	\begin{align*}
		\|\wtbphi_t-\bphi_t\|_{L^4_{t, x}(I_j\times \R^3)} \lesssim \|\wtbphi_t-\bphi_t\|_{L^4_tW^{1, 3}_x(I_j\times \R^3)} \lesssim C(j)\eps_j, \\ 
		\|(i\partial_t -\vect{H}_0) (\wtbphi_t-\bphi_t) \|_{L^2_t W^{1,6/5}_x(I_j\times \R^3)}	\lesssim C(j)\eps_j,
	\end{align*}
	for all $j$. Summing over the intervals yields \eqref{concludetinylong1} and \eqref{concludetinylong2}. 
\end{proof}

\begin{prop}\label{prop:uniform_Morawetz}
	Let $\wtbphi^{(N)}_t$ and $\bphi_t$ be $H^1$ solutions to \eqref{eq:modified_GP_system_vector_form} and \eqref{eq:GP_system_vector_form}, respectively, with initial data satisfying 
	\begin{align}
		\|\wtbphi^{(N)}_0-\bphi_0\|_{H^1_x}\le \frac{1}{N}. 
	\end{align}
    We impose Assumption~\ref{assume: L3} for $\lambda = 1$, while both \ref{assume: L3} and~\ref{assume: boundedness} are required when $\lambda = \ln N^\gamma$ with $\gamma \in (0,\frac{1}{2})$.
    We suppose that the matrix $\cc$ in \eqref{eq:GP_system_vector_form} will take the value in the following way, which depends on the constant $\lambda\ge 1$ in system \eqref{eq:modified_GP_system_vector_form}:
    \begin{equation}
        \cc_{\ii} = 
        \begin{cases}
            \asc_{\ii} & \text{ if } \lambda = 1,\\
            \bsc_{\ii} 
            & \text{ if } \lambda = \ln N^{\gamma} \text{ for some } \gamma\in (0,\tfrac{1}{2}).
        \end{cases}
    \end{equation}
	Then we have the following uniform-in-$N$ a priori estimate: there exists a constant $C>0$, depending only on the initial data $\norm{\wtbphi^{(N)}_0}_{H^1_x}, \norm{\bphi_0}_{H^1_x}$, and uniformly bounded in $N$, such that 
	\begin{align}
		\norm{\wtbphi_t^{(N)}}_{L^4_{t, x}(\R\times\R^3)} \le C.
	\end{align} 
	Interpolating with the conservation of laws, we have that 
	\begin{align}
		\norm{\wtbphi_t^{(N)}}_{L^8_{t}L^4_x(\R\times\R^3)} \le C.
	\end{align}
	Moreover, as a consequence of the proof, we also have that 
	\begin{align}
		\|\wtbphi^{(N)}_t-\bphi_t\|_{L^\infty_tH^1_x(\R\times \R^3)} \lesssim \frac{1}{N}+\upepsilon(\lambda).
	\end{align}
\end{prop}
\begin{remark}
	The idea behind the approach of this proposition is to avoid making additional assumptions on the interaction potential $V_\ii$, which was necessary in the work \cite{grillakis2013pair}. This approach is inspired by the analysis in \cite{nam2020derivation}.    
\end{remark}

\begin{proof}
	By the triangle inequality, we have that 
	\begin{align}
		\norm{\wtbphi_t^{(N)}}_{L^4_{t, x}(\R\times\R^3)}\le \norm{\wtbphi_t^{(N)}-\bphi_t}_{L^4_{t, x}(\R\times\R^3)}+\norm{\bphi_t}_{L^4_{t, x}(\R\times\R^3)}.
	\end{align}
	The second term can be bounded by the two-component interaction Morawetz estimate for the cubic defocusing NLSE (see Appendix~\ref{app:interaction_morawetz}). Hence, it remains to bound the first term. 
	
	We want to apply Lemma~\ref{Long time Pertubation}. Write 
	\begin{align*}
		i\partial_t \bphi &= \vect{H}_0 \bphi+F(\bphi)\bphi =  \vect{H}_0 \bphi+F_N(\bphi)\bphi + \vect{e}_N,
	\end{align*}
	with initial state $\bphi_{t}\big|_{t=0} = \bphi_0$ where $\vect{e}_N := \(F_{\cc}(\bphi)-F_N(\bphi)\)\bphi$.
	In order to apply Lemma~\ref{Long time Pertubation}, we need to check that $\|\vect{e}_N\|_{L^2_t W^{1, 6/5}_x (\R \times \R^3)}$ is arbitrarily small and that the initial data are $\wtbphi^{(N)}_0$ and $\bphi_0$ close in the sense of \eqref{eq:closesness2} for $N$ large. 
	
	First, notice that 
	\begin{align*}
		\(F_{\cc}(\bphi)-F_N(\bphi)\)\bphi=
		\begin{pmatrix}
			\((8\pi\cc_1\delta-U_1)\ast|\phi_1|^2+\(8\pi\cc_{12}\delta-U_{12}\)\ast|\phi_2|^2\)\phi_1\\
			\((8\pi\cc_2\delta-U_2)\ast|\phi_2|^2+\(8\pi\cc_{12}\delta-U_{12}\)\ast|\phi_1|^2\)\phi_2
		\end{pmatrix}
	\end{align*}
	then it suffices to estimate the first term
	\begin{align}\label{eq:difference_conv-delta}
		\((8\pi\cc_1\delta-U_1)\ast|\phi_1|^2\)\phi_1
	\end{align}
	since the other terms can be handled in a similar manner. We discuss this term in the following to cases.
    
	We have that 
	\begin{align*}
		\Nrm{\grad \eqref{eq:difference_conv-delta}}{L^2_tL^{6/5}_x(\R\times\R^3)}
		\le&\, \Nrm{(8\pi\cc_1\delta-U_1)\ast|\phi_1|^2}{L^2_tL^{3}_x}\Nrm{\phi_1}{L^\infty_tH^1_x}\\
		&\, + \Nrm{(8\pi\cc_1\delta-U_1)\ast\grad |\phi_1|^2}{L^4_tL^{4/3}_x}\Nrm{\phi_1}{L^4_tL^{12}_x}.
	\end{align*}
	Let us focus on the second term since the first term can be handled similarly. Notice that 
	\begin{subequations}
		\begin{align}
			&\Nrm{8\pi\cc_1 \grad |\phi_1(x)|^2-\intd U_1(y)\grad|\phi_1(x-y)|^2\dd y}{L^4_tL^{4/3}_x} \notag\\
			&\le \Nrm{8\pi(\cc_1-\asc_1^\lambda) \grad |\phi_1(x)|^2}{L^4_tL^{4/3}_x}\label{def:b-a^lambda term}\\
            &\phantom{\le} +\Nrm{\int_{\n{z}\le N\ell} V^\lambda_1(z)\(f^\lambda_1(z)-f^\lambda_{1, \ell}(z)\)\dd z\,\grad |\phi_1(x)|^2}{L^4_tL^{4/3}_x}\label{def:f-f_ell term}\\
			&\phantom{\le}+ \Nrm{\intd U_1(y)\grad_x\(|\phi_1(x)|^2-|\phi_1(x-y)|^2\)\dd y}{L^4_tL^{4/3}_x}\label{def:grad phi-phi_ell term},
		\end{align}
	\end{subequations}
	Then, by Lemma~\ref{lem:neumann_scattering_function}, we have that 
	\begin{align*}
		\n{\eqref{def:b-a^lambda term}+\eqref{def:f-f_ell term}} \le \(\upepsilon(\lambda)+\frac{C\cc_1}{\ell N}\)\Nrm{\phi_1}{L^\infty_t H^1_x} \Nrm{\phi_1}{L^4_{t, x}}.
	\end{align*}
    
	For the other term, we have that 
	\begin{align*}
		\n{\eqref{def:grad phi-phi_ell term}} &\le C\int_{\n{y}\le R_1N^{-1}} U_1(y)\Nrm{\phi_1(x)-\phi_1(x-y)}{L^4_{t, x}}\Nrm{\phi_1}{L^\infty_t H^1_x}\dd y\\
		&\phantom{\le} +C\int_{\n{y}\le R_1N^{-1}} U_1(y)\Nrm{\grad\phi_1(x)-\grad\phi_1(x-y)}{L^\infty_{t}L^2_x}\Nrm{\phi_1}{L^4_{t, x}}\dd y \\
		&\le  C\sup_{\n{y}\le R_1N^{-1}}\Nrm{\phi_1(x)-\phi_1(x-y)}{L^4_{t, x}}\Nrm{\phi_1}{L^\infty_t H^1_x}\\
		&\phantom{\le} + C\sup_{\n{y}\le R_1N^{-1}}\Nrm{\grad\phi_1(x)-\grad\phi_1(x-y)}{L^\infty_{t}L^2_x}\Nrm{\phi_1}{L^4_{t, x}}.
	\end{align*}
	Using the fact that $U_1(y)$ is zero for $|y| > R_1N^{-1}$ for some $R_1>0$, the interaction Morawetz estimate for $\phi_1$, i.e.,  $\left\|\phi_1\right\|_{L^{4}_{t, x}} \le C$, the conservation of energy, and the continuity by translation
	\begin{align*}
		&\lim_{\n{y}\rightarrow 0}\Nrm{\phi_1(x)-\phi_1(x-y)}{L^4_{t, x}} = 0\quad \text{ and } \quad
		\lim_{\n{y}\rightarrow 0}\Nrm{\grad\phi_1(x)-\grad\phi_1(x-y)}{L^\infty_{t}L^2_x}=0,
	\end{align*}
	we see that $\Nrm{\grad \eqref{eq:difference_conv-delta}}{L^2_t L^{6/5}_x(\R\times\R^3)}$ is arbitrarily small when $N$ is large. 
	
	Next, notice by the Strichartz estimate, we have that 
	\begin{align*}
		\|e^{-i(t-t_0)\vect{H}_0}(\wtbphi^{(N)}_{0}-\bphi_{0})\|_{L^{4}_t W^{1, 3}_x}
		\le  \|\wtbphi^{(N)}_{0}-\bphi_{0}\|_{H^1_x}\le \frac{1}{N}. 
	\end{align*}
	Then, by Lemma~\ref{Long time Pertubation}, we arrive at the desired result. 
\end{proof}

\subsection{Propagation of regularity} Let us now prove the uniform-in-$N$ propagation of regularity for \eqref{eq:modified_GP_system_vector_form}.
\begin{prop}\label{prop:propagation_regularity}
	Let $\wtbphi^{(N)}_t$ be as in Proposition~\ref{prop:uniform_Morawetz} with initial data $\wtbphi^{(N)}_0 \in H^s(\RR^3; \CC^2)$ for $s>1$.
	Then, there exists $C$ depending only on $\|\wtbphi^{(N)}_0\|_{H^{s}_x}$ and $\|\bphi_0\|_{H^1_x}$ such that
	\[
	\|\wtbphi^{(N)}_t\|_{L^\infty_tH^{s}_x(\R\times\R^3)}\leq C.
	\]
	Moreover, as an immediate consequence, we have that: for any $j, k \in \N_{0}$, there exist $C=C(\|\bphi_0\|_{H^1_x},\|\wtbphi^{(N)}_0\|_{H^{s}_x} )>0$ for $s$ sufficiently large such that 
	\[
	\|\bd_t^j\wtbphi^{(N)}_t\|_{L^\infty_{t}W^{k, \infty}_x(\R\times\R^3)}\leq C.
	\]
\end{prop}

\begin{proof}
	Split the time interval $[0,\infty)$ into finitely many intervals $I_{k}$ such that
	\[
	\|\wtbphi^{(N)}_t\|_{L_{t}^{8}L_{x}^{4}(I_{k}\times\mathbb{R}^{3})}\leq\delta
	\]
	where we fix $\varepsilon$ later. Differentiating \eqref{eq:modified_GP_system}, we obtain
	\[
	i \partial_{t}D^{s}\widetilde \phi_{\ii}^{(N)} + \Delta D^{s}\widetilde \phi_{\ii}^{(N)} = D^{s}\bigg(\sum_{\jj=1}^{2}(U_{\ii\jj}*|\phi_{\jj}|^{2})\widetilde \phi_{\ii}^{(N)}\bigg)
	\]
	where
	\[
	D_{s}\bigg(\sum_{\jj=1}^{2}(U_{\ii\jj}*|\widetilde\phi_{\jj}^{(N)}|^{2})\widetilde \phi_{\ii}^{(N)}\bigg)=\sum_{\jj=1}^{2}(U_{\ii\jj}*|\widetilde\phi_{\jj}^{(N)}|^{2})D^{s}\widetilde \phi_{\ii}^{(N)}+\substack{\text{ \normalsize similar and lower}\\ \text{\normalsize order terms} }\;.
	\]
	For the first interval $I_{1}$, using the $L_{t}^{8/3}L_{x}^{4}$ Strichartz estimate along with the $L_{t}^{8/5}L_{x}^{4/3}$ dual Strichartz estimate and H\"older inequality, we get 
	\begin{align*}
		\|D^{s}\widetilde \phi_{\ii}^{(N)}\|_{L_{t}^{8/3}L_{x}^{4}(I_{1}\times\mathbb{R}^{3})}  
		\leq\, C\|\widetilde \phi_{\ii, 0}^{(N)}\|_{H^{s}_x}
		+\sum_{\jj=1}^{2}\lambda C_{\ii\jj}\|\widetilde\phi_{\jj}^{(N)}\|_{L_{t}^{8}L_{x}^{4}(I_{1}\times\mathbb{R}^{3})}^{2}\|\widetilde \phi_{\ii}^{(N)}\|_{L_{t}^{8/3}L_{x}^{4}(I_{1}\times\mathbb{R}^{3})}.
	\end{align*}
	We choose $\delta$ such that $\max_{1\leq \ii\leq 2}\sum_{\jj=1}^{2}\lambda C_{\ii\jj}\delta^{2}\leq\frac{1}{2}$
	so could have
	\[
	\|D^{s}\widetilde \phi_{\ii}^{(N)}\|_{L_{t}^{8/3}L_{x}^{4}(I_{1}\times\mathbb{R}^{3})}\leq2C\|\widetilde \phi_{\ii, 0}^{(N)}\|_{H^{s}_x}.
	\]
	Then, we control the inhomogeneity, i.e., 
	\begin{align*}
		&\sum_{\jj=1}^{2}\|(U_{\ii\jj}*|\widetilde\phi_{\jj}^{(N)}|^{2})D^{s}\widetilde \phi_{\ii}^{(N)}\|_{L_{t}^{8/5}L_{x}^{4/3}(I_{1}\times\mathbb{R}^{3})}
		\le C\|\widetilde \phi_{\ii, 0}^{(N)}\|_{H^{s}_x}.
	\end{align*}
	Therefore, we have the estimate
	\begin{align*}
		\|\widetilde \phi_{\ii}^{(N)}(t,\cdot)\|_{H^{s}_x}\lesssim&\, \|\widetilde \phi_{\ii, 0}^{(N)}\|_{H^{s}_x}+\sum_{\jj=1}^{2}\|(U_{\ii\jj}*|\widetilde\phi_{\jj}^{(N)}|^{2})D^{s}\widetilde \phi_{\ii}^{(N)}\|_{L_{t}^{8/5}L_{x}^{4/3}(I_{1}\times\mathbb{R}^{3})}
		\lesssim\, \|\widetilde \phi_{\ii, 0}^{(N)}\|_{H^{s}_x}
	\end{align*}
	for all $t\in I_{1}$. Repeating the process a finite number of times yields the desired result.
\end{proof}

\subsection{Dispersive Estimate} If we assume that the data $\wtbphi^{(N)}_0$ with sufficiently
many derivatives are not only in $L^{2}$ but also in $L^{1}$, we
can also obtain dispersive estimate.
\begin{cor}\label{cor:Linfty-t-indep-bdd}
	Let $\wtbphi^{(N)}_t$ be as in Proposition~\ref{prop:propagation_regularity}. There exists $C$ depending only on $\|\wtbphi_{0}^{(N)}\|_{W^{k,1}_x}$ and $\|\wtbphi_{0}^{(N)}\|_{H^{k}_x}$ for $k$ sufficiently large such that
	\begin{align}
		\|\wtbphi_{t}^{(N)}\|_{W^{1, \infty}_x}+\|\partial_{t}\wtbphi_{t}^{(N)}\|_{L^{\infty}_x}  \leq\frac{C}{1+t^{\frac{3}{2}}}\;.
	\end{align}
\end{cor}

\begin{remark}
	The proof of linear-like dispersive estimates for solutions to the defocusing cubic NLSE can be traced back to the work of Lin and Strauss in \cite{lin1978decay}, which assumes sufficient regularity in the initial data. More recently, the regularity requirement has been substantially reduced in \cite{fan2021decay, fan2024decaying}. For our purposes, however, it suffices to follow the former approach.
\end{remark}

\begin{proof}
	The proof is similar to the one given in \cite[Corollary 3.4]{grillakis2013pair}. Hence, we only provide a quick sketch of the argument. 
	
	First, notice that $\norm{\wtbphi^{(N)}_t}_{L^\infty_x}\rightarrow 0$ as $t\rightarrow \infty$. Indeed, by Proposition~\ref{prop:uniform_Morawetz}, the second part of Proposition~\ref{prop:propagation_regularity}, and the non-sharp Sobolev inequality, we have that 
	\begin{align*}
		\norm{(\widetilde\phi^{(N)}_{\ii})^2}_{L^p_{t, x}([n, n+1]\times\R^3)} \lesssim&\,  \norm{\grad_{x, t}(\widetilde\phi^{(N)}_{\ii})^2}_{L^4_{t, x}([n, n+1]\times\R^3)}\\
		\lesssim&\, \norm{\grad_{x, t}\widetilde\phi^{(N)}_{\ii}}_{L^\infty_{t, x}([n, n+1]\times\R^3)}\norm{\widetilde\phi^{(N)}_{\ii}}_{L^4_{t, x}([n, n+1]\times\R^3)}\rightarrow 0
	\end{align*}
	provided $p \in (4, \infty)$. Applying the argument again, we see that 
	\begin{align*}
		\norm{(\widetilde\phi^{(N)}_{\ii})^2}_{L^\infty_{t, x}([n, n+1]\times\R^3)} \lesssim&\,  \norm{\weight{\grad_{x, t}}(\widetilde\phi^{(N)}_{\ii})^2}_{L^{9}_{t, x}([n, n+1]\times\R^3)}\\
		\lesssim&\, \norm{\grad_{x, t}\widetilde\phi^{(N)}_{\ii}}_{L^{\infty}_{t, x}([n, n+1]\times\R^3)}\norm{\widetilde\phi^{(N)}_{\ii}}_{L^{9}_{t, x}([n, n+1]\times\R^3)}\rightarrow 0.
	\end{align*}
	
	To prove the main result, it suffices to handle the nonlinearity. Using the standard $L_{t}^{\infty}L_{x}^{1}$
	dispersive estimate for the free evolution, we have
	\begin{align}\label{def:L-infty_nonlinearity}
		&\|\sum_{\jj=1}^{2}e^{i(t-s)\Delta}\left((U_{\ii\jj}*|\widetilde\phi_{\jj}^{(N)}|^{2})\widetilde \phi_{\ii}^{(N)}(s)\right)\|_{L^{\infty}_x} 
		\leq\frac{C}{|t-s|^{\frac32}}\|\widetilde \phi_{\ii}^{(N)}(s)\|_{L^{\infty}_x}. 
	\end{align}
	On the other hand, using Sobolev inequality, we also have
	\begin{align*}
		\text{LHS of }\eqref{def:L-infty_nonlinearity}
		& \leq\sum_{\jj=1}^{2}\|\nabla e^{i(t-s)\Delta}\left((U_{\ii\jj}*|\widetilde\phi_{\jj}^{(N)}|^{2})\widetilde \phi_{\ii}^{(N)}(s)\right)\|_{L^{3}_x}\,.
	\end{align*}
	Notice that this is a false endpoint, but becomes
	true if one replaces 3 by $3+\epsilon$ for any fixed $\epsilon>0$. To avoid unnecessary notations, we stick to the numerology provided by the false endpoint since it will not effect the outcome.
	
	Applying Young's inequality, interpolation, $H^1$-norm bound, and Strichartz estimate for $L^4$, we get that
	\begin{equation}\label{est:dispersive_L3}
		\|\nabla e^{\mathrm{i}(t-s)\Delta}(U_{\ii\jj}*|\widetilde\phi_{\jj}^{(N)}|^{2})\widetilde \phi_{\ii}^{(N)}(s)\|_{L^{3}_x}
		\leq\frac{C}{|t-s|^{\frac12}} \sum^2_{\jj=1}\|\widetilde\phi_{\jj}^{(N)}(s)\|_{L^{\infty}_x}^{\frac43}.
	\end{equation}
	Combining \eqref{def:L-infty_nonlinearity} and \eqref{est:dispersive_L3} we see there exists a kernel $k\in L^1([0, \infty))$ such that 
	\[
	\text{LHS of }\eqref{def:L-infty_nonlinearity}
	\leq 
	k(t-s)
	\sum_{\jj=1}^{2}\|\widetilde\phi_{\jj}^{(N)}(s)\|_{L^{\infty}_x}^{1+\delta}\,.
	\]
	In all, we have
	\begin{align*}
		\|\widetilde \phi_{\ii}^{(N)}(t)\|_{L^{\infty}_x}\leq&\, \frac{C}{t^{\frac32}}\|\widetilde \phi_{\ii}^{(N)}(0)\|_{L^{1}_x} +C\int_{0}^{\frac{t}{2}}\frac{\|\widetilde \phi_{\ii}^{(N)}(s)\|_{L^{\infty}_x}}{|t-s|^{\frac32}}\dd s
		+
		\int_{\frac{t}{2}}^{t}k(t-s)\sum_{j=1}^{2}\|\widetilde\phi_{\jj}^{(N)}(s)\|_{L^{\infty}_x}^{1+\delta}\dd s\,.
	\end{align*}
	Denoting 
	\[
	M_\ii(t)=\sup_{0<s<t}(1+s^{\frac32})\|\widetilde\phi_{\jj}^{(N)}(s)\|_{L^{\infty}_x}\, \quad \text{ and } \quad M(t) = \sum^2_{\jj=1} M_{\ii}(t)
	\]
	we obtain for $t>1$ that
	\begin{align}
		M(t)\leq C\sum_{\ii=1}^{2}\|\widetilde \phi_{\ii}^{(N)}(0)\|_{L^{1}_x}&+C\int_{0}^{t/2}\frac{M(s)}{1+s^{\frac32}} \dd s
		+C\sum_{\ii=1}^{2} \sup_{\frac{t}{2}<s<t}\|\widetilde \phi_{\ii}^{(N)}(s)\|_{L^{\infty}_x}^{\delta} M_{\ii}(t)\,.
	\end{align}
	
	As we have $\|\widetilde \phi_{\ii}^{(N)}(t)\|_{L^\infty_x}\to 0$ when $t\to \infty$, then there exists a $T>0$ such that $\|\widetilde \phi_{\ii}^{(N)}(t)\|_{L^\infty_x}^\delta<\frac{1}{2C}$. 
	Hence, the result follows by applying Grönwall's inequality.
	
	Estimating $\grad\widetilde \phi_{\ii}^{(N)}$ and $\partial_{t}\widetilde \phi_{\ii}^{(N)}$ in $L^\infty_x$ is straightforward using the previously obtained estimate $\norm{\wtbphi^{(N)}_{ t}}_{L^\infty_x} \lesssim (1+t^\frac32)^{-1}$ and Proposition~\ref{prop:propagation_regularity}. More precisely, we have that
	\begin{align*}
		\|\wtbphi_{t}^{(N)}\|_{W^{1, \infty}_x} \le&\, \frac{C}{t^{\frac32}} \|\wtbphi^{(N)}_0\|_{W^{1, \infty}_x} + \int_0^t
		\Big\|e^{-i (t-s)\vect{H}_0} F_N(\wtbphi^{(N)})\wtbphi^{(N)} (s)\Big\|_{W^{1, \infty}_x}\d s\\
		\le&\, \frac{C}{t^{\frac32}}+ C \int_0^{t-1} \frac{1}{1+|t-s|^{\frac32}}\Big\|F_N(\wtbphi^{(N)})\wtbphi^{(N)} (s)\Big\|_{W^{1, 1}_x}\d s\\
		&+ C \int_{t-1}^{t} \frac{1}{1+|t-s|^{\frac12 + \epsilon}}\Big\|F_N(\wtbphi^{(N)})\wtbphi^{(N)} (s)\Big\|_{W^{1, \frac32-\epsilon'}_x} \d s\\
		\le&\ \frac{C}{t^{\frac32}} + C \int_0^{t-1} \frac{1}{1+|t-s|^{\frac32}}\|\wtbphi^{(N)}_s\|_{L^{\infty}_x} \d s\\
		&+ C \int_{t-1}^{t} \frac{1}{1+|t-s|^{\frac12 + \epsilon}}\|\wtbphi^{(N)}_{s}\|_{L^{\infty}_x} \d s,
	\end{align*}
	which yields the desired estimate. The same argument applies to estimating $\bd_t\wtbphi^{(N)}$.
\end{proof}

\section{Fluctuation Dynamics}\label{sect:fluct_dynamic}

\subsection{Fock Space Formalism}\label{sec: Fock}
We denote the $\jj$-th species one-particle state space by $\h_\jj = L^2(\R^3;\CC)$, which is endowed with the inner product $\inprod{\cdot}{\cdot}_{\mathfrak{h}}$ that is linear
in the second entry and conjugate linear in the first entry. Let $\fH=\h_1\oplus\h_2$ be the two-component state space. Consider the corresponding bosonic Fock space $\cF=\cF(\fH):=\bigoplus_{N=0}^\infty \bigotimes_{s}^N \fH$.
By standard results in algebra, we have the following canonical isomorphism
\begin{equation*}
	\cF=\cF_1\otimes\cF_2 :=\cF(\h_1)\otimes\cF(\h_2)=\bigoplus_{L=0}^\infty\bigoplus_{\substack{n,m\in\mathbb{N}_0 \\ n+m=L}} \h_1^{\otimes_{s}n} \otimes \h_2^{\otimes_{s}m}
\end{equation*}
where every state vector $\Psi \in \cF$ can be expressed as a doubly-infinite array  
\begin{align*}
	\Psi= 
	\begin{pmatrix}
		\psi_{00} & \psi_{1, 0}(x_1) & \psi_{2, 0}(x_1, x_2) & \cdots \\
		\psi_{0, 1}(y_1) & \psi_{1, 1}(x_1; y_1) & \psi_{2, 1}(x_1, x_2; y_1) & \cdots\\
		\psi_{0, 2}(y_1, y_2) & \psi_{1, 2}(x_1; y_1, y_2) & \psi_{2, 2}(x_1, x_2; y_1, y_2) & \cdots\\
		\vdots & \vdots & \vdots & \ddots
	\end{pmatrix}
\end{align*}
with $\psi_{n, m}:\mathbb{R}^{3n} \times \mathbb{R}^{3m}\to \mathbb{C}$ for a pair of nonnegative integers $(n,m)$ except $(0,0)$. For $(0,0)$, we set $\psi_{00} \in \mathbb{C}$.

It is clear that $\cF$ is a Hilbert space when endowed with the inner product
\begin{align*}
	\inprod{\Psi}{\Phi}_{\cF} =  \overline{\psi_{00}}\,\phi_{00} +  \sum^\infty_{\substack{n,m=0\\(n,m)\neq(0,0)}}\inprod{\psi_{n, m}}{\phi_{n, m}}_{L^2(\R^{3(n+m)})}.
\end{align*}
We also write $\Nrm{\Psi}{}=\sqrt{\inprod{\Psi}{\Psi}_{\cF}}$ to denote the Fock space norm. 
The vacuum is defined as the Fock element $\Omega=\Omega_{\cF_1}\otimes \Omega_{\cF_2}$ with $\Omega_{\cF_1} = (1, 0, 0, \ldots)$, which describes a state without particles.

For every $x \in \R^3$, we define the creation and annihilation operator-valued distributions, denoted by $a^{\ast}_x$ and $a_x$ respectively, by their actions on the sectors $\h_1^{\otimes_s n-1}$ and $\h_1^{\otimes_s n+1}$ of $\cF_1$ as follows
\begin{align*}
	(a^{\ast}_x\Psi)_{n, m} =&\, \frac{1}{\sqrt{n}}\sum_{j=1}^n \delta(x-x_{j})\,\psi_{n-1, m}(x_1, \ldots, \cancel{x_{j}},\ldots, x_n; y_1, \ldots, y_m) \, ,\\
	(a_x\Psi)_{n, m} =&\, \sqrt{n+1}\, \psi_{n+1, m}(x, x_1, \ldots, x_n; y_1, \ldots, y_m) \, .
\end{align*}
Similarly, we define creation and annihilation operator-valued distributions on $\cF_2$, denoted by $b^{\ast}_x$ and $b_x$.
It is straightforward to check that they satisfy the canonical commutation relations (CCR)
\begin{equation}\label{CCR}
	\begin{aligned}
		[a_x,a_y]  =&\, [a^{*}_x,a^{*}_y] =0\,, \quad& [a_x,a^{*}_y]=&\, \delta(x-y)\,,\\
		[b_x,b_y]  =&\, [b^{*}_x,b^{*}_y] =0\,, \quad& [b_x,b^{*}_y]=&\, \delta(x-y)\,,
	\end{aligned}
\end{equation}
and that $a, a^\ast$ commute with $b, b^\ast$. Moreover, for every $f \in \h_i$, we write 
\begin{align}
	a(f):= \intd \conj{f(x)}\,a_x\dd x\quad \text{ and }\quad a^\ast(f):= \intd f(x)\,a^\ast_x \dd x\ .
\end{align}
Likewise for $b(f), b^\ast(f)$. Furthermore, by direct computation, one can show that the creation and annihilation operators are bounded by the square root of the number of particles operator, i.e., we have 
\begin{equation}\label{est:creation and annihilation} 
	\begin{aligned}
		\Nrm{a (f) \Psi}{}\leq \Nrm{f}{L^2} \norm{\cN_1^{\frac12} \Psi}, \quad \Nrm{a^\ast(f) \Psi}{}\leq \Nrm{f}{L^2} \Nrm{(\cN_1+1)^{\frac12} \Psi}{}, \\
		\Nrm{b(f) \Psi}{}\leq \Nrm{f}{L^2} \norm{\cN_2^{\frac12} \Psi}, \quad \Nrm{b^\ast(f) \Psi}{}\leq \Nrm{f}{L^2} \Nrm{(\cN_2+1)^{\frac12} \Psi}{}, 
	\end{aligned}
\end{equation}
for every $f \in L^2 (\R^3)$ and $\Psi \in \cF$. 

Given any $\vect{f} = (f, g) \in \fH$, we define the annihilation and creation operators on $\fH$  by 
\begin{align*}
	z(\vect{f}) :=&\, 
	a(f)+b(g)
	=: \intd 
	(a_x,\; b_x) 
	\begin{pmatrix}
		\conj{f(x)}\\[.5ex]
		\conj{g(x)}
	\end{pmatrix}
	\d x \,,
	\\
	z^\ast(\vect{f}) :=&\, 
	a^\ast(f)+b^\ast(g)
	=: \intd \(a_x^\ast,\; b_x^\ast\) 
	\begin{pmatrix}
		f(x)\\[.5ex]
		g(x)
	\end{pmatrix}
	\d x \,.
\end{align*}
We also write the corresponding operator-valued distribution vectors $z_x = (a_x,\; b_x)^\top$ and $z_x^\ast = (a^\ast_x,\; b^\ast_x)^\top$. Notice the operators satisfy the following commutation relations: for every $\vect{f}, \vect{g} \in \fH$, we have
\begin{equation}
	\com{z(\vect{f}), z^\ast(\vect{g})}= \inprod{\vect{f}}{\vect{g}}_{\fH}, \quad \com{z(\vect{f}), z(\vect{g})}=\com{z^\ast(\vect{f}), z^\ast(\vect{g})}=0.
\end{equation}

For  general self-adjoint operators $O_1$ on $\h_1$ and $O_2$ on $\h_2$, we denote
\begin{equation} \label{eq:second_quantization}
	\begin{split}
		\dGamma_1(O_{1})=&\, \dGamma(O_{1}\oplus 0)\;:=\;\iintd O_{1}(x, y)\, a^\ast_x a_y\dd x\d y, \\
		\dGamma_2(O_{2})=&\, \dGamma(0\oplus O_{2})\;:=\;\iintd O_{2}(x, y)\, b^\ast_x b_y\dd x\d y.
	\end{split}
\end{equation}
Number operators for each species are given by
\begin{equation}
	\cN_1:=\dGamma_1(\id)=\;\intd a^\ast_x a_x\dd x\,,\qquad  \cN_2:=\dGamma_2(\id)=\;\intd b^\ast_x b_x\dd x\,,
\end{equation}
and we define the total number operator on $\cF$ such that
\begin{equation}
	\cN\;:=\intd (z_x^\ast)^\top z_x\dd x=\;\cN_1+\cN_2\,.
\end{equation}
By direct computation, one can show that 
\begin{align}\label{eq:number_operator_commuting_z}
	\com{\cN, z(\vect{f})} = -z(\vect{f}) \quad \text{ and } \quad \com{\cN, z^\ast(\vect{f})} = z^\ast(\vect{f}). 
\end{align}
Furthermore, using the notations introduced above,  we can rewrite the Hamiltonian \eqref{intra and inter Hamiltonian} in the following form:
\begin{equation}\label{eq:multiFockH}
	\left\{\begin{aligned}
		\cH_{1}=&\, \dGamma_1(-\lapl)+\frac12\iintd N^2 V^{\lambda}_{1}(N(x-y))\, a^{\ast}_x a^{\ast}_y a_y a_x\dd x\d y =:\cK_1+\cV_1\,, \\
		\cH_{2}=&\, \dGamma_2(-\lapl) +\frac12\iintd N^2 V^{\lambda}_{2}(N(x-y))\, b^{\ast}_x b^{\ast}_y b_y b_x\dd x\d y =:\cK_2+\cV_2\,, \\
		\cV_{12}=&\, \iintd N^2 V^{\lambda}_{12}(N(x-y))\,a^{\ast}_x a_x b^{\ast}_y b_y\dd x\d y \, .
	\end{aligned}\right.
\end{equation}

Let $\wt{\vect{\phi}}^{(N)}_t = (\wt{\phi}^{(N)}_{1,t},\,\wt{\phi}^{(N)}_{2,t})^\top$ be the solution of the modified GP system \eqref{eq:modified_GP_system}, we associate the corresponding operators
\begin{align}
	\cA(\wt{\vect{\phi}}^{(N)}_t) =a(\wt{\phi}^{(N)}_{1,t})+b(\wt{\phi}^{(N)}_{2,t})-a^\ast(\wt{\phi}^{(N)}_{1,t})-b^\ast(\wt{\phi}^{(N)}_{2,t})=:\cA_1(\wt{\phi}^{(N)}_{1,t})+\cA_2(\wt{\phi}^{(N)}_{2,t})\, .
\end{align}
 Let us introduce the two-component Weyl operator
\begin{align}\label{eq:two-component_weyl_conjugation}
	\cW(\sqrt{N}\wt{\vect{\phi}}^{(N)}_t)
    :=e^{-\sqrt{N}\cA(\wt{\vect{\phi}}^{(N)}_t)} =
    e^{-\sqrt{N}\cA_1(\wt{\phi}^{(N)}_{1,t})}
    e^{-\sqrt{N}\cA_2(\wt{\phi}^{(N)}_{2,t})} \,. 
\end{align} 
Since the operator $\cA(\wt{\vect{\phi}}^{(N)}_t)$ is skew-adjoint, the corresponding Weyl operator is unitary.  A notable property of the Weyl operator is its actions on the creation and annihilation operators. 
More precisely, we have 
\begin{equation}\label{eq:weyl_conjugation}
	\begin{aligned}
		\cW^\ast(\sqrt{N}\wt{\vect{\phi}}^{(N)}_t)a_x\cW(\sqrt{N}\wt{\vect{\phi}}^{(N)}_t)=&\, a_x + \sqrt{N}\wt\phi_{1, t}^{(N)}(x),\\
		\cW^\ast(\sqrt{N}\wt{\vect{\phi}}^{(N)}_t)b_x\cW(\sqrt{N}\wt{\vect{\phi}}^{(N)}_t) =&\,  b_x + \sqrt{N}\wt\phi_{2, t}^{(N)}(x), 
	\end{aligned}
\end{equation}
or, equivalently,
\begin{align}
	\cW^\ast(\sqrt{N}\wt{\vect{\phi}}^{(N)}_t)z_x\cW(\sqrt{N}\wt{\vect{\phi}}^{(N)}_t) = z_x+\sqrt{N}\wt{\vect{\phi}}^{(N)}_t.
\end{align}
Moreover, the Weyl operator can be used to describe a two-component BEC in the Fock space $\cF$ via a two-component coherent state, which takes the form:
\begin{align}
	\cW(\sqrt{N}\wt{\vect{\phi}}^{(N)}_t) \Omega
	= e^{-N/2}
    \sum^{\infty}_{n,m=0}
    \frac{a^\ast(\sqrt{N} \wt\phi^{(N)}_{1, t})^n\, b^\ast(\sqrt{N} \wt\phi^{(N)}_{2, t})^m\,\Omega}{ n!\;m!} \, .
\end{align} 





\subsection{Generator of the Fluctuation Dynamics}

Recall the fluctuation dynamics satisfies the following Schr\"odinger-type equation
\begin{equation}\label{fluctuation dynamics}
	i \bd_t \mathcal{U}_N(t ; s)=\cL_N(t) \cU_N(t ; s) = \Big(\(i \bd_t \cT_t^*\) 
	\cT_t +
	\cT_t^\ast
	\cG_N(t)
	\cT_t  \;\Big)\cU_N(t; s)
\end{equation}
where 
\begin{align*}
	\cG_N:=
	\Big(i \bd_t \cW^\ast(\sqrt{N} \wt{\boldsymbol{\phi}}^{(N)}_t)\Big)
	\cW(\sqrt{N} \wt{\boldsymbol{\phi}}^{(N)}_t)
	+
	\cW^\ast(\sqrt{N} \wt{\boldsymbol{\phi}}^{(N)}_t)
	\cH_N 
	\cW(\sqrt{N} \wt{\boldsymbol{\phi}}^{(N)}_t).
\end{align*}
It is straightforward to compute $\cG_N(t)$ using the identities in \eqref{eq:weyl_conjugation}. In fact, the computation for $\cG_N(t)$ can be readily found in \cite{MichelangeliOliveiraCoulomb}. Hence, let us simply state the result here without proof. One can write the generator $\cG_N(t)$ as follows: suppose $\wtbphi^{(N)}_{t}$ solves \eqref{eq:modified_GP_system}, then we have that
\begin{equation}\label{decompose generator}
	\cG_{N}= N\mu_0 + \cH_N+\cG_{1}+\cG_{2}+\cG_{3},
\end{equation}
where each $\cG_{i}$ corresponds to the sum of all operators of $\cG_N$ with $i$ number of creation and annihilation operators. 

For $\cG_{1}$, using $w_{\ii, \ell}(N\cdot)=1-f_{\ii, \ell}(N\cdot)$, we have that 
\begin{equation}\label{def_G1}
	\begin{aligned}
		\cG_1=&\, \sqrt{N}\,a^*\((N^{3}V^\lambda_1(N\;\cdot\;)w_{1, \ell}(N\;\cdot\;)\ast |\wt\phi^{(N)}_{1, t}|^2)\wt\phi^{(N)}_{1, t}\)\\
		& +\sqrt{N}\, a^\ast\((N^{3}V^\lambda_{12}(N\;\cdot\;)w_{12, \ell}(N\;\cdot\;)\ast |\wt\phi^{(N)}_{2, t}|^2)\wt\phi^{(N)}_{1, t}\)\\
		&+\sqrt{N}\,b^*\((N^{3}V^\lambda_2(N\;\cdot\;)w_{2, \ell}(N\;\cdot\;)\ast |\wt\phi^{(N)}_{2, t}|^2)\wt\phi^{(N)}_{2, t}\)\\
		& +\sqrt{N}\, b^\ast\((N^{3}V^\lambda_{12}(N\;\cdot\;)w_{12, \ell}(N\;\cdot\;)\ast |\wt\phi^{(N)}_{1, t}|^2)\wt\phi^{(N)}_{2, t}\) + \mathrm{h.c.}.
	\end{aligned}
\end{equation}
For $\cG_{2}$, we write $\cG_2 = \cG^{(1)}_2 + \cG^{(2)}_2 + \cG^{(12)}_{2}$, where $\cG^{(\ii)}_2$ are the intra-species terms given by 
\begin{equation}\label{def:intra_quadratic}
	\left\{
	\begin{aligned}
		\cG_{2}^{(1)}
		=&\, \frac12\intd \big(N^{3}V^\lambda_{1}(N\;\cdot\;)\ast |\wt\phi^{(N)}_{1, t}|^2\big)(x)\, a^{\ast}_xa_x\dd x\\
		&\,+\frac12\iintd N^{3}V^\lambda_{1}(N(x-y))\, \wt\phi^{(N)}_{1, t}(x) \, \conj{\wt\phi^{(N)}_{1, t}(y)}\, a^{\ast}_xa_y\dd x\d y \\
		&\, +\frac{1}{2}\iintd N^{3}V^\lambda_{1}(N(x-y))\, \wt\phi^{(N)}_{1, t}(x) \,\wt\phi^{(N)}_{1, t}(y)\, a^{\ast}_xa^{\ast}_y \dd x\d y + \mathrm{h.c.}, \\
		\cG_{2}^{(2)}
		=&\, \frac12\intd \big(N^{3}V^\lambda_{1}(N\;\cdot\;)\ast |\wt\phi^{(N)}_{1, t}|^2\big)(x)\, b^{\ast}_xb_x\dd x\\
		&\,+\frac12\iintd N^{3}V^\lambda_{1}(N(x-y))\, \wt\phi^{(N)}_{1, t}(x)  \conj{\wt\phi^{(N)}_{1, t}(y)}\, b^{\ast}_xb_y\dd x\d y \\
		&\, +\frac{1}{2}\iintd N^{3}V^\lambda_{2}(N(x-y))\, \wt\phi^{(N)}_{2, t}(x) \wt\phi^{(N)}_{2, t}(y)\, b^{\ast}_xb^{\ast}_y \dd x\d y + \mathrm{h.c.}.
	\end{aligned}
	\right.
\end{equation}
The inter-species term reads
\begin{align}\label{def: inter in quadratic}
	\cG_{2}^{(12)}=& \frac12\intd (N^{3} V^\lambda_{12}(N\;\cdot\;)\ast |\wt\phi^{(N)}_{2, t}|^2)(x)\,a^{\ast}_xa_x\dd  x  \\
	& + 
	\frac12\intd (N^{3} V^\lambda_{12}(N\;\cdot\;)\ast |\wt\phi^{(N)}_{1, t}|^2)(x)\, b^{\ast}_xb_x\dd  x \notag\\
	&+\iintd N^{3} V^\lambda_{12}(N(x-y))\wt\phi^{(N)}_{1, t}(x) \conj{\wt\phi^{(N)}_{2, t}(y)}\,a^{\ast}_x b_y 
	\dd x\d y \notag\\
	&+\iintd N^{3} V^\lambda_{12}(N(x-y))\, \wt\phi^{(N)}_{1, t}(x)
	\wt\phi^{(N)}_{2, t}(y)\, a^{\ast}_x b^{\ast}_y \dd x\d y + \mathrm{h.c.}. \notag
\end{align}
Again, for $\cG_{3}$, we  write $\cG_{3} = \cG_{3}^{(1)} + \cG_{3}^{(2)} + \cG_{3}^{(12)}$.
Here, the intra-species terms are given by
\begin{equation}\label{def:G_cubic_terms_intra-species}
	\left\{\begin{aligned}
		&\cG_{3}^{(1)}=\frac{1}{\sqrt{N}}\iintd N^{3}V^\lambda_{1}(N(x-y))  \wt\phi^{(N)}_{1, t}(y)\, a^{\ast}_y a^{\ast}_x   a_x \dd x\d y +\mathrm{h.c.}\\
		&\cG_{3}^{(2)}= \frac{1}{\sqrt{N}}\iintd N^{3}V^\lambda_{2}(N(x-y)) \wt\phi^{(N)}_{2, t}(y)\, b^{\ast}_y b^{\ast}_x   b_x \dd x\d y+\mathrm{h.c.}\,,
	\end{aligned}\right.
\end{equation}
and inter-species term is
\begin{align}\label{def:G_cubic_terms_inter-species}
	\cG^{(12)}_{3}=&  \frac{1}{\sqrt{N}}\iintd N^{3} V^\lambda_{12}(N(x-y))\,\wt\phi^{(N)}_{2, t}(y)\, b^{\ast}_y a^{\ast}_x   a_x\dd x\d y\\
	&+ \frac{1}{\sqrt{N}}\iintd N^{3} V^\lambda_{12}(N(x-y))\,\wt\phi^{(N)}_{1, t}(y)\, a^{\ast}_y b^{\ast}_x b_x\dd x\d y+ \mathrm{h.c.}.\notag
\end{align}
Lastly, the constant term is given by 
\begin{equation}\label{define constant term}
	\begin{aligned}
		\mu_0(t) =
		-\frac{1}{2}\sum^2_{\ii, \jj=1}\iintd N^{3}V^\lambda_{\ii\jj}(N(x-y))|\wt\phi^{(N)}_{\ii, t}(x)|^2|\wt\phi^{(N)}_{\jj, t}(y)|^2\dd x\d y \,.
	\end{aligned}
\end{equation}

The computation of $\cL_N(t)$ is a bit tedious. We shall devote the next subsection to the Bogoliubov transformation and develop the necessary identities to compute the generator of the fluctuation dynamics. 

\subsection{Bogoliubov Transformation}\label{subsect:bogoliubov_transf}

The results of this section are 
standard. For a more comprehensive review, we refer the reader to \cite{folland1989harmonic, grillakis2013beyond, solovej2007many}.

Let $\mathbf{f}=(f_1, f_2)$, $\mathbf{g}=(g_1, g_2) \in \fH$, and $\vect{J}=J \oplus J:\fH \to \fH$ so that
$\vect{J}\vect{g} = (J g_1, Jg_2)$ for any antilinear map $J$, i.e., $Jf = \conj{f}$. 
Moreover, let $\mathbf{T}:\fH\to\fH$ be a bounded operator with the matrix of kernels 
\begin{align}
	\mathbf{T}(x, y)=
	\begin{pmatrix}
		T_{11}(x, y) & T_{12}(x, y)\\
		T_{21}(x, y) & T_{22}(x, y)
	\end{pmatrix}
	\quad \text{ where } \quad T_{\ii\jj}:\h_{\jj}\to \h_{\ii}.
\end{align}
We write 
\begin{align}
	\n{\mathbf{T}(x, y)}_{\rm F}^2:=\sum^2_{\ii=1}\sum^2_{\jj=1} \n{T_{\ii\jj}(x, y)}^2,
\end{align}
then the Hilbert--Schmidt norm of $\mathbf{T}$ can be written as 
\begin{align}
	\Nrm{\mathbf{T}}{\rm HS}^2 =\iintd \n{\mathbf{T}(x, y)}_{\rm F}^2\dd x\d y.
\end{align}

The generalized annihilation and creation operators over $\fH\oplus\fH^\ast$ are defined by 
\begin{align}\label{def: Z}
	Z(\mathbf{f}\oplus\mathbf{J}\mathbf{g}) =\, z(\mathbf{f}) + z^\ast(\mathbf{g})\quad \text{ and } \quad
	Z^\ast(\mathbf{f}\oplus\mathbf{J}\mathbf{g}) =\, z^\ast(\mathbf{f}) + z(\mathbf{g}).
\end{align}
Here, $Z$ is a conjugate linear map from $\fH\oplus\fH^\ast$ to the space of operators over $\cF$, while $Z^\ast$ is linear. 
Then similar to the one-component case, 
we have the following adjoint relation
\begin{equation}\label{adjoint relation}
	Z^\ast(\vect{F})
	=
	Z(\sfJ\vect{F})
	\quad \text{ where } \quad \sfJ = 
	\begin{pmatrix}
		0 & \mathbf{J}^\ast \\
		\mathbf{J} & 0 
	\end{pmatrix}:\fH\oplus\fH^\ast \rightarrow \fH\oplus\fH^\ast, 
\end{equation}
for every $\vect{F} \in \fH\oplus\fH^\ast$. 
Moreover, we have the commutation relations
\begin{align}\label{commutation relations}
	&\com{Z(\vect{F}_1), Z^*(\vect{F}_2)}
	= \inprod{\vect{F}_1}{\sfS \vect{F}_2}_{\fH \oplus \fH^\ast} ,\\
	\intertext{where }
	& \sfS=
	\begin{pmatrix}
		\id & \phantom{-}0 \\
		0 & -\id 
	\end{pmatrix}:\fH\oplus\fH^\ast \rightarrow \fH\oplus\fH^\ast, \notag
\end{align}
for every $\vect{F}_1, \vect{F}_2 \in \fH\oplus \fH^\ast$.
Here, $\id$ is the identity map on $\fH$. 

Bogoliubov transformations are linear mappings from $\fH\oplus\fH^\ast$ to itself such that the adjoint relations~\eqref{adjoint relation} and the commutation relations~\eqref{commutation relations} are preserved. To this end, a linear bounded isomorphism $\Theta:\fH\oplus\fH^\ast\rightarrow \fH\oplus\fH^\ast$ is called a Bogoliubov transformation if it satisfies
\begin{equation}\label{def:properties_of_Bogoliubov_transformation}
	\Theta\sfJ = \sfJ\Theta \quad \text{ and } \quad 
	\sfS=\Theta^* \sfS \Theta.
\end{equation}

\subsubsection{Lie algebra of ``symplectic matrices''} For the construction of Bogoliubov transformation, we start with the Lie algebra  of real or complex
symplectic matrices of the  form
\begin{equation}
	\sfL=
	\begin{pmatrix} 
		\vect{A} & \vect{B}\\
		\vect{C} &-\vect{A}^\top
	\end{pmatrix}
\end{equation}
where $\vect{A}$ is a self-adjoint Hilbert--Schmidt operator, and $\vect{B}$ and $\vect{C}$ are symmetric Hilbert--Schmidt operators. We denote the Lie algebra by $\mathsf{sp}(\RR)$ or $\mathsf{sp}(\CC)$ depending on whether the entries of $\vect{A}$, $\vect{B}$, and $\vect{C}$ are real or complex. 

The natural setting for us is the Lie subalgebra  $\mathsf{sp}_{\CC}(\RR):=\sfW \mathsf{sp}(\RR) \sfW^{-1}\subsetneq \mathsf{sp}(\CC)$ where
\begin{align*}
	\sfW=
	\frac{1}{\sqrt 2}
	\begin{pmatrix}
		\id & \phantom{-}i \id\\
		\id & -i\id
	\end{pmatrix}: L^2(\R^3; \R^2)\oplus L^2(\R^3; \R^2)\rightarrow \fH\oplus \fH^\ast.
\end{align*}
Moreover, observe that the elements of $\mathsf{sp}_{\CC}(\RR)$ have the form
\begin{equation}\label{def:spc}
	\sfM=\begin{pmatrix}
		\boldsymbol{d}& \boldsymbol{k}\\ 
		\conj{\boldsymbol{k}} &-\boldsymbol{d}^\top
	\end{pmatrix},
\end{equation}
where $\boldsymbol{d}$ and $\boldsymbol{k}$ are Hilbert--Schmidt operators with $\boldsymbol{d}$ self-adjoint and $\boldsymbol{k}$ symmetric. In particular, it can be checked that $\Theta:=e^\sfM$ satisfies the properties of \eqref{def:properties_of_Bogoliubov_transformation}.
\subsubsection{Lie algebra isomorphism} We define the mapping from $\sfL \in \mathsf{sp}(\CC)$ to the
multivariate quadratic polynomials in $(a,\, b,\, a^\ast,\, b^\ast)$, denoted by $\mathsf{Quad}$, in the following manner
\begin{multline}\label{liemap}
	\cI\big(\sfL\big) :=\, 
	\frac{1}{2}\iintd 
	((z^{\ast}_{x})^\top,\, z_{x}^\top)
	\begin{pmatrix}
		\vect{A}(x,y)& \phantom{-}\vect{C}(x,y)\\
		\vect{B}(x,y)&-\vect{A}(y,x)
	\end{pmatrix}
	\begin{pmatrix}-z_{y}\\ 
		\phantom{-}z^{\ast}_{y}
	\end{pmatrix}
	\dd x\d y\\
	=\, \frac{1}{2}\iintd 
	z^\top_{x}\vect{A}(x,y)z^{\ast}_{y}+(z_x^\ast)^\top\vect{A}(y,x) z_y-z^\top_{x}\vect{B}(x,y)z_{y}+(z^\ast_{x})^\top\vect{C}(x,y)z^{\ast}_{y}\dd x\d y  .
\end{multline}

This is the infinite-dimensional Segal--Shale--Weil infinitesimal representation. The group representation was studied in \cite{shale1962linear}.
The crucial property of this map is that it is a Lie algebra isomorphism
\begin{equation}\label{Lieisomorph}
	\com{\cI(\sfL_{1}),\cI(\sfL_{2})}=\cI\(\com{\sfL_{1},\sfL_{2}}\) .
\end{equation}
Notice that if $\sfL \in \mathsf{sp}_{\CC}(\RR)$, then $\sfL$ has the form~\eqref{def:spc} and $\cI(\sfL)$ is skew-adjoint. Thus,
$e^{\cI(\sfL)} $ is a unitary operator on Fock space.

As an application to this work, we organize the kernels~\eqref{def:k_i} and \eqref{def:k_ij} in the following symplectic matrix form
\begin{equation}\label{def:bold_k}
	\sfK=  
	\begin{pmatrix}
		\boldsymbol{0} & \boldsymbol{k}\\
		\overline{\boldsymbol{k}} & \boldsymbol{0}
	\end{pmatrix} \quad 
	\text{ with } \quad 
	\boldsymbol{k}=
	\begin{pmatrix}
		k_{1} & k_{12} \\
		k_{21} & k_{2} 
	\end{pmatrix}    .
\end{equation}
With $\sfK$, we can recast \eqref{def: operator B} formally in the form
\begin{equation}\label{eq:defB}
	\cB(\boldsymbol{k}) = -\cI(\sfK) = -\frac 12 \iintd \big( (z^\ast_x)^\top, z_x^\top \big)
	\begin{pmatrix}
		\boldsymbol{0} & \boldsymbol{k}\\
		\overline{\boldsymbol{k}} & \boldsymbol{0}
	\end{pmatrix}
	\begin{pmatrix}
		-z_y\\
		\phantom{-}z^\ast_y
	\end{pmatrix}
	\dd x \d y 
\end{equation}
with the corresponding unitary transformation $\cT(\boldsymbol{k})= e^{\cI(\sfK)}=e^{-\cB(\boldsymbol{k})}$. Here, $\boldsymbol{k}$ is called the pair excitation matrix. 

Let us summarize some useful properties of the Lie algebra isomorphism and the Bogoliubov transformation. 

\begin{prop}\label{prop: bogoliubov transformation}
	Let $\sfM \in \mathsf{sp}_{\CC}(\RR)$ be of the form~\eqref{def:spc} where $\boldsymbol{d}$ and $\boldsymbol{k}$ are Hilbert--Schmidt operators with $\boldsymbol{d}$ self-adjoint and $\boldsymbol{k}$ is symmetric.  Then we have 
	\begin{enumerate}[(i)]
		\item We have the following commutation relation:
		\begin{align}
			\com{\cI(\sfM), Z^\ast(\vect{F})} = Z^\ast(\sfM \vect{F}).
		\end{align}
		Moreover, we also have that for every $\vect{F}\in \fH\oplus\fH^\ast$
		\begin{align}
			e^{\cI(\sfM)}Z^\ast(\vect{F})e^{-\cI(\sfM)} = Z^\ast(e^{\sfM} \vect{F}).
		\end{align}
		\item The mapping $\cI:\mathsf{sp}(\CC)\rightarrow \mathsf{Quad}$ is a Lie algebra isomorphism.    In particular, the mapping $\cI$ restricted to $\mathsf{sp}_{\CC}(\RR)$ is also a Lie algebra isomorphism of $\mathsf{sp}_{\CC}(\RR)$ with a subalgebra of $\mathsf{Quad}$.
		\item The operator $e^{\cI(\sfM)}$ is unitary and we have that 
		\begin{equation*}
			(e^{\cI(\sfM)})^*
			=(e^{\cI(\sfM)})^{-1}
			=e^{-\cI(\sfM)}.
		\end{equation*}
		\item We have 
		\begin{align}\label{eq:time-derivative_lie_identity}
			\(\frac{\bd}{\bd t}e^{\cI(\sfM)}\)e^{-\cI(\sfM)} = \cI\(\(\frac{\bd}{\bd t}e^{\sfM}\)e^{-\sfM}\).
		\end{align}
		\item If $\sfR \in \mathsf{sp}(\CC)$, then we have 
		\begin{align}\label{eq:conjugation_lie_identity}
			e^{\cI(\sfM)}\cI(\sfR)e^{-\cI(\sfM)} = \cI\(e^{\sfM}\sfR e^{-\sfM}\).
		\end{align}
		\item For every $\vect{F}\in \fH\oplus\fH^\ast$, one has
		\begin{equation*}
			\cT(\boldsymbol{k})^* Z(\vect{F}) \cT(\boldsymbol{k})
			=
			Z\left(\Theta \vect{F}\right).
		\end{equation*}
		Here, $\Theta: \fH\oplus\fH^\ast \rightarrow \fH\oplus\fH^\ast$ is the Bogoliubov transformation defined by the matrix
		\begin{equation}\label{def: Theta}
			\Theta=\left(\begin{array}{cc}
				\ch(\boldsymbol{k}) & \sh(\boldsymbol{k}) \\[.5ex]
				\conj{\sh(\boldsymbol{k})} & \conj{\ch(\boldsymbol{k})}
			\end{array}\right)
		\end{equation}
		where $\ch(\boldsymbol{k})$, $\sh(\boldsymbol{k}): \fH\rightarrow \fH$ are defined by 
		\begin{equation*}
			\begin{aligned}
				\ch(\boldsymbol{k})
				=
				\sum_{n =0}^\infty \frac{1}{(2 n)!}(\boldsymbol{k} \overline{\boldsymbol{k}})^n
				\quad\text{and}\quad
				\sh(\boldsymbol{k})
				=
				\sum_{n =0}^\infty \frac{1}{(2 n+1)!}(\boldsymbol{k} \overline{\boldsymbol{k}})^n \boldsymbol{k}.
			\end{aligned}
		\end{equation*}
		The products of $\boldsymbol{k}$ and $\overline{\boldsymbol{k}}$ are understood in the sense of operator compositions.
		
	\end{enumerate}
\end{prop}

\begin{proof}
	The proof is similar to the ones given in \cite[Theorems 4.1--4.2]{grillakis2010second} (cf. also \cite{Benedikter2015quantitative}).
\end{proof}

Let us state some estimates for the pair excitation matrix $\boldsymbol{k}$.

\begin{lem}\label{lem: estimate of k}
	Let $\widetilde{\boldsymbol{\phi}}_t^{(N)}$ be a solution to the \eqref{eq:modified_GP_system_vector_form} with $H^1$ initial data. Let $\boldsymbol{k}$ be of the form defined in \eqref{def:bold_k} with entries given in \eqref{def:k_i} and \eqref{def:k_ij}.
	Then, we have the following:
	\begin{enumerate}[(i)]
		\item For the pair excitation matrix $\boldsymbol{k}_t$, we have the following Hilbert--Schmidt norm bounds
		\begin{equation}
			\begin{aligned}
				\left\|\boldsymbol{k}_t\right\|_{\rm HS}
				& \leq\, 
				C\min\Big(\|\widetilde{\boldsymbol{\phi}}_t^{(N)}
				\|_{L^\infty_x}, 1\Big), \\
				\left\|\nabla_1 \boldsymbol{k}_t\right\|_{\rm HS},\left\|\nabla_2 \boldsymbol{k}_t\right\|_{\rm HS} 
				& \leq\, 
				C\sqrt{N}
				\min\Big(\|\widetilde{\boldsymbol{\phi}}_t^{(N)}
				\|_{L^\infty_x}, 1\Big),\\
				\|\nabla_1 (\boldsymbol{k}_t \overline{\boldsymbol{k}}_t ) \|_{\rm HS},
				\|\nabla_2(\boldsymbol{k}_t \overline{\boldsymbol{k}}_t) \|_{\rm HS} 
				& \leq\, 
				C\min\Big(\|\widetilde{\boldsymbol{\phi}}_t^{(N)}
				\|_{L^\infty_x}^2, 1\Big).
			\end{aligned}
		\end{equation}
		Here, the constant $C$ depends on $\|\widetilde{\boldsymbol{\phi}}_t^{(N)} \|_{H^1}$ and  we use the notation $\nabla_1$, $\nabla_2$ to denote the gradient operator $\nabla_x$, $\nabla_y$ for the kernel $k_{\ii\jj}(x,y)$.
		
		Moreover, consider the decomposition
		\begin{equation*}
			\ch(\boldsymbol{k})
			=
			\id + \operatorname{p}(\boldsymbol{k})
			\quad \text{and}\quad
			\sh(\boldsymbol{k})
			=
			\boldsymbol{k} + \operatorname{r}(\boldsymbol{k}),
		\end{equation*}
		where $\id$ denotes the identity operator on $\fH$. Then we have the following bounds 
		\begin{equation}
			\begin{aligned}
				\|\operatorname{p}(\boldsymbol{k}_t) \|_{\rm HS}\,,
				\;
				\|\operatorname{r}(\boldsymbol{k}_t) \|_{\rm HS} 
				\leq&\, C \min\Big( \|\widetilde{\boldsymbol{\phi}}_t^{(N)}
				\|_{L^\infty_x}^2, 1\Big)\,, \\ 
				\|\nabla_\alpha \operatorname{p}(\boldsymbol{k}_t) \|_{\rm HS}\,,
				\;
				\|\nabla_\alpha \operatorname{r}(\boldsymbol{k}_t) \|_{\rm HS} 
				\leq&\,
				C\min\Big( \|\widetilde{\boldsymbol{\phi}}_t^{(N)}
				\|_{L^\infty_x}^2, 1\Big)\,,
			\end{aligned}
		\end{equation}
		for $\alpha = 1, 2$, where 
		\begin{align*}
			\Nrm{\grad_1\operatorname{p}(\boldsymbol{k})}{\rm HS}^2=\sum^2_{\ii=1}\sum^2_{\jj=1}\iintd 
			\n{\grad_x\operatorname{p}(\boldsymbol{k})_{\ii\jj}(x, y)}^2\dd x \d y,
		\end{align*}
		and similarly for $\Nrm{\grad_2\operatorname{p}(\boldsymbol{k})}{\rm HS}^2$.
		\item For every $x, y \in \R^3$, we also have the following pointwise estimates 
		\begin{equation}\label{pointwise bound}
			\begin{aligned}
				\n{\boldsymbol{k}_{t}(x, y)}_{\rm F} \le&\, \frac{C}{\n{x-y}+N^{-1}}|\widetilde{\boldsymbol{\phi}}_t^{(N)}(x)|
				|\widetilde{\boldsymbol{\phi}}_t^{(N)}(y)|\,,\\
				\n{\operatorname{p}(\boldsymbol{k}_{t})(x, y)}_{\rm F} \le&\, C|\widetilde{\boldsymbol{\phi}}_t^{(N)}(x)|
				|\widetilde{\boldsymbol{\phi}}_t^{(N)}(y)|\,,\\
				\n{\operatorname{r}(\boldsymbol{k}_{t})(x, y)}_{\rm F} \le&\, C|\widetilde{\boldsymbol{\phi}}_t^{(N)}(x)|
				|\widetilde{\boldsymbol{\phi}}_t^{(N)}(y)|\,.
			\end{aligned}
		\end{equation}
		\item     Moreover, we have
		\begin{equation}
			\begin{aligned}
				\sup _{x \in \R^3}\left\|\boldsymbol{k}_t(x, \cdot)\right\|_{L^2}
				\le&\,  C\|\widetilde{\boldsymbol{\phi}}_t^{(N)}
				\|_{L^\infty_x}\,,\\
				\sup _{x \in \mathbb{R}^3}\left\|\operatorname{r}(\boldsymbol{k}_t)(x,\cdot)\right\|_{L^2}, \;             \sup _{x \in \mathbb{R}^3}\left\|\operatorname{p}(\boldsymbol{k}_t)(x,\cdot) \right\|_{L^2}
				\le&\,  C\|\widetilde{\boldsymbol{\phi}}_t^{(N)}
				\|_{L^\infty_x}\, .
			\end{aligned}
		\end{equation}
	\end{enumerate}
\end{lem}

\begin{lem}\label{bound of derivative}
	Let $\widetilde{\boldsymbol{\phi}}_t^{(N)}$ be a solution to the \eqref{eq:modified_GP_system_vector_form} with $\wtbphi^{(N)}_0 \in H^4$. Then we have the following estimates
	\begin{equation}
		\begin{aligned}
			\|\dot{\boldsymbol{k}}_t\|_{\rm HS}
			\le&\,
			C\min\Big(\| \widetilde{\boldsymbol{\phi}}_t^{(N)}
			\|_{L^\infty_x}, 1\Big)\, ,\\
			\|\ddot{\boldsymbol{k}}_t \|_{\rm HS}  
			\leq&\, 
			C\min\Big(\|\partial_t \widetilde{\boldsymbol{\phi}}_t^{(N)}
			\|_{L^\infty_x}
			+ \|\widetilde{\boldsymbol{\phi}}_t^{(N)}
			\|_{L^\infty_x}, 1\Big) \, ,\\
			\|\nabla_{\alpha}(\dot{\boldsymbol{k}} \overline{\boldsymbol{k}})\|_{\rm HS},
			\;
			\|\nabla_{\alpha}(\boldsymbol{k} \dot{\overline{\boldsymbol{k}}})\|_{\rm HS} 
			\leq&\, 
			C\min\Big(\|\partial_t \widetilde{\boldsymbol{\phi}}_t^{(N)}
			\|_{L^\infty_x}^2, 1\Big) \, .
		\end{aligned}
	\end{equation}
	Here, the constant $C$ depends on $\|\widetilde{\boldsymbol{\phi}}_t^{(N)}
	\|_{H^1_x}$, $\|\partial_t \widetilde{\boldsymbol{\phi}}_t^{(N)}
	\|_{H^1_x}$ and $\|\partial^2_t \widetilde{\boldsymbol{\phi}}_t^{(N)}
	\|_{L^2_x}$.
	Furthermore, we also have
	\begin{equation}
		\begin{aligned}
			\| \dot{\operatorname{p}}(\boldsymbol{k}) \|_{\rm HS},
			\;
			\| \dot{\operatorname{r}}(\boldsymbol{k})\|_{\rm HS} 
			\le &\,
			C\min\Big(\|\partial_t \widetilde{\boldsymbol{\phi}}_t^{(N)}
			\|_{L^\infty_x}, 1\Big)\, ,\\
			\|\nabla_{\alpha} \dot{\operatorname{p}}(\boldsymbol{k}) \|_{\rm HS},
			\;
			\|\nabla_{\alpha} \dot{\operatorname{r}}(\boldsymbol{k})\|_{\rm HS} 
			\leq &\, 
			C\min\Big(\|\partial_t \widetilde{\boldsymbol{\phi}}_t^{(N)}
			\|_{L^\infty_x}^2, 1\Big)\, ,
		\end{aligned}
	\end{equation}
	and similarly,
	\begin{equation}
		\begin{aligned}
			\sup _{x \in \R^3} \|\dot{\boldsymbol{k}}_t (x, \cdot) \|_{L^2_y}
			\le\,  C(\|\widetilde{\boldsymbol{\phi}}_t^{(N)}
			\|_{L^\infty_x}
			+ \|\partial_t\widetilde{\boldsymbol{\phi}}_t^{(N)}
			\|_{L^\infty_x})\,,\\
			\sup _{x \in \mathbb{R}^3}\left\|\dot{\operatorname{r}}(\boldsymbol{k}_t)(x,\cdot)\right\|_{L^2_y}, \;             \sup _{x \in \mathbb{R}^3}\left\|\dot{\operatorname{p}}(\boldsymbol{k}_t)(x,\cdot) \right\|_{L^2_y}
			\le\,  
			C(\|\widetilde{\boldsymbol{\phi}}_t^{(N)}
			\|_{L^\infty_x}
			+ \|\partial_t\widetilde{\boldsymbol{\phi}}_t^{(N)}
			\|_{L^\infty_x})\, .
		\end{aligned}
	\end{equation}
\end{lem}

\begin{proof}[Proof of Lemma~\ref{lem: estimate of k} and Lemma~\ref{bound of derivative}]
	The proofs are minor modifications of the ones given in \cite[Appendix B]{Benedikter2015quantitative}.
\end{proof}

\subsection{Growth of Fluctuations} 
The proofs of the results in this section rely heavily on the result of Section~\ref{sect:bounds_fluctuation_dynamics}.
\begin{prop}\label{prop:estimate of fluctuation}
	There exists a time-dependent constant $C_N(t)$, such that
	\begin{equation*}
		\widetilde{\cL}_N(t):=\cL_N(t)-C_N(t)
	\end{equation*}
	satisfies the estimates:
	\begin{align}\label{lower bound of L}
		\widetilde{\mathcal{L}}_N(t) \geq&\, \tfrac{1}{2} \mathcal{H}_N
		-
		C \lambda \upeta(t)\Big(\frac{\mathcal{N}^2}{N}+\mathcal{N}+1\Big), \\
		\widetilde{\mathcal{L}}_N(t) \leq&\,  \tfrac{3}{2} \mathcal{H}_N
		+
		C \lambda \upeta(t)\Big(\frac{\mathcal{N}^2}{N}+\mathcal{N}+1\Big).
	\end{align}
	
	Moreover, we also have
	\begin{equation}\label{bound of commutator}
		\pm\big[\mathcal{N}, \widetilde{\mathcal{L}}_N(t)\big] \leq \upeta(t) \cH_N
		+
		C \lambda \upeta(t)\Big(\frac{\mathcal{N}^2}{N}+\mathcal{N}+1\Big),
	\end{equation}
	and
	\begin{equation}\label{esitmate of L dot}
		\pm \dot{\widetilde{\mathcal{L}}}_N(t) 
		\leq 
		\upeta(t) \mathcal{H}_N
		+
		C 
        \lambda \upeta(t)\Big(\frac{\mathcal{N}^2}{N}+\mathcal{N}+1\Big).
	\end{equation}
	Here, $\upeta(t)$ is given by
	\begin{equation}\label{def:lambda}
		\upeta(t) := \|\widetilde{\boldsymbol{\phi}}_t^{(N)}
		\|_{W^{1,\infty}_x} 
		+
		\|\partial_t\widetilde{\boldsymbol{\phi}}_t^{(N)}
		\|_{L^\infty_x}.
	\end{equation}
	The constant $C$ depends on $\|\widetilde{\boldsymbol{\phi}}_t^{(N)}
	\|_{H^2_x}$, $\|\partial_t \widetilde{\boldsymbol{\phi}}_t^{(N)}
	\|_{H^2_x}$, and $\|\partial^2_t \widetilde{\boldsymbol{\phi}}_t^{(N)}
	\|_{L^2_x}$.
\end{prop}

\begin{proof}
	We write the generator of the fluctuation dynamics in the form
	\begin{equation}\label{Generator}
		\cL_N(t) 
		= 
		N\mu_0
		+
		\cT_t^\ast
		\cH_N
		\cT_t
		+
		\cT_t^\ast
		\big(\cG_{1}+\cG_{2}+\cG_{3}\big)
		\cT_t
		+
		\big(i \bd_t \cT_t^*\big) 
		\cT_t.
	\end{equation}
	The terms $\cT_t^\ast
	\big(\cG_{1}+\cG_{2}+\cG_{3}\big)
	\cT_t$ are estimated in Propositions~\ref{prop: estimate of quadratic term} and \ref{prop: T G1 G2 T}. The term $\cT_t^\ast \cH_N \cT_t$ is estimated in Propositions \ref{prop: Kinetic term} and \ref{prop: quartic term}. Moreover, there are cancellation between the Hamiltonian term and the quadratic term, as stated in Proposition \ref{prop: cancellation}, i.e., $\cT_t^\ast \(\cH_N+\cG_{2}\) \cT_t$. 
	Finally, the term $\big(i \bd_t \cT_t^*\big) 
	\cT_t$ is estimated in Proposition \ref{prop: time derivative}. In short, combining \eqref{Generator} and the results in Propositions~\ref{prop: estimate of quadratic term}, \ref{prop: T G1 G2 T}, \ref{prop: Kinetic term}, \ref{prop: quartic term}, \ref{prop: cancellation}, and \ref{prop: time derivative}, it follows that
	\begin{equation*}
		\begin{aligned}
			\cL_N(t) &= C_N(t) + \cH_N + \cE(t),\\
			\big[\cN, \cL_N(t)\big]
			&=
			\big[\mathcal{N}, \mathcal{E}(t)\big],\\
			\dot{\mathcal{L}}_N(t)
			&=
			\dot{C}_N(t) + \dot{\mathcal{E}}(t).
		\end{aligned}
	\end{equation*}
	Here, $C_N(t) = N\mu_0 +C_{N, \chi} +C_{N, K}^{(1)}(t)+C_{N,K}^{(2)}(t)+ C_{N, 2}(t)+ C_{N, V}(t)$ with $C_{N, \sharp}(t):=C^{(1)}_{N, \sharp}(t)+C^{(2)}_{N, \sharp}(t)+C^{(12)}_{N, \sharp}(t)$ arising from Propositions~\ref{prop: estimate of quadratic term}, \ref{prop: Kinetic term}, and \ref{prop: quartic term} (see \eqref{eq:C_N term}). Moreover, $\mathcal{E}(t)$ satisfies estimates
	\begin{equation*}
		\begin{aligned}
			\pm \mathcal{E}(t)
			\le&\,
			\upeta(t) \delta  \mathcal{H}_N
			+
			C_\delta 
            \lambda \upeta(t) \Big(\frac{\mathcal{N}^2}{N}+\mathcal{N}+1\Big),
		\end{aligned}
	\end{equation*}
	and $\pm \big[\mathcal{N}, \mathcal{E}(t)\big]$, $\pm\dot{\mathcal{E}}(t)$ yield the same bound.
	
	Using the embedding inequality, one has $\upeta(t)\le C$, where $C$ depends on $\|\widetilde{\boldsymbol{\phi}}_t^{(N)}
	\|_{H^2_x}$ and $\|\partial_t \widetilde{\boldsymbol{\phi}}_t^{(N)}
	\|_{H^2_x}$. Choosing $\delta$ small enough, the desired estimates can then be derived.
\end{proof}

The following proposition gives the growth of the fluctuation:
\begin{prop}\label{prop: growth of fluctuation}
	Suppose $\Xi_N \in \cF$ such that
	\begin{equation*}
		\inprod{\Xi_N}{\(\frac{\mathcal{N}^2}{N}+\cN+\cH_N\) \Xi_N} \leq C\,,
	\end{equation*}
	then there exists a constant $C$ independent of $t$, $\lambda$ and $N$, such that
	\begin{equation*}
		\inprod{\Xi_N}{\cU_N^*(t ; 0)\, \cN\, \cU_N(t ; 0) \Xi_N} \leq e^{C\lambda}\,.
	\end{equation*}
\end{prop}

To prove the proposition, we need the following lemma.
\begin{lem}\label{lem:conjugation of number by Bogoliubov transformation}
	Let $\boldsymbol{k}_t$ be of the form defined in \eqref{def:bold_k} with entries given in \eqref{def:k_i} and \eqref{def:k_ij}, and $\cT_t = e^{-\cB(\boldsymbol{k}_t)}$ be defined in \eqref{eq:defB}. Then there exists a constant $C$, dependent only on $\norm{\boldsymbol{k}_t}_{\rm HS}$, such that
	\begin{equation}\label{first assertion}
		\begin{aligned}
			\cT_t^* \cN \cT_t & \leq C(\cN+1), \quad\text{and}\quad
			\cT_t^* \cN^2 \cT_t & \leq C(\cN+1)^2.
		\end{aligned}
	\end{equation}
\end{lem}

\begin{proof}
	To prove the first inequality in \eqref{first assertion}, we use the decomposition \eqref{eq:creation_annihilation_expansion} and \eqref{est: A, B, alpha, beta} in Lemma~\ref{lem: quadratic}. Then, it follows that
	\begin{equation*}
		\begin{aligned}
			\left\langle\Psi, \cT_t^* \mathcal{N} \cT_t \Psi\right\rangle 
			& =
			\intd \left\|
			\left(a_x+ A^*_x + \alpha_x \right) \Psi\right\|^2\dd x
			+
			\intd \left\|
			\left(b_x+ B^*_x + \beta_x \right) \Psi\right\|^2\dd x\\
			& \leq 
			C\left(1 + \| \boldsymbol{k}_t\|_{\rm HS}^2 
			\right)
			\big\|(\mathcal{N}+1)^{\frac12} \Psi\big\|^2.
		\end{aligned}
	\end{equation*}
	For the second inequality in \eqref{first assertion}, we observe that
	\begin{equation*}
		\begin{aligned}
			\mathcal{N}^2 
			= &
			\intd  a_x^*\, \mathcal{N}\, a_x \dd x
			+ \intd b_x^*\, \mathcal{N}\, b_x \dd x
			+ \mathcal{N} .
		\end{aligned}
	\end{equation*}
	Therefore, we have
	\begin{equation*}
		\begin{aligned}
			\big\langle\Psi,\, \cT_t^* \mathcal{N}^2 \cT_t \Psi\big\rangle 
			&=
			\intd \d x\,\big\langle (a_x+ A_x^* + \alpha_x)\Psi,\, \cT_t^* \mathcal{N} \cT_t (a_x+ A^*_x + \alpha_x) \Psi\big\rangle \\
			&\phantom{=} +
			\intd \d x\,\big\langle (b_x+ B_x^* + \beta_x)\Psi,\, \cT_t^* \mathcal{N} \cT_t (b_x+ B^*_x + \beta_x) \Psi\big\rangle 
			+
			\big\langle\Psi, \cT_t^* \mathcal{N} \cT_t \Psi\big\rangle .
		\end{aligned}
	\end{equation*}
	Then, by the above inequality and Lemma~\ref{lem: quadratic}, we get the desired result.
\end{proof}

\begin{proof}[Proof of Proposition~\ref{prop: growth of fluctuation}]
	As a corollary of Lemma~\ref{lem:conjugation of number by Bogoliubov transformation}, one can follow the proof of \cite[Proposition 4.2]{Benedikter2015quantitative} to deduce the bound
	\begin{align}\label{second assertion}
		\cU_N^*(t ; 0) \mathcal{N}^2 \cU_N(t ; 0) 
		\leq 
		C N \cU_N^*(t ; 0) \cN \cU_N(t ; 0)
		+
		C \left( N(\mathcal{N}+1)+(\mathcal{N}+1)^2\right).
	\end{align}
	Now, let us use the notation in Proposition~\ref{prop:estimate of fluctuation}, and denote $\widetilde{\cU}_N(t ; s)=e^{i \int_s^t C_N(\tau) \dd \tau} \cU_N(t ; s)$,
	then we have
	\begin{equation}\label{def: tilde cL}
		i \bd_t \widetilde{\cU}_N(t ; s)
		=
		\widetilde{\cL}_N(t) \widetilde{\cU}_N(t ; s) \quad \text { with } \quad 
		\widetilde{\cU}_N(s ; s) = 1.
	\end{equation}
	Note also that substituting \eqref{lower bound of L} with \eqref{bound of commutator}, and using the fact $\upeta(t)\le C$,
	one can derive
	\begin{equation*}
		\begin{aligned}
			\pm\big[\cN, \widetilde{\cL}_N(t)\big] 
			\le&\,
			\upeta(t)\Big( 2\widetilde{\cL}_N(t)
			+
			C \lambda
            \Big(\frac{\cN^2}{N}+\mathcal{N}+1\Big)\Big).
		\end{aligned}
	\end{equation*}
	Using the estimate above and \eqref{second assertion}, it follows that
	\begin{multline}\label{evolution of N}
		\frac{\d}{\d t}\Big\langle\Xi_N, \widetilde{\cU}_N^\ast(t ; 0) \cN \widetilde{\cU}_N(t ; 0) \Xi_N\Big\rangle \\
		\leq
		\upeta(t) \left\langle\Xi_N, \widetilde{\cU}_N^\ast(t ; 0)\left(2 \widetilde{\cL}_N(t)
		+
		C_1 \lambda \cN \right) \widetilde{\cU}_N(t ; 0) \Xi_N\right\rangle  + \widetilde{C}_1\lambda \upeta(t).
	\end{multline}
	Here, $\widetilde{C}_1$ relies on $\langle\Xi_N,(\mathcal{N}^2/ N + \mathcal{N} + 1 ) \Xi_N\rangle$. 
	Similarly, we substitute \eqref{lower bound of L} with \eqref{esitmate of L dot}, and utilize \eqref{second assertion}. This leads to
	\begin{multline}\label{evolution of L}
		\frac{\d}{\d t}\Big\langle\Xi_N, \widetilde{\cU}_N^\ast (t ; 0) \widetilde{\mathcal{L}}_N(t) 
		\widetilde{\cU}_N(t ; 0) \Xi_N\Big\rangle 
		=
		\left\langle\Xi_N, \widetilde{\cU}_N^\ast (t ; 0) \dot{\widetilde{\mathcal{L}}}_N(t) \widetilde{\cU}_N(t ; 0) \Xi_N\right\rangle \\
		\leq
		\upeta(t) \left\langle\Xi_N, \widetilde{\cU}_N^\ast (t ; 0)\left(2 \widetilde{\mathcal{L}}_N(t)
		+
		C_2 \lambda \mathcal{N} \right) \widetilde{\cU}_N(t ; 0) \Xi_N\right\rangle  + \widetilde{C}_2 \lambda\upeta(t).
	\end{multline}
    Note that $\cN\ge0$, together with the following non-negativity by substituting \eqref{second assertion} with \eqref{lower bound of L}:
	\begin{equation} 
    \label{estimate of L + N}
    \begin{aligned}
		0 &\leq 
		\left\langle\Xi_N, \widetilde{\cU}_N^\ast (t ; 0) \mathcal{H}_N \widetilde{\cU}_N(t ; 0) \Xi_N\right\rangle\\
		&\leq
		\Big\langle\Xi_N, \widetilde{\cU}_N^\ast (t ; 0)\Big(2 \widetilde{\mathcal{L}}_N(t)
		+
		C_3\lambda
        \mathcal{N} \Big) \widetilde{\cU}_N(t ; 0) \Xi_N\Big\rangle
		+ \widetilde{C}_3\lambda.
    \end{aligned}
	\end{equation}
    Let $D=\max\{C_1,C_2,C_3+1\}$ and compute the inequality $\lambda D \times \eqref{evolution of N} + 2\times \eqref{evolution of L}$. Using the non-negativity above and applying the Gr\"{o}nwall inequality together with Corollary~\ref{cor:Linfty-t-indep-bdd}, i.e., $\int_0^\infty \upeta(t)\d t\le C$, then one has that there exists $C>0$, independent of $t$, $\lambda$ and $N$, such that
	\begin{equation*}
		\begin{aligned}
			\Big\langle\Xi_N, \widetilde{\cU}_N^\ast (t ; 0) 
			\left( 2\widetilde{\mathcal{L}}_N(t) +
			\lambda D \mathcal{N} \right) 
			\widetilde{\cU}_N(t ; 0) \Xi_N\Big\rangle  
			\leq e^{C\lambda}.
		\end{aligned}
	\end{equation*}
	Combining it with \eqref{estimate of L + N} and using the fact $D\ge C_3+1$, we finally prove the theorem.  
\end{proof}

\subsection{Proof of Main Theorem}

\begin{proof}[Proof of Theorem~\ref{thm:main}]
	Let $\Psi_{N,t}$ be the many-body Fock state defined in \eqref{eq:full-fock}, and $\Gamma_{N,t}^{(1)}$ be the associated two-component one-particle reduced density. 
	
	We use the fact that $\cN$ commutes with $\cH_N$ to get
	\begin{equation*}
		\begin{aligned}
			\big\langle\Psi_{N, t},\, \cN \Psi_{N, t}\big\rangle 
			& =
			\Big\langle \Xi_{N},\, \cT^\ast_0 
			\cW^\ast(\sqrt{N} \wtbphi^{(N)}_0) 
			\cN \cW(\sqrt{N} \wtbphi^{(N)}_0) \cT_0 \Xi_{N} \Big\rangle .
		\end{aligned}
	\end{equation*}
	Then, one can derive from \eqref{eq:two-component_weyl_conjugation} that
	\begin{equation*}
		\begin{aligned}
			\cW(\sqrt{N} \wtbphi^{(N)}_0)^\ast \cN \cW(\sqrt{N} \wtbphi^{(N)}_0) 
			&=
			\cN - \sqrt{N} \(z(\wtbphi^{(N)}_0)+z^\ast(\wtbphi^{(N)}_0)\)
			+ N .
		\end{aligned}
	\end{equation*}
	This substitution leads us to
	\begin{equation*}
		\begin{aligned}
			\big\langle\Psi_{N, t},\, \mathcal{N} \Psi_{N, t}\big\rangle 
			& =
			N+ \Big\langle\Xi_{N},\, \cT^\ast_0 
			\big(\cN - \sqrt{N} \(z(\wtbphi^{(N)}_0)+z^\ast(\wtbphi^{(N)}_0)\)
			\big) 
			\cT_0 \Xi_{N}\Big\rangle .
		\end{aligned}
	\end{equation*}
	By \eqref{first assertion}, we have
	\begin{equation*}
		\begin{aligned}
			\big\langle\Xi_{N},\, \cT^\ast_0
			\mathcal{N} \cT_0 \Xi_{N} \big\rangle \leq C \quad\text{and}\quad
			\big\langle\Xi_{N},\, \cT^\ast_0
			\(z(\wtbphi^{(N)}_0)+z^\ast(\wtbphi^{(N)}_0)\) \cT_0 \Xi_{N} \big\rangle 
			\le C.
		\end{aligned}
	\end{equation*}
	Therefore, we obtain
	\begin{equation}\label{estimate of number operator}
		\big| \big\langle\Psi_{N, t},\, \mathcal{N} \Psi_{N, t}\big\rangle
		- N \big |
		\le C \sqrt{N}.
	\end{equation}
	On the other hand, by \eqref{eq:full-fock} and \eqref{eq:two-component_weyl_conjugation}, we have that 
	\begin{multline*}
		\inprod{\Psi_{N, t}}{z^\ast_y z_x^\top\,\Psi_{N, t}}
		= N\conj{\wtbphi^{(N)}_t(y)}\wtbphi^{(N)}_t(x)^\top+\inprod{\cT_t \Xi_{N, t}}{  z^\ast_yz_x^\top\, \cT_t \Xi_{N, t}}\\
		+\sqrt{N}\inprod{\cT_t \Xi_{N, t}}{ z^\ast_y \, \cT_t \Xi_{N, t}}\wtbphi^{(N)}_t(x)^\top+\sqrt{N}\conj{\wtbphi^{(N)}_t(y)}\inprod{\cT_t \Xi_{N, t}}{  z_x^\top\, \cT_t \Xi_{N, t}}.
	\end{multline*}
	Using \eqref{estimate of number operator} and the Cauchy--Schwarz inequality (cf. \cite[Section 5]{Benedikter2015quantitative}), it follows that
	\begin{equation*}
		\Big\| \Gamma_{N,t}^{(1)} - |\wtbphi^{(N)}_t\rangle\!\langle\wtbphi^{(N)}_t | \Big\|_{\mathrm{HS}} 
		\leq 
		\frac{C}{\sqrt{N}} \|(\cN+1)^{\frac{1}{2}} \cT_t \Xi_{N, t} \|^2 .
	\end{equation*}
	Applying \eqref{first assertion}, we obtain
	\begin{equation*}
		\Big\| \Gamma_{N,t}^{(1)} - |\wtbphi^{(N)}_t\rangle\!\langle\wtbphi^{(N)}_t | \Big\|_{\mathrm{HS}} 
		\leq 
		\frac{C}{\sqrt{N}} \|(\cN+1)^{\frac{1}{2}} \Xi_{N, t} \|^2.
	\end{equation*}
	Using the fact that $|\wtbphi^{(N)}_t\rangle\!\langle\wtbphi^{(N)}_t |$ is a rank-one projection and $\Gamma^{(1)}_{N, t}$ is a nonnegative trace-class operator, then their difference $\Gamma_{N,t}^{(1)} - |\wtbphi^{(N)}_t\rangle\!\langle\wtbphi^{(N)}_t |$ has at most one negative eigenvalue. Notice the trace of $\Gamma_{N,t}^{(1)} - |\wtbphi^{(N)}_t\rangle\!\langle\wtbphi^{(N)}_t |$ vanishes since $\wtbphi^{(N)}_t$ is normalized, then it follows that the difference must have one negative eigenvalue, with absolute value equal to the sum of all positive eigenvalues. As a consequence, the trace norm of the difference is controlled by the
	operator norm (given by the absolute value of the negative eigenvalue) and therefore also by
	the Hilbert--Schmidt norm, i.e. 
	\begin{align*}
		\tr_{\fH}\n{\Gamma_{N,t}^{(1)} - |\wtbphi^{(N)}_t\rangle\!\langle\wtbphi^{(N)}_t |}\le \frac{C}{\sqrt{N}} \|(\cN+1)^{\frac{1}{2}} \Xi_{N, t} \|^2.
	\end{align*}
	Finally, by Proposition
	~\ref{prop:uniform_Morawetz}, we have that 
	\begin{align*}
		\tr_{\fH}\n{ |\wtbphi^{(N)}_t\rangle\!\langle\wtbphi^{(N)}_t |-|\bphi_t\rangle\!\langle \bphi_t |} \le 2\norm{\wtbphi^{(N)}_t-\bphi_t}_{H^1_x} \lesssim \frac{1}{N}.
	\end{align*}
	Combining two inequalities above and applying triangle inequality, together with Proposition~\ref{prop: growth of fluctuation}, we complete the proof of the theorem.
\end{proof}

\section{Bounds on the Generator of the Fluctuation Dynamics}\label{sect:bounds_fluctuation_dynamics}
In this section, we estimate the generator of the fluctuation dynamics given in \eqref{Generator}.
From Lemma~\ref{prop: bogoliubov transformation} that we have the identity 
\begin{equation}\label{eq:bogoliubov_conjugation_identities}
	\begin{aligned}
		\cT_t^* z^\ast(\vect{f}) \cT_t=&\,  z^\ast(\ch(\boldsymbol{k}_t)\vect{f})+z(\sh(\boldsymbol{k}_t)\conj{\vect{f}}),\\
		\cT_t^* z(\vect{g}) \cT_t=&\,z(\ch(\boldsymbol{k}_t)\vect{g})+z^\ast(\sh(\boldsymbol{k}_t)\conj{\vect{g}}),
	\end{aligned}
\end{equation}
for every $\vect{f}, \vect{g}\in \fH$, or equivalent, we could express the relations in terms of operator-valued distributions as follow

\begin{equation}\label{transform of z}
	\cT_t^* \big((z^\ast_x)^\top,\; z_x^\top\big) \cT_t
	= \intd 
	\big((z^\ast_y)^\top,\; z_y^\top\big)
	\left(\begin{array}{cc}
		\ch(\boldsymbol{k}_t)_x (y) 
		& \sh(\boldsymbol{k}_t)_x (y) \\[.5ex]
		\conj{\sh(\boldsymbol{k}_t)}_x(y)& \conj{\ch(\boldsymbol{k}_t)}_x (y)
	\end{array}\right)
	\d y
\end{equation}
where $\ch(\boldsymbol{k}_t)_x (y) = \ch(\boldsymbol{k}_t)(y, x)$. 
Now, we make the following decomposition
\begin{equation*}
	\left(\begin{array}{cc}
		\ch
		(\boldsymbol{k}_t) & \sh(\boldsymbol{k}_t) \\[.5ex]
		\conj{\sh(\boldsymbol{k}_t)} & \conj{\ch(\boldsymbol{k}_t)}
	\end{array}\right)
	=
	\left(\begin{array}{cc}
		\id &  0 \\[.5ex]
		0 & \id
	\end{array}\right)
	+
	\left(\begin{array}{cc}
		0 & \boldsymbol{k}_t \\[.5ex]
		\conj{\boldsymbol{k}_t} & 0
	\end{array}\right)
	+
	\left(\begin{array}{cc}
		\operatorname{p}(\boldsymbol{k}_t) & \operatorname{r}(\boldsymbol{k}_t) \\[.5ex]
		\conj{\operatorname{r}(\boldsymbol{k}_t)} & \conj{\operatorname{p}(\boldsymbol{k}_t)}
	\end{array}\right).
\end{equation*}
Moreover, we define the notation
\begin{equation}\label{def: A and B}
	\begin{aligned}
		(A_x,\; B_x )
		:=&\, \intd z_y^\top \conj{\boldsymbol{k}_t}(y,x)\dd y
		=\,
		\(a(k_{1, x}) + b(k_{21, x}),\;
		a(k_{12, x})+ b(k_{2, x})\),
	\end{aligned}
\end{equation}
where $k_{\ii\jj, x}(y) :=k_{\ii\jj}(y, x)$, and
\begin{equation}\label{def: alpha and beta}
	\begin{aligned}
		(\alpha_x,\, \beta_x )
		:=&\, 
		\intd
		\( (z^\ast_y)^\top,\; z_y^\top \)
		\begin{pmatrix}
			\operatorname{r}(\boldsymbol{k}_t)  (y,x)\\[.5ex]
			\conj{\operatorname{p}(\boldsymbol{k}_t)} (y,x) 
		\end{pmatrix}
		\d y
		=\, 
		z^\ast(\operatorname{r}(\boldsymbol{k}_t)_x)+z(\operatorname{p}(\boldsymbol{k}_t)_x).
	\end{aligned}
\end{equation}
With these notations, we now write \eqref{transform of z} as
\begin{equation}\label{eq:creation_annihilation_expansion}
	\begin{aligned}
		\cT_t^* z_x^\top \cT_t
		=&\, 
		\cT_t^* (  a_x,\, b_x ) \cT_t
		=
		(  a_x,\, b_x )
		+
		(  A_x^\ast,\, B_x^\ast )
		+
		(  \alpha_x,\, \beta_x ),\\
		\cT_t^* (z_x^\ast)^\top \cT_t
		=&\, 
		\cT_t^* (  a_x^\ast,\; b_x^\ast ) \cT_t
		=
		(  a_x^\ast,\; b_x^\ast )
		+
		(  A_x,\; B_x )
		+
		(  \alpha_x^\ast,\, \beta_x^\ast ).
	\end{aligned}
\end{equation}

\subsection{Estimates for the Linear, Quadratic and Cubic Terms}
In this section, we provide the estimates of the linear, quadratic, and cubic parts of generator \eqref{Generator}.

\noindent\textbf{(i) Linear terms}.
First, we compute the linear terms. Since $\cG_{1}$ is defined by \eqref{def_G1}, then by \eqref{eq:creation_annihilation_expansion}, we have
\begin{equation}\label{compute TG^1T}
	\begin{aligned}
		\cT_t^*\, \cG_{1}\, \cT_t
		&=\, 
		\sqrt{N}\intd\,
		\Big((N^{3}V^\lambda_1(N\;\cdot\;)w_{1, \ell}(N\;\cdot\;)\ast |\wt\phi^{(N)}_{1, t}|^2)\wt\phi^{(N)}_{1, t}(x) \\
		& \qquad 
		+(N^{3}V^\lambda_{12}(N\;\cdot\;)w_{12, \ell}(N\;\cdot\;)\ast |\wt\phi^{(N)}_{2, t}|^2)\wt\phi^{(N)}_{1, t} (x) \Big)
		\big( a^{\ast}_x + A_x + \alpha^{\ast}_x \big)\dd x \\
		&\phantom{=}\, + 
		\sqrt{N}\intd\,
		\Big((N^{3}V^\lambda_2(N\;\cdot\;)w_{2, \ell}(N\;\cdot\;)\ast |\wt\phi^{(N)}_{2, t}|^2) \wt\phi^{(N)}_{2, t}(x) \\
		& \qquad 
		+(N^{3}V^\lambda_{12}(N\;\cdot\;)w_{12, \ell}(N\;\cdot\;)\ast |\wt\phi^{(N)}_{1, t}|^2)\wt\phi^{(N)}_{2, t} (x) \Big)
		\big( b^{\ast}_x + B_x + \beta^{\ast}_x \big)\dd x + \mathrm{h.c.}
	\end{aligned}
\end{equation}

\noindent\textbf{(ii) Quadratic terms}.
Recall that $\cG_{2} = \cG^{(1)}_2 + \cG^{(2)}_2 + \cG^{(12)}_{2}$ are defined by \eqref{def:intra_quadratic} and \eqref{def: inter in quadratic}. The intra-species parts read:
\begin{equation*}
	\begin{aligned}
		\cT_t^*\, & \cG^{(1)}_2\, \cT_t\\
		=& \intd N^{3}V^\lambda_1(N\;\cdot\;)\ast 
		|\wt\phi^{(N)}_{1, t}|^2(x) 
		\big(a^{\ast}_x + A_x + \alpha_x\big) 
		\big(a_x + A^{\ast}_x + \alpha^{\ast}_x\big)
		\dd x\\
		&+\iintd N^{3}V^\lambda_{1}(N(x-y)) 
		\wt\phi^{(N)}_{1, t}(x)  
		\conj{\wt\phi^{(N)}_{1, t}(y)} 
		\big(a^{\ast}_x + A_x + \alpha_x\big) 
		\big(a_y + A^{\ast}_y + \alpha^{\ast}_y\big)
		\dd x\d y \\
		&+ \bigg( \frac{1}{2}\iintd N^{3}V^\lambda_{1}(N(x-y))\, 
		\wt\phi^{(N)}_{1, t}(x) \,
		\wt\phi^{(N)}_{1, t}(y)\\
		&\qquad\qquad\qquad\qquad\qquad \times 
		\big(a^{\ast}_x + A_x + \alpha_x\big) 
		\big(a^{\ast}_y + A_y + \alpha_y\big) 
		\dd x\d y + \mathrm{h.c.} \bigg),
	\end{aligned}
\end{equation*}
and similarly for $\cT_t^*\,\cG^{(2)}_2\, \cT_t$, with $V^\lambda_1, \wt\phi^{(N)}_{1, t}$, and $a^\ast_x + A_x+\alpha_x$ replaced by $V^\lambda_2, \wt\phi^{(N)}_{2, t}$, and $b^\ast_x + B_x+\beta_x$, respectively.

The inter-species parts are given by
\begin{equation*}
	\begin{aligned}
		\cT_t^*\, \cG^{(12)}_{2}\, \cT_t
		&=\,
		\intd (N^{3} V^\lambda_{12}(N\;\cdot\;)\ast 
		|\wt\phi^{(N)}_{2, t}|^2)(x) 
		\big(a^{\ast}_x + A_x + \alpha_x\big) 
		\big(a_x + A^{\ast}_x + \alpha^{\ast}_x\big)
		\dd  x  \\
		&\quad + 
		\intd (N^{3} V^\lambda_{12}(N\;\cdot\;)\ast 
		|\wt\phi^{(N)}_{1, t}|^2)(x) 
		\big(b^{\ast}_x + B_x + \beta_x\big) 
		\big(b_x + B^{\ast}_x + \beta^{\ast}_x\big)
		\dd x \\
		&\quad+ \bigg(
		\iintd N^{3} V^\lambda_{12}(N(x-y))
		\wt\phi^{(N)}_{1, t}(x) 
		\conj{\wt\phi^{(N)}_{2, t}(y)} \\
		&\qquad\qquad\qquad\qquad\qquad \times
		\big(a^{\ast}_x + A_x + \alpha_x\big) 
		\big(b_y + B^{\ast}_y + \beta^{\ast}_y\big)
		\dd x\d y + \mathrm{h.c.}\bigg) \\
		&\quad+ \bigg(
		\iintd N^{3} V^\lambda_{12}(N(x-y))\, 
		\wt\phi^{(N)}_{1, t}(x)
		\wt\phi^{(N)}_{2, t}(y) \\ 
		&\qquad\qquad\qquad\qquad\qquad \times 
		\big(a^{\ast}_x + A_x + \alpha_x\big) 
		\big(b^{\ast}_y + B_y + \beta_y\big)
		\dd x\d y + \mathrm{h.c.} \bigg).
	\end{aligned}
\end{equation*}
We consider the terms that contain $a_x^\ast a_y^\ast$, $b_x^\ast b_y^\ast$, $a_x^\ast b_y^\ast$ and their conjugate. For example, in the term $\cT_t^*\, \cG^{(1)}_{2}\, \cT_t$, we write the following formula
\begin{multline*}
		\big(a^{\ast}_x + A_x + \alpha_x\big) 
		\big(a^{\ast}_y + A_y + \alpha_y\big)\\
		=
		a^{\ast}_x a^{\ast}_y
		+
		[\big( A_x + \alpha_x\big),\, a^{\ast}_y]
		+
		a^{\ast}_y \big( A_x + \alpha_x\big)
		+
		\big(a^{\ast}_x + A_x + \alpha_x\big) 
		\big( A_y + \alpha_y\big)  ,
\end{multline*}
where in the last line we write the terms with $a^{\ast}_y$ in normal order. Similarly, in the terms $\cT_t^*\, \cG^{(2)}_{2}\, \cT_t$ and $\cT_t^*\, \cG^{(12)}_{2}\, \cT_t$, we write
\begin{multline*}
	\big(b^{\ast}_x + B_x + \beta_x\big) 
	\big(b^{\ast}_y + B_y + \beta_y\big) \\
	=
	b^{\ast}_x b^{\ast}_y
	+
	[\big( B_x + \beta_x\big),\, b^{\ast}_y]
	+
	b^{\ast}_y \big( B_x + \beta_x\big)
	+
	\big(b^{\ast}_x + B_x + \beta_x\big) 
	\big( B_y + \beta_y\big),
\end{multline*}
\begin{multline*}
	\big(a^{\ast}_x + A_x + \alpha_x\big) 
	\big(b^{\ast}_y + B_y + \beta_y\big)\\
	=
	a^{\ast}_x b^{\ast}_y
	+
	[\big( A_x + \alpha_x\big),\, b^{\ast}_y]
	+
	b^{\ast}_y \big( A_x + \alpha_x\big)
	+
	\big(a^{\ast}_x + A_x + \alpha_x\big) 
	\big( B_y + \beta_y\big)  .
\end{multline*}
Substitute the above formulas, and pick up the leading terms that contains $a_x^\ast a_y^\ast$, $b_x^\ast b_y^\ast$, $a_x^\ast b_y^\ast$ and their conjugate. Then we have the following proposition.
\begin{proposition}\label{prop: estimate of quadratic term}
	The intra-species parts of the quadratic terms in the generator can be written as
	\begin{equation*}
		\begin{aligned}
			\cT_t^*\, \cG^{(1)}_{2}\, \cT_t
			=&\, 
			\frac{1}{2}\iintd N^{3}V^\lambda_{1}(N(x-y)) \wt\phi^{(N)}_{1, t}(x) \,
			\wt\phi^{(N)}_{1, t}(y)\, 
			a^{\ast}_xa^{\ast}_y 
			\dd x\d y+\mathrm{h.c.}+C_{N, 2}^{(1)}  + \mathcal{E}_2^{(1)},
		\end{aligned}
	\end{equation*}
	and
	\begin{equation*}
		\begin{aligned}
			\cT_t^*\, \cG^{(2)}_{2}\, \cT_t
			=&\, \frac{1}{2}\iintd N^{3}V^\lambda_{2}(N(x-y))\wt\phi^{(N)}_{2, t}(x) \,
			\wt\phi^{(N)}_{2, t}(y)\, 
			b^{\ast}_x b^{\ast}_y
			\dd x\d y+\mathrm{h.c.}+ C_{N, 2}^{(2)} +  \mathcal{E}_2^{(2)}.
		\end{aligned}
	\end{equation*}
	Moreover, the inter-species parts of the quadratic terms yields
	\begin{equation*}
		\begin{aligned}
			\cT_t^*\, \cG^{(12)}_{2}\, \cT_t
			=&\,
			\iintd N^{3} V^\lambda_{12}(N(x-y))\wt\phi^{(N)}_{1, t}(x)
			\wt\phi^{(N)}_{2, t}(y)\,
			a^{\ast}_x b^{\ast}_y \dd x\d y+ 
			\mathrm{h.c.} + C_{N, 2}^{(12)} + \mathcal{E}_{2}^{(12)}.
		\end{aligned}
	\end{equation*}
	Here, the remainder terms satisfy
	\begin{equation}\label{estimate of quadratic remainder}
		\begin{aligned}
			\pm \mathcal{E}_{2}^{(\ii)}(t) 
			& \leq  C\lambda \|\widetilde{\boldsymbol{\phi}}_t^{(N)}
			\|_{L^\infty_x} (\mathcal{N}+1) ,\\
			\pm\left[\cN, \cE_{2}^{(\ii)}(t)\right] 
			& \leq C\lambda \|\widetilde{\boldsymbol{\phi}}_t^{(N)}
			\|_{L^\infty_x} (\mathcal{N}+1), \\
			\pm \dot{\cE}_{2}^{(\ii)}(t) 
			& \leq C
			\lambda \big(\|\widetilde{\boldsymbol{\phi}}_t^{(N)}
			\|_{L^\infty_x} + \|\partial_t\widetilde{\boldsymbol{\phi}}_t^{(N)}
			\|_{L^\infty_x}\big)
			(\cN+1),
		\end{aligned}
	\end{equation}
	for $\ii\in \{ 1,2,12\}$.  The constant $C$ depends on $\|\widetilde{\boldsymbol{\phi}}_t^{(N)}
	\|_{H^1_x}$, $\|\partial_t \widetilde{\boldsymbol{\phi}}_t^{(N)}
	\|_{H^1_x}$ and $\|\partial^2_t \widetilde{\boldsymbol{\phi}}_t^{(N)}
	\|_{L^2_x}$.
\end{proposition}
We use the following lemma to prove the proposition.

\begin{lem}\label{lem: quadratic}
	Let $\theta_x = A_x,B_x,\alpha_x$, or $\beta_x$. Then for any $\ii\in\{1,2,12\}, \jj\in\{1,2\}, \sharp \in \{\ast, \circ\}$ (i.e. $A^\circ_x:= A_x$), and $k\in \{0, 1\}$, we have the estimates
	\begin{align}\label{est: A, B, alpha, beta}
		\intd \|\bd^k_t\theta^{\sharp}_x \Psi\|^2 \dd x \lesssim \norm{(\cN+1)^{\frac12} \Psi}^2,
	\end{align}
	and 
	\begin{align}\label{est: quadratic 1}
		\iintd N^3 V^\lambda_{\ii}(N(x-y)) 
		|\wt\phi^{(N)}_{\jj, t}(x)|
		\|\partial_t^k \theta^{\sharp}_y \Psi\|^2 \dd x \d y 
		\lesssim\, 
		\lambda
        \|\wt\phi^{(N)}_{\jj, t}\|_{L^\infty_x}
		\norm{(\cN+1)^{\frac12} \Psi}^2.
	\end{align}
	Moreover, let $\eta_x = a_x$ or $b_x$, then we have
	\begin{align}\label{estimate of quodratic 2}
		\iintd N^3 V^\lambda_{\ii}(N(x-y)) 
		|\wt\phi^{(N)}_{\jj, t}(x) |
		\|\eta_y \Psi \|^2 \dd x \d y 
		\lesssim
		\lambda 
        \|\wt\phi^{(N)}_{\jj, t} \|_{L^\infty_x}
		\|\(\cN+1\)^{\frac12} \Psi \|^2.
	\end{align}
\end{lem}
\begin{proof}
	To prove \eqref{est: A, B, alpha, beta} with $k=0$, we apply the 
	the Cauchy--Schwarz inequality to get
	\begin{equation*}
		\begin{aligned}
			\intd
			\| A_x^{\sharp} \Psi \|^2 \dd x
			\le&\, 
			\|\boldsymbol{k}_t \|_{\rm HS}^2 
			\|(\mathcal{N}+1)^{\frac12} \Psi \|^2,\\
			\intd
			\| \alpha_x^{\sharp} \Psi \|^2 \dd x
			\le&\, 
			\(\|\operatorname{p}(\boldsymbol{k}_t)\|_{\rm HS}^2
			+ 
			\|\operatorname{r}(\boldsymbol{k}_t)\|_{\rm HS}^2\)
			\|(\mathcal{N}+1)^{\frac12} \Psi \|^2.
		\end{aligned}
	\end{equation*}
	Hence, we obtain the desired estimate using the lemma~\ref{lem: estimate of k}. The estimates for $B_x$ and $\beta_x$ are similar. For \eqref{est: quadratic 1}, notice that
	\begin{align*}
		\mathrm{LHS} \text{ of } \eqref{est: quadratic 1}
		\leq 
		\Nrm{N^3 V^\lambda_{\ii}(N \cdot) * |\wt\phi^{(N)}_{\jj, t} |}{L^\infty_x} 
		\intd \|\theta^{\sharp}_y \Psi\|^2\d y.
	\end{align*}
	Now, by Young's convolution inequality, we arrive at the desired result.   The case $k=1$ follows similarly by applying Lemma~\ref{bound of derivative}. Inequality~\eqref{estimate of quodratic 2} follows by a similar argument.
\end{proof}

\begin{proof}[Proof of Proposition \ref{prop: estimate of quadratic term}]
	To prove the first inequality in \eqref{estimate of quadratic remainder}, we apply the Young's inequality to the remainder terms and bound them by Lemma \ref{lem: quadratic}. 
	For example, we have
	\begin{equation*}
		\begin{aligned}
			&\iintd  N^3 V^\lambda_1(N(x-y))|\wt\phi_{1,t}^{(N)}(x) || \wt\phi_{1,t}^{(N)}(y) |   \|A^{\sharp}_x \Psi \| \|A^{\sharp}_y \Psi \| \dd x \d y  \\
			&\leq 
			\iintd  N^3 V^\lambda_1(N(x-y))|\wt\phi_{1,t}^{(N)}(x) |^2  \|A^{\sharp}_y \Psi \|^2 \dd x \d y\\
			&\quad +
			\iintd  N^3 V^\lambda_1(N(x-y)) | \wt\phi_{1,t}^{(N)}(y) |^2   \|A^{\sharp}_x \Psi \|^2 \dd x \d y.
		\end{aligned}
	\end{equation*}
	Therefore, it can be estimated using \eqref{est: quadratic 1}. Similar treatment for the other terms in the remainder $\mathcal{E}_{2}^{(\ii)}$.
	
	The second inequality in \eqref{estimate of quadratic remainder} follows similarly, using the fact that the commutator of $\cN$ with the remainder terms $\mathcal{E}_{2}^{(\ii)}$ leaves their form unchanged (apart from the constant terms and the quadratic terms with one creation and one annihilation operators, whose
	contribution to the commutator $[\cN , \mathcal{E}_{2}^{(\ii)}(t)]$ vanishes).
	
	Also the third inequality in \eqref{estimate of quadratic remainder} can be proven analogously, by using \eqref{est: quadratic 1}.
\end{proof}

\noindent\textbf{(iii) Cubic term}.
For the cubic terms in the generator, we revisit the notations $\cG_{3} = \cG_{3}^{(1)} + \cG_{3}^{(2)} + \cG_{3}^{(12)}$ which are given by \eqref{def:G_cubic_terms_intra-species} and \eqref{def:G_cubic_terms_inter-species}. The intra-species parts give 
\begin{equation}\label{generator: intra cubic}
	\begin{aligned}
		\cT_t^*\, \cG^{(1)}_3\, \cT_t =& 
		\frac{1}{\sqrt{N}}\iintd 
		N^{3}V^\lambda_{1}(N(x-y))\, \wt\phi^{(N)}_{1, t}(y)\, 
		\big( a^{\ast}_y + A_y + \alpha^{\ast}_y \big) \\
		& \qquad\qquad\qquad\times 
		\big( a^{\ast}_x + A_x + \alpha^{\ast}_x \big) 
		\big( a_x + A^{\ast}_x + \alpha_x \big) \dd x\d y +\mathrm{h.c.}\\
		\cT_t^*\, \cG^{(2)}_3 \,\cT_t
		=&
		\frac{1}{\sqrt{N}}\iintd
		N^{3}V^\lambda_{2}(N(x-y))\, \wt\phi^{(N)}_{2, t}(y)\, 
		\big( b^{\ast}_y + B_y + \beta^{\ast}_y \big)\\
		& \qquad\qquad\qquad\times 
		\big( b^{\ast}_x + B_x + \beta^{\ast}_x \big)
		\big( b_x + B^{\ast}_x + \beta_x \big)
		\dd x\d y +\mathrm{h.c.}
	\end{aligned}
\end{equation}
For the inter-species part one obtains
\begin{equation}\label{generator: inter cubic}
	\begin{aligned}
		\cT_t^*\, \cG^{(12)}_{3}\, \cT_t 
		=&
		\frac{1}{\sqrt{N}}\iintd 
		N^{3} V^\lambda_{12}(N(x-y))\, \wt\phi^{(N)}_{1, t}(y)\, 
		\big( a^{\ast}_y + A_y + \alpha^{\ast}_y \big)\\
		& \qquad\qquad\qquad\times 
		\big( b^{\ast}_x + B_x + \beta^{\ast}_x \big)
		\big( b_x + B^{\ast}_x + \beta_x \big)\dd x\d y\\
		&+ \frac{1}{\sqrt{N}}\iintd N^{3} V^\lambda_{12}(N(x-y))\, \wt\phi^{(N)}_{2, t}(y)\, 
		\big( b^{\ast}_y + B_y + \beta^{\ast}_y \big)\\
		& \qquad\qquad\qquad\quad\times 
		\big( a^{\ast}_x + A_x + \alpha^{\ast}_x \big) 
		\big( a_x + A^{\ast}_x + \alpha_x \big) \dd x\d y+\mathrm{h.c.}
	\end{aligned}  
\end{equation}
By direct computation, we obtain the following commutation identities
\begin{equation}\label{eq:commute formula}
	\begin{aligned}
		[A_y, a^{\ast}_x] =&\, \conj{k_{1}(x, y)}, & 
		[B_y, b^{\ast}_x] =&\, \conj{k_{2}(x, y)}, \\
		[A_y, b^{\ast}_x] =&\, \conj{k_{21}(x, y)},
		&  [B_y, a^{\ast}_x] =&\, \conj{k_{12}(x, y)}.
	\end{aligned}
\end{equation}
In particular, it follows that 
\begin{multline}\label{eq:expansion_identity}
	\big( a^{\ast}_y + A_y + \alpha^{\ast}_y \big)
	\big( a^{\ast}_x + A_x + \alpha^{\ast}_x \big)\\
	=
	\big( a^{\ast}_y + A_y  \big)
	\big( a^{\ast}_x + A_x  \big)
	+
	\alpha^{\ast}_y 
	\big( a^{\ast}_x + A_x  \big)
	+
	\big( a^{\ast}_y + A_y + \alpha^{\ast}_y \big)
	\alpha^{\ast}_x .
\end{multline}
Rewriting the RHS in the normal order regarding to $a^\sharp$ (leaving the operator $A^\sharp$ and $\alpha^\sharp$ non-normal ordered), we obtain 
\begin{subequations}\label{eq:expansion_identities}
	\begin{equation}
		\text{RHS of }\eqref{eq:expansion_identity} = \com{A_y, a^\ast_x}+ \eps_1(x,y)
		=
		\conj{k_{1}(x, y)} + \eps_1(x,y)
	\end{equation}
	where 
	\begin{equation*}
		\eps_1(x,y) = a^\ast_y a^\ast_x + [\alpha^{\ast}_y, a^{\ast}_x] + \tilde{\eps}_{1}(x, y),
	\end{equation*}
	and $\tilde{\eps}_{1}(x, y)$ consists of everything else. 
	
	Similarly, we write
	\begin{align}
		\big( b^{\ast}_y + B_y + \beta^{\ast}_y \big)
		\big( b^{\ast}_x + B_x + \beta^{\ast}_x \big)
		=&\,
		\conj{k_{2}(x, y)} + \eps_2(x,y)\\ \notag
		=&\, \conj{k_{2}(x, y)} + b^\ast_y b^\ast_x + [\beta^{\ast}_y, b^{\ast}_x] + \tilde{\eps}_{2}(x, y),\\
		\big( a^{\ast}_y + A_y + \alpha^{\ast}_y \big)
		\big( b^{\ast}_x + B_x + \beta^{\ast}_x \big)
		=&\,
		\conj{k_{21}(x, y)} + \eps_{21}(x,y)\\ \notag
		=&\, \conj{k_{21}(x, y)} + a^\ast_y b^\ast_x + [\alpha^{\ast}_y, b^{\ast}_x] + \tilde{\eps}_{21}(x, y),\\
		\big( b^{\ast}_y + B_y + \beta^{\ast}_y \big)
		\big( a^{\ast}_x + A_x + \alpha^{\ast}_x \big)
		=&\, 
		\conj{k_{12}(x, y)} + \eps_{12}(x,y)\\ \notag
		=&\, \conj{k_{12}(x, y)} + b^\ast_y a^\ast_x + [\beta^{\ast}_y, a^{\ast}_x] + \tilde{\eps}_{12}(x, y).
	\end{align}
\end{subequations}
A key observation here is that $\eps_{\ii}(x,y)$ always has a term containing two $a^\ast$, $b^\ast$, and $\tilde{\eps}_{\ii}(x,y)$ is an operator-valued distribution with parameters $x, y$. Moreover, it follows from \eqref{def: alpha and beta} that
\begin{equation}\label{estimate of commutator terms}
	\pm [\alpha^{\ast}_y, a^{\ast}_x],\; 
	\pm [\alpha^{\ast}_y, b^{\ast}_x]
	\le |\operatorname{r}(\boldsymbol{k}_t)(x,y)|_{\rm F}.
\end{equation}
It can be bounded similarly when we replace $\alpha$ with $\beta$.
See also \eqref{eq:normal ordered conjugated cubic terms} for a fully normal-ordered rewriting. 

Substituting \eqref{eq:expansion_identities} into \eqref{generator: intra cubic} and \eqref{generator: inter cubic}, and summing up the terms yields
\begin{equation*}
	\begin{aligned}
		\cT_t^*\, \cG_3\, \cT_t 
		&=\, \frac{1}{\sqrt{N}}\iintd N^{3}V^\lambda_{1}(N(x-y))\, \wt\phi^{(N)}_{1, t}(y)\, 
		\conj{k_{1}(x, y)}
		\big( a_x + A^{\ast}_x + \alpha_x \big) \dd x\d y \\
		&\phantom{=}\, + 
		\frac{1}{\sqrt{N}}\iintd N^{3}V^\lambda_{2}(N(x-y))\, \wt\phi^{(N)}_{2, t}(y)\, 
		\conj{k_{2}(x, y)}
		\big( b_x + B^{\ast}_x + \beta_x \big)
		\dd x\d y \\
		&\phantom{=}\, + \frac{1}{\sqrt{N}}\iintd N^{3} V^\lambda_{12}(N(x-y))\, \wt\phi^{(N)}_{1, t}(y)\, 
		\conj{k_{21}(x,y)}
		\big( b_x + B^{\ast}_x + \beta_x \big)\dd x\d y \\
		&\phantom{=}\, + \frac{1}{\sqrt{N}}\iintd N^{3} V^\lambda_{12}(N(x-y))\, \wt\phi^{(N)}_{2, t}(y)\, 
		\conj{k_{12}(x, y)}
		\big( a_x + A^{\ast}_x + \alpha_x \big) \dd x\d y \\
		&\phantom{=}\, +\mathrm{h.c.} +
		\mathcal{E}_3.
	\end{aligned}
\end{equation*}
Here, the remainder terms in $\mathcal{E}_3$ have the same form as the other terms in $\cT_t^*\, \cG_{3}\, \cT_t$, except that we replace $\conj{k_{\ii}(x, y)}$ by $\eps_{\ii}(x,y)$.

Moreover, by the definition of $k_{\ii}$ with $\ii\in \{1,2,12\}$, together with the \eqref{compute TG^1T}, we finally obtain
\begin{equation*}
	\cT_t^*\, \cG_{3}\, \cT_t
	=
	- \cT_t^*\, \cG_{1}\, \cT_t
	+
	\mathcal{E}_3.
\end{equation*}
The properties of $\cE_3$ are summarized in the following proposition.
\begin{prop}\label{prop: T G1 G2 T}
	It follows from the above computation that
	\begin{equation*}
		\cT_t^*\, (\cG_1+\cG_3)\, \cT_t
		=
		\mathcal{E}_3.
	\end{equation*}
	Moreover, the remainder term $\mathcal{E}_3$ satisfies the estimates:
	\begin{align}
		\pm \mathcal{E}_3(t) & 
		\leq 
		\|\widetilde{\boldsymbol{\phi}}_t^{(N)}
		\|_{L^\infty_x}
		\( \delta \mathcal{V}
		+
		C_\delta \lambda \big( 1+\mathcal{N}
		+\frac{(\mathcal{N}+1)^2}{N}
		\big)\), \label{est:E3_operator_bound}\\
		\pm\com{\mathcal{N}, \mathcal{E}_3(t)} 
		& \leq \|\widetilde{\boldsymbol{\phi}}_t^{(N)}
		\|_{L^\infty_x}
		\( \delta \mathcal{V}
		+
		C_\delta \lambda \big( 1+\mathcal{N}
		+  \frac{(\mathcal{N}+1)^2}{N}\big)
		\), \label{est:E3_number_commutator_operator_bound}\\
		\pm \dot{\mathcal{E}}_3(t) 
		& \leq 
		\big(\|\widetilde{\boldsymbol{\phi}}_t^{(N)}
		\|_{L^\infty_x}
		+ \|\partial_t\widetilde{\boldsymbol{\phi}}_t^{(N)}
		\|_{L^\infty_x}\big)\label{est:E3_time-derivative_operator_bound} \( \delta \mathcal{V}
		+
		C_\delta \lambda \big(1+ \mathcal{N}
		+  \frac{(\mathcal{N}+1)^2}{N}\big)
		\), 
	\end{align}
	where $\mathcal{V} = \mathcal{V}_1 + \mathcal{V}_2 + \mathcal{V}_{12}$ is the interaction part in the Hamiltonian. The constant $C_\delta$ depends on $\delta$, $\|\widetilde{\boldsymbol{\phi}}_t^{(N)}
	\|_{H^1_x}$, $\|\partial_t \widetilde{\boldsymbol{\phi}}_t^{(N)}
	\|_{H^1_x}$ and $\|\partial^2_t \widetilde{\boldsymbol{\phi}}_t^{(N)}
	\|_{L^2_x}$.
\end{prop}

In order to prove the proposition, we introduce the following lemma.
\begin{lem}\label{lemma: quatic}
	Let $\theta_x, \widetilde{\theta}_x = A_x,B_x,\alpha_x$, or $\beta_x$, and $\eta_x = a_x$ or $b_x$.
	Then for any $\ii\in\{1,2,12\}$ we have the estimates
	\begin{equation}\label{ineqn: quatic}
		\begin{aligned}
			\iintd N^3 V^\lambda_{\ii}(N(x-y)) \|\theta_x^{\sharp} \eta_y \Psi \|^2 \dd x \d y 
			\lesssim \lambda
            \|\widetilde{\boldsymbol{\phi}}_t^{(N)}
			\|_{L^\infty_x}^2
			\|(\mathcal{N}+1) \Psi\|^2,\\
			\iintd N^3 V^\lambda_{\ii}(N(x-y)) \|\theta_x^{\sharp} \widetilde{\theta}_y^{\sharp} \Psi \|^2 \dd x \d y 
			\lesssim
			\lambda
            \|\widetilde{\boldsymbol{\phi}}_t^{(N)}
			\|_{L^\infty_x}^2 
			\|(\mathcal{N}+1) \Psi\|^2.
		\end{aligned}
	\end{equation}
	Moreover, we have
	\begin{equation}\label{ineqn: quatic time derivative}
		\begin{aligned}
			\iintd N^3 V^\lambda_{\ii}(N(x-y)) &\|(\partial_t \theta_x^{\sharp}) \eta_y \Psi \|^2 \dd x \d y \\
			\leq &
			C 
            \lambda
            \Big(\|\widetilde{\boldsymbol{\phi}}_t^{(N)}
			\|_{L^\infty_x}^2
			+ \|\partial_t\widetilde{\boldsymbol{\phi}}_t^{(N)}
			\|_{L^\infty_x}^2
            \Big)
			\|(\mathcal{N}+1) \Psi\|^2,\\
			\iintd N^3 V^\lambda_{\ii}(N(x-y)) &\|(\partial_t \theta_x^{\sharp}) \widetilde{\theta}_y^{\sharp} \Psi \|^2 \dd x \d y \\
			\leq& 
			C \lambda
            \Big(\|\widetilde{\boldsymbol{\phi}}_t^{(N)}
			\|_{L^\infty_x}^2
			+ \|\partial_t\widetilde{\boldsymbol{\phi}}_t^{(N)}
			\|_{L^\infty_x}^2
            \Big)
			\|(\mathcal{N}+1) \Psi\|^2.
		\end{aligned}
	\end{equation}
	
\end{lem}
\begin{proof}
	In order to estimate the first inequality in \eqref{ineqn: quatic}, we note that
	\begin{equation*}
		\begin{aligned}
			\iintd N^3 V^\lambda_{\ii}(N(x-y)) \|\theta_x^{\sharp} \eta_y \Psi \|^2 \dd x \d y   
			\leq C
            \lambda 
			\intd \sup _{x \in \mathbb{R}^3}
			\|\theta_x^{\sharp} \eta_y \Psi\|^2\d y .
		\end{aligned}
	\end{equation*}
	Moreover, let us denote
	\begin{equation*}
		M_1(t):= \max\left\{
		\sup _{x \in \mathbb{R}^3} \|\boldsymbol{k}_t(\cdot, x) \|_{L^2}, \,
		\sup _{x \in \mathbb{R}^3}  \|\operatorname{p}(\boldsymbol{k}_t) (\cdot, x) \|_{L^2},\,
		\sup_{x \in \mathbb{R}^3}  \|\operatorname{r}(\boldsymbol{k}_t) (\cdot, x) \|_{L^2}
		\right\},
	\end{equation*}
	then it holds that
	\begin{equation*}
		\begin{aligned}
			\intd \sup _{x \in \mathbb{R}^3}
			\|\theta_x^{\sharp} \eta_y \Psi\|^2\d y
			\le &\, 
			M_1(t) \|(\mathcal{N}+1)^{\frac12} \Psi \|.
		\end{aligned}
	\end{equation*}
	Applying Lemma~\ref{lem: estimate of k}, the first inequality in \eqref{ineqn: quatic} follows.
	The other inequalities can be derived similarly if we employ Lemma~\ref{bound of derivative}.
\end{proof}

\begin{proof}[Proof of the Proposition~\ref{prop: T G1 G2 T}]
	The main idea is to apply the Cauchy--Schwarz inequality to the terms in $\cE_3$. In general, we control the cubic terms by means of the quartic and quadratic contributions, which
	are then estimated using Lemma~\ref{lem: quadratic} (the quadratic part) and Lemma~\ref{lemma: quatic} (the quartic part).  For simplicity, it suffices to consider the following term of $\cE_3$ 
	\begin{align}
		\frac{1}{\sqrt{N}}\iintd N^{3}V^\lambda_{1}(N(x-y))\, \wt\phi^{(N)}_{1, t}(y)\, 
		\big( a_x + A^{\ast}_x + \alpha_x \big)   \eps_{1}(x, y)\dd x\d y. 
	\end{align}
	
	To prove the \eqref{est:E3_operator_bound}, notice we have that
	\begin{equation*}
		\begin{aligned}
			&\frac{1}{\sqrt{N}}\n{\iintd N^{3}V^\lambda_{1}(N(x-y))\, \wt\phi^{(N)}_{1, t}(y)\, 
				\left\langle\Psi, \, \eps_{1}(x, y)
				\big( a_x + A^{\ast}_x + \alpha_x \big)
				\Psi\right\rangle \dd x\d y} \\
			&\le
			\frac{\delta}{2N}\iintd N^{3}V^\lambda_{1}(N(x-y))\, 
			|\wt\phi^{(N)}_{1, t}(y)|\Nrm{\eps_{1}^{\ast}(x, y) \Psi}{}^2
			\dd x\d y\\
			&\quad+ C_\delta
			\iintd N^{3}V^\lambda_{1}(N(x-y))\,|\wt\phi^{(N)}_{1, t}(y)|
			\Nrm{( a_x + A^{\ast}_x + \alpha_x ) 
				\Psi}{}^2 \dd x\d y.
		\end{aligned}
	\end{equation*}
	Here, we estimate the second term by Lemma~\ref{lem: quadratic}. As for the first term, we recall that $\eps_{1}(x, y) = a^\ast_y a^\ast_x + [\alpha^{\ast}_y, a^{\ast}_x] + \tilde{\eps}_{1}(x, y)$, together with \eqref{estimate of commutator terms}, it then leads to
	\begin{align*}
		&\frac{\delta}{2N}\iintd  N^{3}V^\lambda_{1}(N(x-y))\, |\wt\phi^{(N)}_{1, t}(y)|
		\Nrm{\eps_{1}^{\ast}(x, y) \Psi}{}^2
		\dd x\d y\\
		&\leq \,\delta \norm{\wt\phi^{(N)}_{1, t}}_{L^\infty_x}
        \left\langle\Psi, \,
				\mathcal{V}_1 
				\Psi
                \right\rangle
        +
		\frac{\delta}{2N}\iintd N^{3}V^\lambda_{1}(N(x-y))\, |\wt\phi^{(N)}_{1, t}(y)|
		|\operatorname{r}(\boldsymbol{k}_t)(x,y)|_{\rm F}^2
		\dd x\d y\\
		&\quad+
		\frac{\delta}{2N}\iintd N^{3}V^\lambda_{1}(N(x-y))\, |\wt\phi^{(N)}_{1, t}(y)|
		\Nrm{\tilde{\eps}_{1}^{\ast}(x, y) \Psi}{}^2
		\dd x\d y.
	\end{align*}
	Here, the term in the last line can be estimated by using Lemma~\ref{lemma: quatic}. The second term can be estimated by using \eqref{pointwise bound}, which yields
	\begin{equation*}
		\frac{\delta}{2N}\iintd N^{3}V^\lambda_{1}(N(x-y))\, |\wt\phi^{(N)}_{1, t}(y)|
		|\operatorname{r}(\boldsymbol{k}_t)(x,y)|_{\rm F}^2
		\dd x\d y
		\le
		C \lambda\|\widetilde{\boldsymbol{\phi}}_t^{(N)}
		\|_{L^\infty_x}^3.
	\end{equation*}
	The other terms in $\cE_3$ are handled in a similar manner. Furthermore, the proof of the \eqref{est:E3_time-derivative_operator_bound} is the same. 
	
	The proof of the \eqref{est:E3_number_commutator_operator_bound} follows from the \eqref{eq:number_operator_commuting_z} and Lemmas~\ref{lem: quadratic}--\ref{lemma: quatic}. Thus, this completes the proof of the theorem.
\end{proof}

\subsection{Estimates for the Kinetic Terms}
Recall the kinetic operator $\mathcal{K} = \mathcal{K}_1 + \mathcal{K}_2$ defined in \eqref{eq:multiFockH}. Using \eqref{eq:creation_annihilation_expansion}, we write  
\begin{equation*}\label{transform of kinetic}
	\begin{aligned}
		\cT_t^*\, \mathcal{K}_1\, \cT_t
		=&\, \intd \nabla_x
		\Big( a^*_x + A_x + \alpha_x \Big) 
		\nabla_x\Big(	a_x + A^*_x + \alpha^*_x \Big)  \dd x \\
		=&\, \cK_1 
		+
		\intd  \nabla_xA_x\nabla_xA^*_x \dd x
		+
		\intd \nabla_x\alpha_x \nabla_x\alpha^*_x \dd x \\
		& +
		\(
		\intd \{\nabla_xa^*_x \nabla_xA^*_x  
		+ 
		\nabla_xa^*_x \nabla_x\alpha^*_x 
		+
		\nabla_xA_x\nabla_x\alpha^*_x\} \dd x + \mathrm{h.c.} \) .
	\end{aligned}
\end{equation*}
Here, we rewrite the term $ \nabla_xA_x\nabla_x \alpha^*_x$ and its conjugate in normal order, together with $\nabla_xA_x\nabla_xA^*_x$, lead to
\begin{equation*}
	\begin{aligned}
		\cT_t^*\, \mathcal{K}_1\, \cT_t
		=&\, 
		C_{N, K}^{(1)} + \cK_1 
		+\intd  \nabla_xA^*_x \nabla_xA_x \dd x
		+
		\intd \nabla_x\alpha_x\nabla_x\alpha^*_x \dd x \\
		& +
		\(
		\intd \{\nabla_xa^*_x \nabla_xA^*_x  
		+ 
		\nabla_xa^*_x \nabla_x\alpha^*_x 
		+
		\nabla_x\alpha^*_x \nabla_xA_x\}  \dd x + \mathrm{h.c.} \).
	\end{aligned}
\end{equation*}
Similarly, for the $\mathcal{K}_2$ term, it follows that
\begin{equation*}
	\begin{aligned}
		\cT_t^*\, \mathcal{K}_2\, \cT_t
		=&\, C_{N, K}^{(2)} + \mathcal{K}_2 
		+
		\intd \nabla_xB^*_x\nabla_xB_x \dd x 
		+\intd \nabla_x\beta_x\nabla_x\beta^*_x \dd x \\
		&+\(
		\intd \{\nabla_xb^*_x \nabla_xB^*_x  
		+ 
		\nabla_xb^*_x \nabla_x\beta^*_x
		+
		\nabla_x\beta^*_x\nabla_xB_x\} \dd x
		+ \mathrm{h.c.} \).
	\end{aligned}
\end{equation*}
Now we assert the following proposition:
\begin{prop}\label{prop: Kinetic term}
	We have
	\begin{align*}
		\cT_t^*\, \mathcal{K}_1\, \cT_t 
		&=\, C_{N, K}^{(1)} + \mathcal{K}_1+ \mathcal{E}_K^{(1)}\\
		&\phantom{=}+  N \iintd \Delta w_{1, \ell}^{(N)} (x-y) \wt\phi^{(N)}_{1,t}(x) \wt\phi^{(N)}_{1,t}(y)\, a_x^* a_y^* \dd x \d y +\mathrm{h.c.}\\
		&\phantom{=}+ N \iintd \Delta w_{12, \ell}^{(N)}(x-y) \wt\phi^{(N)}_{1,t}(x) \wt\phi^{(N)}_{2,t}(y)\, a_x^* b_y^*  \dd x \d y+\mathrm{h.c.},
	\end{align*}  
	and, similarly, we have
	\begin{align*}
		\cT_t^*\, \mathcal{K}_2\, \cT_t 
		&= C_{N, K}^{(2)} + \mathcal{K}_2+ \mathcal{E}_K^{(2)}\\
		&\phantom{=}+  N \iintd \Delta w_{2,\ell}^{(N)} (x-y) \wt\phi^{(N)}_{2,t}(x) \wt\phi^{(N)}_{2,t}(y)\, b_x^* b_y^* \dd x \d y +\mathrm{h.c.}\\
		&\phantom{=}+ N \iintd \Delta w_{12, \ell}^{(N)} (x-y)\wt\phi^{(N)}_{1,t}(x) \wt\phi^{(N)}_{2,t}(y)\, a_x^* b_y^* \dd x \d y+\mathrm{h.c.},
	\end{align*}  
	where the remainder terms satisfy
	\begin{equation}\label{estimate of E_K}
		\begin{aligned}
			\pm \mathcal{E}^{(\jj)}_K(t)
			\leq&\, \delta \|\widetilde{\boldsymbol{\phi}}_t^{(N)}
			\|_{L^\infty_x} \cK_{\jj}
			+C_\delta \|\widetilde{\boldsymbol{\phi}}_t^{(N)}
			\|_{W^{1,\infty}_x} (\mathcal{N}+1), \\
			\pm\left[\mathcal{N}, \mathcal{E}^{(\jj)}_K(t)\right] 
			\leq&\, \delta \|\widetilde{\boldsymbol{\phi}}_t^{(N)}
			\|_{L^\infty_x} \cK_{\jj}
			+ C_\delta \|\widetilde{\boldsymbol{\phi}}_t^{(N)}
			\|_{W^{1,\infty}_x} (\mathcal{N}+1), \\
			\pm \dot{\mathcal{E}}^{(\jj)}_K(t)  
			\leq&\, \delta (\|\widetilde{\boldsymbol{\phi}}_t^{(N)}
			\|_{L^\infty_x}
			+ \|\partial_t\widetilde{\boldsymbol{\phi}}_t^{(N)}
			\|_{L^\infty_x}) \cK_{\jj}\\
			& + C_\delta (\|\widetilde{\boldsymbol{\phi}}_t^{(N)}
			\|_{W^{1,\infty}_x}
			+ \|\partial_t\widetilde{\boldsymbol{\phi}}_t^{(N)}
			\|_{L^\infty_x}) (\mathcal{N}+1).
		\end{aligned}
	\end{equation}
	The constant $C_\delta$ depends on $\delta$, $\|\widetilde{\boldsymbol{\phi}}_t^{(N)}
	\|_{H^1_x}$ and $\|\partial_t \widetilde{\boldsymbol{\phi}}_t^{(N)}
	\|_{H^1_x}$.
\end{prop}

Before we prove the proposition, let us introduce the following estimates:
\begin{lem}\label{lem: quodratic lemma}
	Let $j_1, j_2 \in L^2\left(\mathbb{R}^3 \times \mathbb{R}^3\right)$, and denote $j_{i, x}(z):=j_i(z, x)$ for $i=1,2$. 
	We use the notation $\eta_x = a_x$ or $b_x$, then we have
	\begin{equation}\label{estimate of kinetic 1}
		\n{\intd \big\langle\Psi,\, \eta^{\sharp}(j_{1, x}) \eta^{\sharp}(j_{2, x}) \Psi\big\rangle \dd x} 
		\leq 
		\| j_1 \|_{\rm HS}
		\| j_2 \|_{\rm HS}
		\|(\mathcal{N}+1)^{\frac12} \Psi \|^2 .
	\end{equation}
	
	Moreover, for every $\delta > 0$, there exists $C_\delta > 0$ such that
	\begin{align}\label{estimate of kinetic 2}
		\Big|\intd \left\langle\Psi, \nabla_x \eta_x^*\, \eta^{\sharp}(j_{1, x}) \Psi\right\rangle \dd x \Big| 
		\leq 
		\delta \| j_1 \|_{\rm HS} \langle\Psi, \mathcal{K} \Psi\rangle
		+ C_\delta \| j_1 \|_{\rm HS} \|(\mathcal{N}+1)^{\frac12} \Psi \|^2. 
	\end{align}
	Here, $\cK$ is kinetic energy part in the Hamiltonian. Furthermore, let $\theta_x = A_x$ or $B_x$, we also have
	\begin{multline}\label{estimate of kinetic 3}
		\Big|\intd \big\langle\Psi, \nabla_x \theta^\ast_x \, \eta^{\sharp}(j_{1, x}) \Psi\big\rangle \dd x \Big| 
		\leq\, 
		\delta \| \boldsymbol{k}_t \|_{\rm HS}
		\| j_1 \|_{\rm HS}
		\langle\Psi, \mathcal{K} \Psi\rangle\\
		+
		C_\delta \big(
		\| \boldsymbol{k}_t \|_{\rm HS} + \| j_1 \|_{\rm HS} \big) \langle\Psi,(\mathcal{N}+1) \Psi\rangle.
	\end{multline}
	Finally, we have
	\begin{equation}\label{estimate of kinetic 4}
		\n{\intd \big\langle\Psi, 
			\nabla_x \theta^\ast_x 
			\nabla_x \theta_x  \Psi\big\rangle \dd x} 
		\leq 
		C \|\widetilde{\boldsymbol{\phi}}_t^{(N)}
		\|_{L^\infty_x}^2 
		\|\mathcal{N}^{\frac12} \Psi \|^2.
	\end{equation}
	To control the time derivative of $\mathcal{E}^{(\jj)}_K(t)$, we also use the following bounds. For every $\delta > 0$, there exists $C_\delta > 0$ such that
	\begin{multline}\label{estimate of kinetic 5}
		\Big|\intd \big\langle\Psi, \nabla_x \dot{\theta}^\ast_x \, \eta^{\sharp}(j_{1, x}) \Psi\big\rangle \dd x \Big| 
		\leq 
		\delta \| \dot{\boldsymbol{k}}_t \|_{\rm HS}
		\| j_1 \|_{\rm HS}
		\langle\Psi, \mathcal{K} \Psi\rangle\\
		+
		C_\delta \big(
		\| \dot{\boldsymbol{k}}_t \|_{\rm HS} + \| j_1 \|_{\rm HS} \big) \langle\Psi,(\mathcal{N}+1) \Psi\rangle.
	\end{multline}
	Moreover, we have
	\begin{equation}\label{estimate of kinetic 6}
		\n{\intd \big\langle\Psi, 
			\nabla_x \dot{\theta}^\ast_x 
			\nabla_x \theta_x  \Psi\big\rangle \dd x} 
		\leq 
		C \big(\|\partial_t \widetilde{\boldsymbol{\phi}}_t^{(N)}
		\|_{L^\infty_x}^2
		+ \|\widetilde{\boldsymbol{\phi}}_t^{(N)}
		\|_{L^\infty_x}^2 \big)
		\|\mathcal{N}^{\frac12} \Psi \|^2.
	\end{equation}
\end{lem}

\begin{proof}
	The proof of \eqref{estimate of kinetic 1} and \eqref{estimate of kinetic 2} follows directly from \cite[Lemma 6.2]{ Benedikter2015quantitative}.
	While the proof of \eqref{estimate of kinetic 3} and \eqref{estimate of kinetic 4} is a bit different. For \eqref{estimate of kinetic 3}, we use the same estimate as in \cite[Lemma 6.2]{ Benedikter2015quantitative} but with a bound explicitly depending on $\| \boldsymbol{k}_t \|_{\rm HS}$. 
	For example, we consider the term
	\begin{equation}\label{def: example_term}
		\intd  a^* (\nabla_x k_{1,x} ) \eta^{\sharp}(j_{1, x}) \dd x.
	\end{equation}
	Using integration by parts, the definition of $k_1$, and the estimate of $w_{1,\ell}^{(N)}$ in Lemma~\ref{lem:neumann_scattering_function} (cf. \cite[Lemma 6.2]{ Benedikter2015quantitative}), we have 
	\begin{equation}\label{estimtate of nabla k and j}
		\begin{aligned}
			\n{\left\langle\Psi, \eqref{def: example_term} \Psi\right\rangle}
			&\leq 
			\iintd  |k_{1}(x, y)| \|\nabla_y a_y \Psi\| \|\eta^{\sharp}(j_{1, x}) \Psi\| 
			\dd x \d y\\
			&
			\quad+ C \iintd \frac{1}{|x-y|}
			\Big( |\nabla \wt\phi_{1,t}^{(N)}(x) |
			|\wt\phi_{1,t}^{(N)}(y) | \\
			&\qquad\qquad\qquad\qquad
			+  | \nabla \wt\phi_{1,t}^{(N)}(y) |
			| \wt\phi_{1,t}^{(N)}(x) |\Big)
			\|a_y \Psi\| \|\eta^{\sharp}(j_{1, x})\Psi\| \dd x \d y.
		\end{aligned}
	\end{equation}
	This implies that
	\begin{equation*}
		\begin{aligned}
			\text{RHS of }\eqref{estimtate of nabla k and j} 
			\leq 
			&\, \|k_{1} \|_{\rm HS} \Big(
			\iintd \|\nabla_y a_y \Psi \|^2
			\|\eta^{\sharp}(j_{1, x}) \Psi \|^2 \dd x \d y \Big)^{\frac{1}{2}}\\
			& 
			+C \Big(\iintd  \frac{1}{|x-y|^2}  |\wt\phi_{1,t}^{(N)}(y) |^2 
			\|\eta^{\sharp}(j_{1, x}) \Psi \|^2 
			\dd x \d y\Big)^{\frac{1}{2}}\\
			&\qquad \times 
			\Big( \iintd |\nabla \wt\phi_{1,t}^{(N)}(x) |^2
			\| a_y \Psi \|^2 \dd x \d y\Big)^{\frac{1}{2}} \\
			& 
			+C \Big(\iintd  \frac{1}{|x-y|^2} |\wt\phi_{1,t}^{(N)}(x) |^2
			\|a_y \Psi \|^2 
			\dd x \d y \Big)^{\frac{1}{2}}\\
			&\qquad \times
			\Big( \iintd |\nabla \wt\phi_{1,t}^{(N)}(y) |^2
			\|\eta^{\sharp}(j_{1, x}) \Psi\|^2 
			\dd x \d y \Big)^{\frac{1}{2}}.
		\end{aligned}
	\end{equation*}
	Using the fact $\intd \|\eta^{\sharp}(j_{1, x}) \Psi \|^2 \dd x \le \| j_1 \|_{\rm HS}^2 \|(\mathcal{N}+1)^{\frac12} \Psi\|^2$, and applying  Hardy’s inequality, we conclude that
	\begin{align*}
		\n{\left\langle\Psi, \eqref{def: example_term} \Psi\right\rangle}
		\leq 
		\delta \|k_{1} \|_{\rm HS} \| j_1 \|_{\rm HS}
		\langle\Psi, \mathcal{K} \Psi\rangle
		+
		C_\delta \big( \|k_{1} \|_{\rm HS} + \| j_1 \|_{\rm HS}\big) \langle\Psi,(\mathcal{N}+1) \Psi\rangle .
	\end{align*}
	Here, the constant $C_\delta$ depends on $\|\wt\phi_{1,t}^{(N)} \|_{H^1_x}$.
	The other cases in \eqref{estimate of kinetic 3} and the bound \eqref{estimate of kinetic 5} can be proven similarly.
	
	As for \eqref{estimate of kinetic 4}, by the Cauchy--Schwarz inequality, we have that
	\begin{equation}\label{quodratic estimate of ak}
		\intd \big\langle\Psi,\,
		a^*\big(\nabla_x k_{1,x}\big) a\big(\nabla_x k_{1,x}\big) \Psi \big\rangle \dd x
		\le 
		\|g\|_{\rm HS} \|\mathcal{N}^{\frac12} \Psi \|^2.
	\end{equation}
	where $g(y_1,y_2)= \intd  \nabla_x k_{1}(y_1, x) \nabla_x \conj{k_{1}}(y_2, x) \dd x$. 
	
	Furthermore, using the definition of $k_1$ and Lemma~\ref{lem:neumann_scattering_function}, we have
	\begin{equation*}
		\begin{aligned}
			&\n{g(y_1, y_2)}^2 \\
			&\leq\, 
			C  \int_{\substack{|x_1-y_1|\le\ell,\\|x_1-y_2|\le\ell}}
			\int_{
				\substack{
					|x_2-y_1|\le\ell,\\|x_2-y_2|\le\ell}
			} 
			\frac{\big|\wt\phi_{1,t}^{(N)}(x_1)\big|^2
				\big|\wt\phi_{1,t}^{(N)}(x_2) \big|^2 
				\big|\wt\phi_{1,t}^{(N)}(y_1)\big|^2
				\big|\wt\phi_{1,t}^{(N)}(y_2)\big|^2}
			{|x_1-y_1|^2\left|x_1-y_2\right|^2\left|x_2-y_1\right|^2\left|x_2-y_2\right|^2} \dd x_2 \d x_1  \\
			&\quad +
			C   \int_{\substack{|x_1-y_1|\le\ell,\\|x_1-y_2|\le\ell}}
			\int_{\substack{|x_2-y_1|\le\ell,\\|x_2-y_2|\le\ell}} 
			\frac{\big|\nabla\wt\phi_{1,t}^{(N)}(x_1)\big|^2
				\big|\nabla\wt\phi_{1,t}^{(N)}(x_2) \big|^2 
				\big|\wt\phi_{1,t}^{(N)}(y_1)\big|^2
				\big|\wt\phi_{1,t}^{(N)}(y_2)\big|^2}{\left|x_1-y_1\right|\left|x_1-y_2\right|\left|x_2-y_1\right|\left|x_2-y_2\right|} \dd x_2 \d x_1 \, .
		\end{aligned}
	\end{equation*}
	Then, integrate and applying the Cauchy--Schwarz inequality to the second term yields 
	\begin{equation*}
		\begin{aligned}
			&\|g\|_{\rm HS}^2 \\
			&\leq\,  
			C \|\wt\phi_{1,t}^{(N)} \|_{L^\infty_x}^4
			\iintd  |\wt\phi_{1,t}^{(N)} (y_1) |^2
			|\wt\phi_{1,t}^{(N)} (y_2) |^2
			\(\int_{
				\substack{
					|x-y_1|\le\ell,\\|x-y_2|\le\ell}
			} 
			\frac{\d x}{\left|x-y_1\right|^2\left|x-y_2\right|^2} \)^2 \dd y_1 \d y_2 \\
			& \phantom{=} +
			C   \|\wt\phi_{1,t}^{(N)} \|_{L^\infty_x}^4\iintd \int_{\substack{|x_1-y_1|\le\ell,\\|x_1-y_2|\le\ell}}
			\int_{\substack{|x_2-y_1|\le\ell,\\|x_2-y_2|\le\ell}}
			\frac{ |\nabla \wt\phi_{1,t}^{(N)} (x_1) |^2
				|\nabla \wt\phi_{1,t}^{(N)}(x_2) |^2}
			{\left|x_1-y_1\right|^2\left|x_2-y_2\right|^2} \dd x_2 \d x_1 \d y_1 \d y_2 
		\end{aligned}
	\end{equation*}
	which leads to
	\begin{equation*}
		\begin{aligned}
			\|g\|_{\rm HS}^2 
			\leq&\, C
			\|\wt\phi_{1,t}^{(N)} \|_{L^\infty_x}^4
			\|\wt\phi_{1,t}^{(N)} \|_{H^1_x}^4.
		\end{aligned}
	\end{equation*}
	Substitute the result back into \eqref{quodratic estimate of ak} yields the desired bound. The other terms in \eqref{estimate of kinetic 4} and \eqref{estimate of kinetic 6} can be bounded similarly.
\end{proof}

Now, we are ready to prove Proposition \ref{prop: Kinetic term}.
\begin{proof}[Proof of Proposition \ref{prop: Kinetic term}] 
	We prove the first bound in \eqref{estimate of E_K}. Note that from the definition of $(\alpha_x,\beta_x)$ in \eqref{def: alpha and beta}, the operator $\alpha_x$ and $\beta_x$ can be written in the form of $a^\sharp(j_{1,x}) + b^\sharp(j_{2,x})$, where $\| j_1 \|_{\rm HS}$, $\| j_2 \|_{\rm HS}$ can be bounded by the sum of $\| \operatorname{p}(\boldsymbol{k}_t) \|_{\rm HS}$ and $\| \operatorname{r}(\boldsymbol{k}_t) \|_{\rm HS}$. Therefore, applying Lemma \ref{lem: quodratic lemma} and Lemma \ref{lem: estimate of k}, we estimate the terms in the remainder $\mathcal{E}_K$ except the following terms
	\begin{equation}\label{difficult terms}
		\begin{aligned}
			&\intd \nabla_xa^*_x\nabla_xA^*_x \dd x,
			\qquad 
			\intd \nabla_xb^*_x\nabla_xB^*_x \dd x,
		\end{aligned}
	\end{equation}
	and their hermitian conjugate. It suffices to consider the first term \eqref{difficult terms}. We start by writing
	\begin{equation*}
		\nabla_x k_{1,t}(y, x)
		=
		-N \nabla( w^{N}_{1,\ell}) (x-y) \wt\phi_{1,t}^{(N)}(x) \wt\phi_{1,t}^{(N)}(y) - N w^{N}_{1,\ell} (x-y) \nabla \wt\phi_{1,t}^{(N)}(x) \wt\phi_{1,t}^{(N)}(y).
	\end{equation*}
	Hence, one can write
	\begin{equation}\label{split}
		\begin{aligned}
			\intd \nabla_x a_x^* a^*\left(\nabla_x k_{1,t}\right)\dd x 
			= & 
			-N \intd \nabla w_{1,\ell}^{(N)} (x-y) \wt\phi_{1,t}^{(N)}(x) \wt\phi_{1,t}^{(N)}(y) \nabla_x a_x^* a_y^* \dd x \d y \\
			& -N \intd w_{1,\ell}^{(N)} (x-y) \nabla \wt\phi_{1,t}^{(N)}(x) \wt\phi_{1,t}^{(N)}(y) \nabla_x a_x^* a_y^* \dd x \d y.
		\end{aligned}
	\end{equation}
	The last term on the RHS of \eqref{split} can be written as
	\begin{equation*}
		\intd \nabla_x a_x^* a^*\left(j_x\right)\dd x
		\quad \text{with}\quad
		j(y, x)  
		= 
		-N w_{1,\ell}^{(N)} (x-y) \nabla \wt\phi_{1,t}^{(N)}(x) \wt\phi_{1,t}^{(N)}(y).
	\end{equation*}
	Moreover, using the bound of $w_{1,\ell}^{(N)}$ one has the estimate:
	\begin{equation*}
		\Vert j\Vert_{\rm HS}
		\le C 
		\|\wt\phi_{1,t}^{(N)}\|_{L^\infty_x}
		\|\wt\phi_{1,t}^{(N)}\|_{H^1_x}.
	\end{equation*}
	Therefore, one can bound the last term on the RHS of \eqref{split} in the same manner as \eqref{estimate of kinetic 2}.
	The first term on the RHS of \eqref{split}, on the other hand, can be written as
	\begin{equation*}
		\begin{aligned}
			-&N \iintd \nabla w_{1,\ell}^{(N)} (x-y) \wt\phi_{1,t}^{(N)}(x)  \wt\phi_{1,1}^{(N)}(y) \nabla_x a_x^* a_y^* \dd x \d y  \\
			&=\,  
			N \iintd  \Delta w_{1,\ell}^{(N)} (x-y) \wt\phi_{1,t}^{(N)}(x) \wt\phi_{1,t}^{(N)}(y) a_x^* a_y^*  \dd x \d y  \\
			&\,\phantom{=} +
			N \iintd  \nabla w_{1,\ell}^{(N)} (x-y) \nabla \wt\phi_{1,t}^{(N)}(x) \wt\phi_{1,t}^{(N)}(y) a_x^* a_y^*  \dd x \d y  .
		\end{aligned}
	\end{equation*}
	Integrating by parts, the second term is written as
	\begin{equation*}
		\begin{aligned}
			& N \iintd \nabla  w_{1,\ell}^{(N)} (x-y) \nabla \wt\phi_{1,t}^{(N)}(x) \wt\phi_{1,t}^{(N)}(y) a_x^* a_y^* \dd x \d y  \\
			&= -N \iintd  w_{1,\ell}^{(N)} (x-y) \nabla \wt\phi_{1,t}^{(N)}(x) \nabla \wt\phi_{1,t}^{(N)}(y) a_x^* a_y^* \dd x \d y  \\
			&\quad -N \iintd w_{1,\ell}^{(N)} (x-y) \nabla \wt\phi_{1,t}^{(N)}(x) \wt\phi_{1,t}^{(N)}(y) a_x^* \nabla_y a_y^* \dd x \d y  \\
			&= - 
			\intd \nabla \phi_{1,t}^{(N)}(x) a_x^* a^*\left(N w_{1,\ell}^{(N)} (x-.) \nabla \wt\phi_{1,t}^{(N)}\right) \dd x  
			+\intd \nabla_y a_y^* a^*\left(j_y\right) \dd y
		\end{aligned}
	\end{equation*}
	with $j(y, x)=-N w_{1,\ell}^{(N)}(x-y) \nabla \wt\phi_{1,t}^{(N)}(x) \wt\phi_{1,t}^{(N)}(y)$. 
	Therefore, the second term in the last line can be bounded similarly as \eqref{estimate of kinetic 2}. On the other hand, by \eqref{est:creation and annihilation}, the first term is estimate by
	\begin{equation*}
		\begin{aligned}
			\Big|\intd \nabla \wt\phi_{1,t}^{(N)}(x)&\Big\langle a_x \Psi, a^*\big(N w_{1,\ell}^{(N)} (x-.) \nabla \wt\phi_{1,t}^{(N)}\big) \Psi\Big\rangle \dd x  \Big| \\
			& \leq C 
			\|\nabla \wt\phi_{1,t}^{(N)}\|_{L^\infty_x} 
			\|N w_{1,\ell}^{(N)} (x-y) \nabla \wt\phi_{1,t}^{(N)}(y)\|_{\rm HS}
			\|(\mathcal{N}+1)^{\frac12} \Psi\|^2 \\
			& \leq C 
			\|\nabla \wt\phi_{1,t}^{(N)}\|_{L^\infty_x}
			\|\nabla \wt\phi_{1,t}^{(N)}\|_{L^2_x}
			\|(\mathcal{N}+1)^{\frac12} \Psi\|^2,
		\end{aligned}
	\end{equation*}
	which completes the proof of the proposition.
\end{proof}

\subsection{Estimates for the Interaction Terms}

Recall the interaction operator $\mathcal{V}$ defined in \eqref{eq:multiFockH}.

The properties of $\cT_t^*\, \mathcal{V}\, \cT_t$ are summarized in the next proposition.

\begin{prop}\label{prop: quartic term}
	The term $\cT_t^*\, \mathcal{V}\, \cT_t$ can be rewritten in the form of
	
	\begin{equation}\label{formula of interaction term in generator}
		\begin{aligned}
			\cT_t^*\, \mathcal{V}\, \cT_t 
			&= C_{N, V}^{(1)} +C_{N, V}^{(2)}+C_{N, V}^{(12)} +\cV+ \mathcal{E}_V\\
			&\quad + \frac{1}{2}\iintd N^{2} V^\lambda_{1} (N(x-y))\,k_{1}(x, y) a_x^* a_y^*\dd x\d y+\mathrm{h.c.}\\
			&\quad + \frac{1}{2}\iintd N^{2} V^\lambda_{2} (N(x-y))\, k_{2}(x, y) b_x^* b_y^*\dd x\d y+\mathrm{h.c.}\\
			&\quad + \iintd N^{2} V^\lambda_{12}(N(x-y))\,k_{12}(x, y) a_x^* b_y^*\dd x\d y+\mathrm{h.c.},
		\end{aligned}   
	\end{equation}
	with the remainder term satisfying
	\begin{equation}\label{est:E_V}
		\begin{aligned}
			\pm \mathcal{E}_V(t) 
			\leq&\, 
			\delta \|\widetilde{\boldsymbol{\phi}}_t^{(N)}
			\|_{L^\infty_x}
			\mathcal{V}
			+
			C_\delta \lambda \|\widetilde{\boldsymbol{\phi}}_t^{(N)}
			\|_{L^\infty_x} 
			\bigg(\frac{\mathcal{N}^2}{N}+\mathcal{N}+1\bigg), \\
			\pm\left[\mathcal{N}, \mathcal{E}_V(t)\right] 
			\leq&\,  
			\delta \|\widetilde{\boldsymbol{\phi}}_t^{(N)}
			\|_{L^\infty_x}
			\mathcal{V}
			+
			C_\delta \lambda \|\widetilde{\boldsymbol{\phi}}_t^{(N)}
			\|_{L^\infty_x}
			\bigg(\frac{\mathcal{N}^2}{N}+\mathcal{N}+1\bigg), \\
			\pm \dot{\mathcal{E}}_V(t) 
			\leq&\, 
			\delta \big(\|\widetilde{\boldsymbol{\phi}}_t^{(N)}
			\|_{L^\infty_x} 
			+ \|\partial_t \widetilde{\boldsymbol{\phi}}_t^{(N)}
			\|_{L^\infty_x} \big)
			\mathcal{V}\\
			&+
			C_\delta \lambda (\|\widetilde{\boldsymbol{\phi}}_t^{(N)}
			\|_{L^\infty_x} 
			+ \|\partial_t \widetilde{\boldsymbol{\phi}}_t^{(N)}
			\|_{L^\infty_x} )
			\bigg(\frac{\mathcal{N}^2}{N}+\mathcal{N}+1\bigg),
		\end{aligned}
	\end{equation}
	where $\mathcal{V}$ is the interaction part in the Hamiltonian, and constant $C_\delta$ depends on $\delta$, $\|\widetilde{\boldsymbol{\phi}}_t^{(N)}
	\|_{H^1_x}$ and $\|\partial_t \widetilde{\boldsymbol{\phi}}_t^{(N)}
	\|_{H^1_x}$.
\end{prop}

\begin{proof}
	To prove the first inequality in \eqref{est:E_V}, we use \eqref{eq:creation_annihilation_expansion} in $\cT_t^*\, \mathcal{V}\, \cT_t$. 
	Then picking up the terms with only $a^\sharp$ and $b^\sharp$ will produce $\cV$ in the RHS of \eqref{formula of interaction term in generator}. 
	
	Next, we consider the terms containing three $a^{\sharp}$ or $b^{\sharp}$.
	In this case, the contribution arising from the following term is treated by
	\begin{equation}\label{three terms}
		\begin{aligned}
			& \iintd  N^2 V^\lambda_1(N(x-y)) \left\langle\Psi, a^*_x a^*_y a_y \alpha_x \Psi\right\rangle  \dd x \d y  \\
			&=
			\iintd  N^2 V^\lambda_1(N(x-y)) \Big\langle a_ya_x \Psi, (\alpha_xa_y + [a_y, \alpha_x]) \Psi\Big\rangle  \dd x \d y  \\
			& \leq
			\Big( \iintd  N^2 V^\lambda_1(N(x-y))\left\|a_x a_y \Psi\right\|^2  \dd x \d y \Big)^{\frac{1}{2}}\\
			&\quad\times 
			\Big(  \iintd  N^2 V^\lambda_1(N(x-y))\left\|(\alpha_xa_y + [a_y, \alpha_x]) \Psi \right\|^2  \dd x \d y \Big)^{\frac{1}{2}}.
		\end{aligned}
	\end{equation}
	Note that by \eqref{pointwise bound}, one has the pointwise bound
	\begin{equation}\label{est: pointwise bound commutator a and alpha}
		\pm [a_y, \alpha_x]\le C |\widetilde{\boldsymbol{\phi}}_t^{(N)}(x)|
		|\widetilde{\boldsymbol{\phi}}_t^{(N)}(y)|.
	\end{equation}
	Then, by substituting \eqref{est: pointwise bound commutator a and alpha} into \eqref{three terms}, applying Lemma~\ref{lemma: quatic}, and Young's inequality, we obtain
	\begin{equation}\label{bound of three a}
		\begin{aligned}
			\text{LHS of }\eqref{three terms}
			\leq&\, 
			\delta \|\widetilde{\boldsymbol{\phi}}_t^{(N)}
			\|_{L^\infty_x} 
			\iintd  N^2 V^\lambda_1(N(x-y))\left\|a_x a_y \Psi\right\|^2  \dd x \d y\\
			& +
			C_{\delta} 
            \lambda
            \|\widetilde{\boldsymbol{\phi}}_t^{(N)}
			\|_{L^\infty_x} 
			\frac{1}{N}\|(\mathcal{N}+1) \Psi\|^2.
		\end{aligned}
	\end{equation}
	On the other hand, for the terms containing $a^*_x a^*_y a_y A^*_x$, using $[a_y, A^*_x] = k_1(y,x)$ by \eqref{eq:commute formula} and writing this commutator term explicitly in the RHS of \eqref{formula of interaction term in generator}, the remainder term can be estimated similarly to \eqref{three terms}. The other terms in this case can be treated with the same argument.
	
	Now let us discuss the terms containing two $a^{\sharp}$ or $b^{\sharp}$.
	In this case, for example, the following term can be treated by
	\begin{align*}
		& \iintd  N^2 V^\lambda_1(N(x-y)) \left\langle\Psi, A_x a^*_y a_y A^*_x \Psi\right\rangle  \dd x \d y  \\
		&=
		\iintd  N^2 V^\lambda_1(N(x-y)) \Big\langle(A^*_xa_y + [a_y, A^*_x])\Psi, (A^*_xa_y + [a_y, A^*_x]) \Psi\Big\rangle  \dd x \d y  \\
		& =
		C_{N,V}^{(1)}
		+
		\iintd  N^2 V^\lambda_1(N(x-y))\left\|A^*_x a_y \Psi\right\|^2  \dd x \d y \\
		&\quad+ 
		\iintd  N^2 V^\lambda_1(N(x-y)) \Big\langle A^*_xa_y \Psi, [a_y, A^*_x] \Psi\Big\rangle  \dd x \d y
		+ \mathrm{h.c.}.
	\end{align*}
	Using \eqref{eq:commute formula} and the pointwise estimate of $k_1$ in \eqref{pointwise bound}, the RHS of the above equality, except for the constant term, has a similar bound to that in \eqref{bound of three a}.
	The other terms in this case can be bounded with the same argument.
	
	Finally, for the terms containing less then two $a^{\sharp}$ or $b^{\sharp}$, similar argument can be applied.
\end{proof}

\subsection{Cancellation between Quadratic and Hamiltonian Terms}
Let us denote $V^{N}_{\ii} = N^2V^\lambda_{\ii}(N\cdot)$.
Summing up the results in Proposition \ref{prop: estimate of quadratic term}, Proposition \ref{prop: Kinetic term} and Proposition \ref{prop: quartic term} leads to
\begin{align*}
	&\cT_t^*\, (\cG_{2} + \cH_N )\, \cT_t \\
	&= N\bigg(
	\iintd \Big( \Delta w_{1,\ell}^{(N)} + \tfrac{1}{2} V_{1}^{N} \big(1-w_{1,\ell}^{(N)}\big) \Big) (x-y) \,\wt\phi^{(N)}_{1,t}(x) \wt\phi^{(N)}_{1,t}(y)a_x^* a_y^*\dd x\d y \\
	&\quad +\iintd \Big( \Delta w_{2,\ell}^{(N)} + \tfrac{1}{2}V_{2}^N \big(1-w_{2,\ell}^{(N)}\big) \Big) (x-y) \,\wt\phi^{(N)}_{2,t}(x) \wt\phi^{(N)}_{2,t}(y)b_x^* b_y^*\dd x\d y \\
	&\quad + 2\iintd \Big(\Delta w_{12,\ell}^{(N)} 
	+\tfrac{1}{2}V_{12}^{N} \big(1-w_{12, \ell}^{(N)}\big) \Big) (x-y) \,\wt\phi^{(N)}_{1,t}(x) \wt\phi^{(N)}_{2,t}(y)a_x^* b_y^*\dd x\d y + \mathrm{h.c.} \bigg)\\
	&\quad + C_{N, 2}+C_{N, K}^{(1)}+C_{N, K}^{(2)}+C_{N, V} +\cH_N+ \mathcal{E}_{2} +
	\mathcal{E}_K + \mathcal{E}_V.
\end{align*} 

Next, we let $N$ large enough such that for any $x\in \mathrm{supp}\{V_{\ii}\}$, we have $|x|<N\ell$.
Using the \eqref{eq:scatlN}, it follows that
\begin{align*}
	\cT_t^*\, (\cG_{2} + \cH_N )\, \cT_t  
	&=\,  N\bigg(
	\frac{1}{2}\iint_{|x-y|<\ell}  \nu_{1,\ell}^{(N)} f_{1,\ell}^{(N)} (x-y) \,\wt\phi^{(N)}_{1,t}(x) \wt\phi^{(N)}_{1,t}(y)a_x^* a_y^*\dd x\d y \\
	&\quad + \frac{1}{2}\iint_{|x-y|<\ell}  \nu_{2,\ell}^{(N)} f_{2,\ell}^{(N)} (x-y) \,\wt\phi^{(N)}_{2,t}(x) \wt\phi^{(N)}_{2,t}(y)b_x^* b_y^*\dd x\d y \\
	&\quad + \iint_{|x-y|<\ell}  \nu_{12,\ell}^{(N)} f_{12,\ell}^{(N)} (x-y) \,\wt\phi^{(N)}_{1,t}(x) \wt\phi^{(N)}_{2,t}(y)a_x^* b_y^*\dd x\d y + \mathrm{h.c.} \bigg)\\
	&\quad + C_{N, 2}+C_{N, K}^{(1)}+C_{N, K}^{(2)}+C_{N, V}+\cH_N+ \mathcal{E}_{2}
	+ \mathcal{E}_K  + \mathcal{E}_V.
\end{align*} 

\begin{prop}\label{prop: cancellation}
	Let us denote the error
	\begin{equation*}
		\mathcal{E}_c
		=
		\cT_t^*\, (\cG_{2} + \cH_N )\, \cT_t
		- (C_{N, 2}+C_{N, K}^{(1)}+C_{N, K}^{(2)}+C_{N, V}+\cH_N + \mathcal{E}_{2} +\mathcal{E}_K+ \mathcal{E}_V).
	\end{equation*}
	Then, we have
	\begin{equation}
		\begin{aligned}
			\pm \mathcal{E}_c(t) & \leq C \|\widetilde{\boldsymbol{\phi}}^{(N)}_{t}\|_{L^\infty_x}  \big(\mathcal{N}+1\big), \\
			\pm\left[\mathcal{N}, \mathcal{E}_c(t)\right] 
			& \leq 
			C \|\wtbphi^{(N)}_t\|_{L^\infty_x} \big(\mathcal{N}+1\big), \\
			\pm \dot{\mathcal{E}}_c(t) 
			& \leq 
			C \|\wtbphi^{(N)}_t\|_{L^\infty_x} \big(\mathcal{N}+1\big).
		\end{aligned}
	\end{equation}
	Here, constant $C$ depends on $\|\wtbphi^{(N)}_t\|_{L^2_x}$ and $\|\partial_t \wtbphi^{(N)}_t\|_{L^2_x}$.
\end{prop}

\begin{proof}
	Write $g_{\ii \jj}(x,y) := N \nu_{\ii\jj,\ell}^{(N)} f_{\ii\jj,\ell}^{(N)} (x-y)  
	\,\wt\phi^{(N)}_{\ii,t}(x) \wt\phi^{(N)}_{\jj,t}(y)$ and  consider 
	\begin{equation}\label{def:lambda_term_example}
		\iint_{|x-y|<\ell}  g_{\ii\jj}(x, y) a_x^* a_y^*\dd x\d y.
	\end{equation}
	Then, utilizing the properties of $f_{\ii\jj,\ell}^{(N)}$ given in Lemma~\ref{lem:neumann_scattering_function}, one has
	\begin{equation*}
		\iint_{|x-y|<\ell} g^2_{\ii\jj}(x,y) \dd x\d y \le C \|\wt\phi^{(N)}_{\ii,t}\|_{L^\infty_x}^2 \|\wt\phi^{(N)}_{\jj,t}\|_{L^2_x}^2 .
	\end{equation*}
	Moreover, by the Cauchy--Schwarz inequality and \eqref{est:creation and annihilation}, we have that
	\begin{equation*}\label{estimate of g}
		\begin{aligned}
			\n{\inprod{\Psi}{\eqref{def:lambda_term_example}\,\Psi}}
			\le&\, 
			\intd \| a_x\Psi\| \| a^*(g_{\ii\jj}(x,\cdot)) \Psi\| \dd x\\
			\le&\, 
			C \Big(\intd \| a_x\Psi\|^2 \dd x  \Big)^{\frac{1}{2}} 
			\Big(\intd \| a^*(g_{\ii\jj}(x,\cdot)) \Psi\|^2  \dd x  \Big)^{\frac{1}{2}}\\
			\le&\,  C
			\|\wt\phi^{(N)}_{\ii,t}\|_{L^\infty_x} \|\wt\phi^{(N)}_{\jj,t}\|_{L^2_x}
			\|(\mathcal{N}+1) \Psi\|^2,
		\end{aligned}
	\end{equation*}
	which is the desired result. 
\end{proof}

\subsection{Estimates for the  \texorpdfstring{$(i\partial_t \cT_t^*)  \cT_t$}{iD(T*) T} Term}
We conclude the estimate of $ (i \partial_t \cT_t^*)  \cT_t$ in the following proposition.
\begin{prop}\label{prop: time derivative}
	There exists a time-dependent constant $C_{N, \chi}(t)$, such that
	\begin{equation*}
		\begin{aligned}
			\big|\big\langle\Psi,\, \left(i\partial_t \cT_t^*\right) \cT_t\, \Psi\big\rangle - C_{N, \chi}(t) \big| 
			& \leq 
			C \|\widetilde{\boldsymbol{\phi}}_t^{(N)}
			\|_{L^\infty_x} (\mathcal{N}+1) ,\\
			\big|\big\langle\Psi,\, \left[\cN,\, \left(i \partial_t \cT_t^*\right) \cT_t\right] \Psi\big\rangle\big| 
			& \leq 
			C \|\widetilde{\boldsymbol{\phi}}_t^{(N)}
			\|_{L^\infty_x}
			(\mathcal{N}+1) ,\\
			\big|\big\langle\Psi,\, \partial_t\(\left(i \partial_t \cT_t^*\right) \cT_t\) \Psi\big\rangle - \partial_t C_{N, \chi}(t)\big| 
			& \leq 
			C \big(\|\partial_t \widetilde{\boldsymbol{\phi}}_t^{(N)}
			\|_{L^\infty_x}
			+ \|\widetilde{\boldsymbol{\phi}}_t^{(N)}
			\|_{L^\infty_x}\big)(\mathcal{N}+1).
		\end{aligned}
	\end{equation*}
	The constant $C$ depends on $\|\widetilde{\boldsymbol{\phi}}_t^{(N)}
	\|_{H^1_x}$, $\|\partial_t \widetilde{\boldsymbol{\phi}}_t^{(N)}
	\|_{H^1_x}$ and $\|\partial_t^2 \widetilde{\boldsymbol{\phi}}_t^{(N)}
	\|_{L^2_x}$.
\end{prop}

\begin{proof}
	By \eqref{eq:time-derivative_lie_identity}, we have that
	\begin{equation}\label{rewrite partial T}
		\left(\partial_t \cT_t^*\right) \cT_t 
		=\, 
		\cI\((\bd_t e^{-\sfK}) e^{\sfK}\) .
	\end{equation}
	Moreover, note that we have
	\begin{align*}
		\(i\bd_te^{-\sfK} \)e^{\sfK}
		=\, 
		\begin{pmatrix}
			i\bd_t\ch(\boldsymbol{k}) & -i\bd_t\sh(\boldsymbol{k})\\
			\conj{i\bd_t\sh(\boldsymbol{k})} & -\conj{i\bd_t\ch(\boldsymbol{k})}
		\end{pmatrix}
		&\times
		\begin{pmatrix}
			\ch(\boldsymbol{k}) & \sh(\boldsymbol{k})\\
			\conj{\sh(\boldsymbol{k})} & \conj{\ch(\boldsymbol{k})}
		\end{pmatrix}
		= 
		\begin{pmatrix}
			\phantom{-}\boldsymbol{\gamma} & \phantom{-}\boldsymbol{\mu}\\
			-\conj{\boldsymbol{\mu}} & -\boldsymbol{\gamma}^\top
		\end{pmatrix}.
	\end{align*}
	Thus writing the RHS of \eqref{rewrite partial T} in normal order, it follows that
	\begin{equation}\label{rewrite partial T in normal order}
		\left(i\partial_t \cT_t^*\right) \cT_t 
		=\, 
		\nor{\cI\( \big(i\bd_t e^{-\sfK} \big)e^{\sfK}\) }
		+ 
		C_{N, \chi}(t),
	\end{equation}
	where $ C_{N, \chi} = \tfrac{1}{2} \tr(\gamma_{11})+\tfrac{1}{2} \tr(\gamma_{22}) $, which is well defined since $\boldsymbol{\gamma}$ is trace class.
	
	To estimate the RHS of \eqref{rewrite partial T in normal order}, we observe that for $f_i\in L^2(\mathbb{R}^3 \times\mathbb{R}^3)$ with $i=1,2,3$, one has
	\begin{multline*}
			\Big|\Big\langle\Psi, \iintd \left(f_1(x, y) a_x^* a_y^*+f_2(x, y) a_x a_y +f_3(x, y) a_x^* a_y\right)\d x \d y\,  \Psi\Big\rangle\Big|\qquad\qquad \\
			\leq
			\Big(\left\|f_1\right\|_{\rm HS}+\left\|f_2\right\|_{\rm HS} +\left\|f_3\right\|_{\rm HS}\Big)
			\langle\Psi,(\mathcal{N}+1) \Psi\rangle.
	\end{multline*}
	
	Similar results hold  in the cases when $a^\sharp$s are replaced by $b^\sharp$s. Applying these results, and note that only off-diagonal terms appear in $\dot{\sfK}$, it follows that
	\begin{equation}\label{estimate of partial_t T}
		\begin{aligned}
			\big|\big\langle\Psi,\, \left(i\partial_t \cT_t^*\right) \cT_t \Psi\big\rangle - C_{N, \chi}(t)\big| 
			\leq&
			\big\| \big(i\bd_t e^{-\sfK} \big)e^{\sfK} \big\|_{\rm HS}  \,
			\big\langle\Psi,(\mathcal{N}+1) \Psi\big\rangle \\
			\leq&\, 
			\|\dot{\sfK} \|_{\rm HS}\, e^{2 \|\sfK\|_{\rm HS}} \,
			\big\langle\Psi,(\mathcal{N}+1) \Psi\big\rangle \, .
		\end{aligned}
	\end{equation}
	
	By the same argument, we also obtain
	\begin{equation*}
		\begin{aligned}
			\big|\big\langle\Psi, \, \partial_t\( \left(i\partial_t \cT_t^*\right) \cT_t\) \Psi\big\rangle- \partial_t C_{N, \chi}(t) \big| 
			\leq&\, 
			\big(2 \|\dot{\sfK} \|_{\rm HS}^2 
			+ \|\ddot{\sfK} \|_{\rm HS}\big)
			e^{2 \|\sfK\|_{\rm HS}} \,
			\big\langle\Psi,(\mathcal{N}+1) \Psi\big\rangle \, .
		\end{aligned}
	\end{equation*}
	
	Then using the fact $\|\sfK\|_{\rm HS}\le \sqrt{2}\|\boldsymbol{k}_t\|_{\rm HS}$ and similar for $\dot{\sfK}$ and $\ddot{\sfK}$, applying Lemma \ref{bound of derivative}, we finally prove the desired theorem.
\end{proof}

\appendix
\section{Normal Order of the Fluctuation Hamiltonian}
This appendix is independent of Section~\ref{sect:bounds_fluctuation_dynamics} because the calculations presented here, while compact, neither simplify the computations nor the analysis in that section. Nevertheless, it provides an alternative verification of the calculations performed in Section~\ref{sect:bounds_fluctuation_dynamics}. Specifically, we present explicit computations of the normal ordering of the full fluctuation Hamiltonian using techniques introduced in \cite{grillakis2013beyond, grillakis2013pair, grillakis2017pair}. Since the time variable plays a passive role in this discussion, we suppress its dependence in our notation to simplify the presentation where no confusion arises.

We write 
\begin{align*}
	\cL_N = N\mu_0 + \cT^\ast\(\cG_1+\cG_3\)\cT+(i\bd_t\cT^\ast)\cT +\cT^\ast\cQ_{\rm Bog}\cT+ \cT^\ast\cV\cT \quad \text{ with } \quad \cQ_{\rm Bog} =  \cK + \cG_2. 
\end{align*}
Let $f, g \in L^2(\R^3)$ and $\vect{f}=f\oplus 0,\, \vect{g} = 0 \oplus g \in \fH$. Using \eqref{eq:bogoliubov_conjugation_identities}, we define the following operators
\begin{equation}
	\begin{aligned}
		c(f):=&\, \cT^* z(\vect{f}) \cT=\,z(\ch(\boldsymbol{k})\vect{f})+z^\ast(\sh(\boldsymbol{k})\conj{\vect{f}}),\\
		d(g):=&\, \cT^* z(\vect{g}) \cT=\,z(\ch(\boldsymbol{k})\vect{g})+z^\ast(\sh(\boldsymbol{k})\conj{\vect{g}}),
	\end{aligned}
\end{equation}
and the corresponding operator-valued distributions 
\begin{equation}
	c_x :=\, \cT^\ast a_x\cT, \quad 
	d_x :=\, \cT^\ast b_x \cT, \quad \text{ and } \quad \zeta_x :=\, \cT^\ast z_x \cT.
\end{equation}

Using \eqref{liemap}, the quadratic part can be further rewritten as 
\begin{align*}
	\cQ_{\rm Bog} =&\, \cH_{\sfG}+\cM,\\
	\cH_{\sfG}
	=&\, \frac12\iintd (z_x^\ast)^\top \vect{G}(x, y) z_y\dd x\d y+\frac12\iintd (z_y^\ast)^\top \vect{G}(y, x) z_x\dd x\d y,\\
	\cM
	=&\, \frac12 \iintd (z^\ast_x)^\top \vect{M}(x, y) z^\ast_y\dd x\d y+\frac12 \iintd z_x^\top \conj{\vect{M}(x, y)} z_y\dd x\d y=\cI(\sfM),
\end{align*}
where 
\begin{align*}
	\sfG =&\, 
	\begin{pmatrix}
		\vect{G} & 0\\
		0 & -\vect{G}^\top
	\end{pmatrix}
	\quad \text{ and } \quad 
	\sfM = 
	\begin{pmatrix}
		0 & \vect{M}\\
		-\conj{\vect{M}} & 0
	\end{pmatrix} \quad \text{ with }\\
	\vect{G} =&\, 
	\begin{pmatrix}
		g^{(1)}_N & g^{(12)}_N\\[1ex]
		g^{(21)}_N & g^{(2)}_N
	\end{pmatrix}
	\quad \text{ and } \quad 
	\vect{M} = 
	\begin{pmatrix}
		m^{(1)}_N & m^{(12)}_N\\[1ex]
		m^{(21)}_N & m^{(2)}_N
	\end{pmatrix}.
\end{align*}
Here, we have the kernels of the above operators
\begin{align*}
	g^{(\ii)}_N(x, y) =&\, \(-\lapl_x + V^{(\ii)}_{\mathrm{c}}(x)\)\delta(x-y)+N^3V_\ii(N(x-y))\wt\phi^{(N)}_{\ii}(x)\conj{\wt\phi^{(N)}_{\ii}(y)},\\
	V^{(\ii)}_{\mathrm{c}}(x)=&\, N^3V_{\ii1}(N\cdot)\ast |\wt\phi^{(N)}_{1}|^2(x)+N^3V_{\ii2}(N\cdot)\ast |\wt\phi^{(N)}_{2}|^2(x),\\
	g^{(\ii\jj)}_N(x, y) =&\, N^3V_{12}(N(x-y))\wt\phi^{(N)}_{\ii}(x) \conj{\wt\phi^{(N)}_{\jj}(y)},\\
	m^{(\ii)}_N(x, y) =&\, N^3V_{\ii}(N(x-y))\wt\phi^{(N)}_{\ii}(x)\wt\phi^{(N)}_{\ii}(y),\\
	m^{(\ii\jj)}_N(x, y) =&\, N^3V_{12}(N(x-y))\wt\phi^{(N)}_{\ii}(x)\wt\phi^{(N)}_{\jj}(y).
\end{align*}

\subsection{The cubic terms} 
Here, we compute the normal order of $\cT^\ast \cQ_3 \cT$ up to the linear in creation and annihilation terms. We start by writing
\begin{equation}\label{eq:conjugated_cubic_terms}
	\begin{aligned}
		\cT^\ast\cG_{3}^{(1)}\cT=&\, \frac{1}{\sqrt{N}}\iintd N^{3}V_{1}(N(x-y))  \wt\phi^{(N)}_{1}(y)\, c^{\ast}_y c^{\ast}_x   c_x \dd x\d y +\mathrm{h.c.}\\
		\cT^\ast\cG_{3}^{(2)}\cT=&\, \frac{1}{\sqrt{N}}\iintd N^{3}V_{2}(N(x-y)) \wt\phi^{(N)}_{2}(y)\, d^{\ast}_y d^{\ast}_x   d_x \dd x\d y+\mathrm{h.c.}\\
		\cT^\ast\cG^{(12)}_{3}\cT=&\,  \frac{1}{\sqrt{N}}\iintd N^{3} V_{12}(N(x-y))\,\wt\phi^{(N)}_{2}(y)\,  c^{\ast}_x   c_xd^{\ast}_y\dd x\d y\\
		&+ \frac{1}{\sqrt{N}}\iintd N^{3} V_{12}(N(x-y))\,\wt\phi^{(N)}_{1}(y)\, c^{\ast}_y d^{\ast}_x d_x\dd x\d y+ \mathrm{h.c.}.
	\end{aligned}
\end{equation}
Let us now use Wick's theorem to put terms of \eqref{eq:conjugated_cubic_terms} in normal order. Define the contraction of $Z(\vect{F}) := z(\vect{f}_1)+z^\ast(\vect{f}_2)$ and $Z(\vect{H}):= z(\vect{h}_1)+z^\ast(\vect{h}_2)$ to be $C(Z(\vect{F}), Z(\vect{H})) = [z(\vect{f}_1), z^\ast(\vect{h}_2)] = \inprod{\vect{f}_1}{\vect{h}_2}$ and denote the normal ordering of an operator $\cP$ by $\nor{\cP}$. Then we need to normal order the terms using Wick's Theorem  which says that 
\begin{multline}\label{eq:cubic_Wick's_Theorem}
	Z(\vect{F})Z(\vect{H})Z(\vect{J})= \nor{Z(\vect{F})Z(\vect{H})Z(\vect{J})} \\
	+C(Z(\vect{F}), Z(\vect{H}))Z(\vect{J})+C(Z(\vect{F}), Z(\vect{J}))Z(\vect{H})+C(Z(\vect{H}), Z(\vect{J}))Z(\vect{F}).
\end{multline}
More precisely, we have 
\begin{align*}
	c^{\ast}_y c^{\ast}_x   c_x =&\, \nor{c^{\ast}_y c^{\ast}_x   c_x}+C(c^{\ast}_y, c^{\ast}_x)   c_x + C(c^{\ast}_y, c_x) c^{\ast}_x +C(c^{\ast}_x, c_x)c^{\ast}_y, \\
	c^{\ast}_x   c_x d^{\ast}_y =&\, \nor{c^{\ast}_x   c_x d^{\ast}_y }+C(c^{\ast}_x, d^{\ast}_y)   c_x+C(c_x, d^{\ast}_y) c^{\ast}_x +C(c^{\ast}_x, c_x) d^{\ast}_y, \\
	c^{\ast}_y d^{\ast}_x d_x=&\, \nor{c^{\ast}_yd^{\ast}_x d_x}+C(c^{\ast}_y, d^{\ast}_x)   d_x+C(c^{\ast}_y, d_x) d^{\ast}_x  +C(d^{\ast}_x, d_x)c^{\ast}_y. 
\end{align*}
which follows by making the formal observation that 
\[
    c_x = c(\delta_x) =       z(\ch(\boldsymbol{k})\vect{f})+z^\ast(\sh(\boldsymbol{k})\conj{\vect{f}})=Z((\ch(\boldsymbol{k})\vect{f}, \sh(\boldsymbol{k})\conj{\vect{f}}))
\]
with $\vect{f} = (\delta_x, 0)^\top$.

Moreover, by direct computation, we have that 
\begin{equation}\label{eq:intra_contraction}
	\begin{aligned}
		C\(c(f_1), c(f_2)\) =&\,   
		\tfrac12\inprod{\vect{f}_1\otimes \vect{f}_2}{\sh(2\boldsymbol{k})}_{\fH^{\otimes 2}}\\
		C\(c^\ast(f_1), c^\ast(f_2)\) =&\,   
		\tfrac12\inprod{\sh(2\boldsymbol{k})}{\vect{f}_2\otimes \vect{f}_1}_{\fH^{\otimes 2}}\\
		C\(c^\ast(f_1), c(f_2)\) =&\,   
		\tfrac12\inprod{\vect{f}_2}{\operatorname{p}(2\boldsymbol{k}))\, \vect{f}_1}.
	\end{aligned}
\end{equation}
The calculation is similar for contractions that involve only $d$s. 
For the cross terms, we have 
\begin{equation}\label{eq:inter_contraction}
	\begin{aligned}
		C\(c^\ast(f), d^\ast(g)\)=&\,  
		\tfrac12 \inprod{\sh(2\boldsymbol{k})}{\vect{g}\otimes \vect{f}}_{\fH^{\otimes 2}},\\
		C\(c(f), d(g)\)=&\,  
		\tfrac12 \inprod{\vect{g}\otimes \vect{f}}{\sh(2\boldsymbol{k})}_{\fH^{\otimes 2}},\\
		C\(c(f), d^\ast(g)\)=&\,  
		\inprod{\vect{f}}{\ch(\boldsymbol{k})^2\, \vect{g}} =
		\tfrac12\inprod{\vect{f}}{\operatorname{p}(2\boldsymbol{k}))\, \vect{g}},\\
		C\(c^\ast(f), d(g)\)=&\, 
		\tfrac12\inprod{\vect{g}}{\operatorname{p}(2\boldsymbol{k}))\, \vect{f}}.
	\end{aligned}
\end{equation}

Given the operator $\vect{A} = (A_{\ii\jj})_{\ii,\jj\in\{1, 2\}}$, we define $V_N\vect{A}$ by the matrix kernel 
\begin{align*}
	(V_N\vect{A})(x, y)
	=
	\begin{pmatrix}
		N^3V_1(N(x-y))A_{11}(x, y) & N^3V_{12}(N(x-y))A_{12}(x, y) \\[1ex]
		N^3V_{12}(N(x-y))A_{21}(x, y) & N^3V_2(N(x-y))A_{22}(x, y)
	\end{pmatrix}.
\end{align*}
We also define the effective potential matrix 
\begin{align*}
	\vect{G}_{\mathrm{p}} = 
	\begin{pmatrix}
		V^{(1)}_{\rm p} & 0\\
		0 & V^{(2)}_{\rm p}
	\end{pmatrix}
\end{align*}
where
\begin{align*}
	V_{\mathrm{p}}^{(\ii)}(x):= \frac{1}{2N} \intd N^3V_{\ii1}(N(x-y))\operatorname{p}(2\boldsymbol{k}))_{\ii1}(y, y)+N^3V_{\ii2}(N(x-y))\operatorname{p}(2\boldsymbol{k}))_{\ii2}(y, y)\dd y.
\end{align*}
Now, we write 
\begin{multline}\label{eq:normal ordered conjugated cubic terms}
	\cT^\ast\cG_3\cT =\,  -\cT^\ast\cG_1\cT+\nor{\cT^\ast\cG_3\cT} \\
	+ \(\sqrt{N} \zeta^\ast\(\tfrac{1}{2N} (V_N \operatorname{r}(2\boldsymbol{k}))\,\conj{\wtbphi^{(N)}}+ \tfrac{1}{2N} (V_N \operatorname{p}(2\boldsymbol{k}))\, \wtbphi^{(N)}+ \vect{G}_{\mathrm{p}}\,\wtbphi^{(N)} \)+\mathrm{h.c.}\).
\end{multline}

\subsection{The quartic terms} It remains to normal order $\cT^\ast \cV\cT$. For the quartic terms, we start by writing
\begin{equation}\label{eq:conjugated_quartic_terms}
	\begin{aligned}
		\cT^*\, \mathcal{V}\, \cT
		=&\, \frac{1}{2}\iintd N^{2} V_{1} (N(x-y))\, c_x^\ast c_y^\ast c_y c_x\dd x\d y\\
		&\, +\frac{1}{2}\iintd N^{2} V_{2} (N(x-y))\, d_x^\ast d_y^\ast d_y d_x\dd x\d y\\
		&\, + \iintd N^{2} V_{12}(N(x-y))\,c^\ast_x c_x d_y^\ast d_y\dd x \d y\, .
	\end{aligned}   
\end{equation}
Then, we use the following Wick's Theorem 
\begin{multline}\label{eq:quartic_Wick's_Theorem}
	Z(\vect{F})Z(\vect{H})Z(\vect{J})Z(\vect{K})= \nor{Z(\vect{F})Z(\vect{H})Z(\vect{J})Z(\vect{K})} \\
	+C(Z(\vect{F}), Z(\vect{H}))Z(\vect{J})Z(\vect{K})+\ldots+C(Z(\vect{J}), Z(\vect{K}))Z(\vect{F})Z(\vect{H}).
\end{multline}
Notice, we have that 
\begin{align*}
	c^\ast_{x}c^\ast_{y}c_{y}c_{x}=&\, \nor{c^\ast_{x}c^\ast_{y}c_{y}c_{x}}
	+ C(c^\ast_{x}, c^\ast_{y})c_{y}c_{x}+C(c_{y}, c_{x})c^\ast_{x}c^\ast_{y}\\
	&\, +C(c^\ast_{x}, c_{x})c^\ast_{y}c_{y}
	+ C(c^\ast_{y}, c_{y})c^\ast_{x}c_{x}
	+C(c^\ast_{y}, c_{x})c^\ast_{x}c_{y}+C(c^\ast_{x}, c_{y})c^\ast_{y}c_{x}, \\
	c^\ast_{x}c_{x}d^\ast_{y}d_{y}=&\, \nor{c^\ast_{x}c_{x}d^\ast_{y}d_{y}}
	+ C(c^\ast_{x}, c_{x})d^\ast_{y}d_{y}+C(d_{y}^\ast, d_{y})c^\ast_{x}c_{x}\\
	&\, +C(c^\ast_{x}, d^\ast_{y})c_{x}d_{y}
	+ C(c_{x}, d^\ast_{y})c^\ast_{x}d_{y}
	+C(c_{x}, d_{y})c^\ast_{x}d^\ast_{y}+C(c^\ast_{x}, d_{y})c_{x}d^\ast_{y}.
\end{align*}
Again, the contractions can be computed using \eqref{eq:intra_contraction} and \eqref{eq:inter_contraction}.

Define the operators
\begin{align}
	\wt\Gamma^{(N)}:=&\, |\wtbphi^{(N)}\rangle\!\langle \wtbphi^{(N)}|+\frac{1}{2N}\operatorname{p}(2\boldsymbol{k})\quad \text{ and } \quad
	\wt\Lambda^{(N)}:=\,\wtbphi^{(N)}\otimes\wtbphi^{(N)}+ \frac{1}{2N}\sh(2\boldsymbol{k}).
\end{align}

Hence after the normal ordering of the cubic and quartic terms, we write the generator as follows
\begin{align}
	\cL_N =&\ N\mu_0 \notag\\
	&\, +\(\sqrt{N} \zeta^\ast\(\tfrac{1}{2N} (V_N \operatorname{r}(2\boldsymbol{k}))\,\conj{\wtbphi^{(N)}}+ \tfrac{1}{2N} (V_N \operatorname{p}(2\boldsymbol{k}))\, \wtbphi^{(N)}+ \vect{G}_{\mathrm{p}}\,\wtbphi^{(N)} \)+\mathrm{h.c.}\)\notag\\
	&\, +(i\bd_t\cT^\ast)\cT \label{def:quadratic_term_1}\\
	&\, +\frac12\iintd (\zeta_x^\ast)^\top \widetilde{\vect{G}}(x, y) \zeta_y\dd x\d y+\frac12\iintd (\zeta_y^\ast)^\top \widetilde{\vect{G}}(y, x) \zeta_x\dd x\d y,\label{def:quadratic_term_2}\\
	&\, +\frac12 \iintd (\zeta^\ast_x)^\top \widetilde{\vect{M}}(x, y) \zeta^\ast_y\dd x\d y+\frac12 \iintd \zeta_x^\top \conj{\widetilde{\vect{M}}(x, y)} \zeta_y\dd x\d y \label{def:quadratic_term_3}\\
	&\, + \nor{\cT^\ast \cG_3 \cT+\cT^\ast\cV\cT} \notag
\end{align}
where 
\begin{align*}
	\widetilde{\vect{G}} =\,  
	\begin{pmatrix}
		-\lapl+ V_{\mathrm{c}}^{(1)}+ V_{\mathrm{p}}^{(1)} & 0\\
		0 & -\lapl +V_{\mathrm{c}}^{(2)}+ V_{\mathrm{p}}^{(2)}
	\end{pmatrix}+V_N\wt\Gamma^{(N)}\quad \text{ and } \quad
	\widetilde{\vect{M}} =\, V_N\wt\Lambda^{(N)}.
\end{align*}

Using the mapping $\cI$ and Proposition~\ref{prop: bogoliubov transformation}, we could recast the quadratic terms as follow  
\begin{align*}
	\wt\cQ_{\rm Bog}:=&\, \eqref{def:quadratic_term_1}+\eqref{def:quadratic_term_2}+\eqref{def:quadratic_term_3} \\
	=&\, \cH_{\wt\sfG}+(i\bd_t\cT^\ast)\cT +\com{\cT^\ast, \cH_{\wt\sfG}}\cT+\cT^\ast\cI(\wt\sfM)\cT
	= \cH_{\wt\sfG}+\cI\(\sfR\) 
\end{align*}
where
\begin{align*}
	\sfR =&\, \(i\bd_te^{-\sfK}+\com{e^{-\sfK}, \wt\sfG}+e^{-\sfK}\wt\sfM \)e^{\sfK}\\
	=&\, \Bigg(
	\begin{pmatrix}
		i\bd_t\ch(\boldsymbol{k}) & -i\bd_t\sh(\boldsymbol{k})\\
		\conj{i\bd_t\sh(\boldsymbol{k})} & -\conj{i\bd_t\ch(\boldsymbol{k})}
	\end{pmatrix}\\
	&\, +
	\begin{pmatrix}
		-\com{ \wt{\vect{G}}, \ch(\boldsymbol{k})}+\sh(\boldsymbol{k})\conj{\wt{\vect{M}}} & \sh(\boldsymbol{k})\wt{\vect{G}}^\top+\wt{\vect{G}}\sh(\boldsymbol{k})+\ch(\boldsymbol{k})\wt{\vect{M}}\\
		-\wt{\vect{G}}^\top\conj{\sh(\boldsymbol{k})}-\conj{\sh(\boldsymbol{k})}\wt{\vect{G}}-\conj{\ch(\boldsymbol{k})}\conj{\wt{\vect{M}}} & \com{ \wt{\vect{G}}, \ch(\boldsymbol{k})}^\top - \;\conj{\sh(\boldsymbol{k})}\wt{\vect{M}}
	\end{pmatrix}
	\Bigg)\\
	&\times
	\begin{pmatrix}
		\ch(\boldsymbol{k}) & \sh(\boldsymbol{k})\\
		\conj{\sh(\boldsymbol{k})} & \conj{\ch(\boldsymbol{k})}
	\end{pmatrix}
	= 
	\begin{pmatrix}
		\phantom{-}\boldsymbol{\omega} & \phantom{-}\boldsymbol{\sigma}\\
		-\conj{\boldsymbol{\sigma}} & -\boldsymbol{\omega}^\top
	\end{pmatrix}
\end{align*}
with 
\begin{align*}
	\boldsymbol{\omega} =&\, \(i\bd_t\ch(\boldsymbol{k})-\com{\wt{\vect{G}}, \ch(\boldsymbol{k})}+\sh(\vect{k})\conj{\wt{\vect{M}}}\)\ch(\boldsymbol{k})\\
	&\, +\(-i\bd_t\sh(\boldsymbol{k})+\sh(\boldsymbol{k})\wt{\vect{G}}^\top+\wt{\vect{G}}\sh(\boldsymbol{k})+\ch(\boldsymbol{k})\wt{\vect{M}}\)\conj{\sh(\boldsymbol{k})} = (\omega_{\ii\jj})_{\ii,\jj \in \{1, 2\}}, \\
	\boldsymbol{\sigma} =&\, \(i\bd_t\ch(\boldsymbol{k})-\com{\wt{\vect{G}}, \ch(\boldsymbol{k})}+\sh(\vect{k})\conj{\wt{\vect{M}}}\)\sh(\boldsymbol{k})\\
	&\, +\(-i\bd_t\sh(\boldsymbol{k})+\sh(\boldsymbol{k})\wt{\vect{G}}^\top+\wt{\vect{G}}\sh(\boldsymbol{k})+\ch(\boldsymbol{k})\wt{\vect{M}}\)\conj{\ch(\boldsymbol{k})}=(\sigma_{\ii\jj})_{\ii,\jj \in \{1, 2\}}.
\end{align*}
Lastly, if we normal order the quadratic terms, then we obtain 
\begin{equation}
	\begin{aligned}
		\cL_N:=&\, C_N\\
		&\,  +\(\sqrt{N} \zeta^\ast\(\tfrac{1}{2N} (V_N \operatorname{r}(2\boldsymbol{k}))\,\conj{\wtbphi^{(N)}}+ \tfrac{1}{2N} (V_N \operatorname{p}(2\boldsymbol{k}))\, \wtbphi^{(N)}+ \vect{G}_{\mathrm{p}}\,\wtbphi^{(N)} \)+\mathrm{h.c.}\)\\
		&\, + \nor{\wt\cQ_{\rm Bog}+\cT^\ast \cG_3 \cT+\cT^\ast\cV\cT}
	\end{aligned}
\end{equation}
with the constant $C_N$ is given by 
\begin{align}\label{eq:C_N term}
	C_N = N\(\mu_0 + \tfrac{1}{2N} \tr(\omega_{11})+\tfrac{1}{2N} \tr(\omega_{22})\).
\end{align}

\section{Two-component NLSE Interaction Morawetz Estimates}\label{app:interaction_morawetz}

Let $\bphi_t=(\phi_{1, t}, \phi_{2, t})^\top$ be a solution to \eqref{eq:GP_system_vector_form} (with $\frac12$ in front of the $\lapl$ for convenience). Since the time variable is relatively passive in the calculation below, we suppress the dependence on $t$ in the notation. 

Define the pseudo energy-stress tensor
\begin{align*}
    \rho =&\, \n{\bphi}^2, &&
    J_{k}=\, \im\(\conj{\bphi}^\top\bd_{x_k}\bphi\)\,\, \text{ with } \,\, \vect{J} = (J_1, J_2, J_3),\\
    p =&\, \tfrac14\lapl \n{\bphi}^2-\tfrac12\conj{\bphi}^\top \vect{F}(\bphi)\bphi, && \sigma_{jk}=\, \re\(\conj{\bd_{x_j}\bphi}^\top \bd_{x_k}\bphi\) \,\, \text{ with } \,\, \boldsymbol{\sigma}= (\sigma_{ij})_{i,j \in \{1, 2, 3\}}.
\end{align*}
A direct computation using \eqref{eq:GP_system_vector_form} verifies, at least for smooth $\bphi$, the local conservation laws
\begin{align}\label{eq:local_conservation_law}
\begin{cases}
    \bd_t \rho+\grad\cdot \vect{J} = 0, \\[.5ex]
    \bd_t \vect{J}+\grad\cdot (\boldsymbol{\sigma}-p\vect{I})=0.
\end{cases}
\end{align}

Define the viriel interaction potential associated to the observable $a(x)=\n{x}$ given by
\begin{align}
    V_a(t)= \iintd\, \rho(x)\, a(x-y)\, \rho(y)\dd x\d y.
\end{align}
Differentiating $V_a(t)$ and using  \eqref{eq:local_conservation_law} yields the Morawetz action 
\begin{align}
    M_a(t) := \dot V_a(t) = \iintd \grad a(x-y)\cdot \(\vect{J}(x)\, \rho( y)-\vect{J}(y)\, \rho(x)\)\d x\d y.
\end{align}
Again, differentiating $M_a$ and applying \eqref{eq:local_conservation_law} yields the Morawetz identity
    \begin{align}
            \dot M_a(t)=&\, 
    4\pi\intd \rho(x)^2\dd x\notag\\
    &\, +\iintd \lapl a(x-y)\, \(\conj{\bphi}^\top \vect{F}(\bphi)\bphi\)(x)\, \rho(y)\dd x\d y\label{eq:one-particle_Morawetz_identity_term2}\\
    &\, +\iintd \grad^2 a(x-y): \lt\{\boldsymbol{\sigma}(x)\rho(y)+\boldsymbol{\sigma}(y)\rho(x)-2\,\vect{J}(x)\otimes \vect{J}(y)\rt\}\d x\d y\,.\label{eq:one-particle_Morawetz_identity_term3}
    \end{align}
Here, $:$ denotes the standard double-dot product for $n\times n$ matrices, i.e., $B:C = \sum_{i, j}b_{ij}c_{ij}$. Moreover, we have used the fact that $-\lapl^2 a(x) = 8\pi\, \delta(x)$.

Notice that the second term \eqref{eq:one-particle_Morawetz_identity_term2} is positive. Let us also show that the third term is also nonnegative. First, notice that the matrix $\vect{A}(x, y)=\grad^2 a(x-y)$ is positive semi-definite, then we could rewrite the term as follows 
\begin{subequations}
    \begin{align}
        \eqref{eq:one-particle_Morawetz_identity_term3} =&\, \sum^2_{\ii,\jj=1}\int_{\R^{12}} \delta(x-x')\delta(y-y')\nonumber\\
        &\, \times \vect{A}(x, y): \Big\{(\grad_y-\grad_{x})(\grad_{y'}-\grad_{x'})^\top\phi_\ii(x)\conj{\phi_\ii(x')}\phi_\jj(y)\conj{\phi_\jj(y')}\label{def:higher-order_term_in_Morawetz_id_subterm1}\\
        &\, +(\grad_x\grad_{x'}^\top + \grad_y \grad_{y'}^\top)\,\phi_\ii(x)\conj{\phi_\ii(x')}\phi_\jj(y)\conj{\phi_\jj(y')}\label{def:higher-order_term_in_Morawetz_id_subterm2}\\
        &\, +(\grad_x\grad_{y}^\top + \grad_{x'} \grad_{y'}^\top)\,\phi_\ii(x)\conj{\phi_\ii(x')}\phi_\jj(y)\conj{\phi_\jj(y')}\Big\}\dd x\d x' \d y \d y'\label{def:higher-order_term_in_Morawetz_id_subterm3}
    \end{align}
\end{subequations}
where $\grad$ is used to denote a column vector. Notice that \begin{align*}
    \eqref{def:higher-order_term_in_Morawetz_id_subterm1}
    =\sum^2_{\ii,\jj=1}\iintd\Nrm{\vect{A}(x, y)^\frac12\, (\grad_y-\grad_{x})\phi_\ii(x)\phi_\jj(y)}{}^2\dd x \d y 
\end{align*}
and 
\begin{align*}
    \eqref{def:higher-order_term_in_Morawetz_id_subterm2}
    =\sum^2_{\ii,\jj=1}\iintd\Nrm{\vect{A}(x, y)^\frac12\, \grad_{x}\phi_\ii(x)\phi_\jj(y)}{}^2+\Nrm{\vect{A}(x, y)^\frac12\, \phi_\ii(x)\grad_{y}\phi_{\jj}(y)}{}^2\dd x \d y
\end{align*}
where $\vect{A}^\frac12$ is the unique squareroot of $\vect{A}$. Finally, we see that 
\begin{align*}
    \n{\eqref{def:higher-order_term_in_Morawetz_id_subterm3}}\le 2\sum^2_{\ii,\jj=1}\iintd\Nrm{\vect{A}(x, y)^\frac12\, \grad_{x}\phi_\ii(x)\phi_\jj(y)}{}\Nrm{\vect{A}(x, y)^\frac12\, \phi_\ii(x)\grad_{y}\phi_\jj(y)}{}\dd x\d y
\end{align*}
which means 
\begin{align*}
    \eqref{eq:one-particle_Morawetz_identity_term3} \ge& \sum^2_{\ii,\jj=1}\iintd\Nrm{\vect{A}(x, y)^\frac12\, (\grad_y-\grad_{x})\phi_{\ii}(x)\phi_{\jj}(y)}{}^2\dd x \d y \\
    +\sum^2_{\ii,\jj=1}\iintd\(\Nrm{\vect{A}(x, y)^\frac12\, \grad_{x}\phi_\ii(x)\phi_\jj(y)}{}-\Nrm{\vect{A}(x, y)^\frac12\, \phi_\ii(x)\grad_{y}\phi_\jj(y)}{}\)^2\dd x\d y\ge 0.
\end{align*}

Finally, let us complete the proof of the interaction Morawetz estimate. By the Morawetz identity, we have the estimate 
\begin{align*}
    4\pi\int^T_{-T}\intd \rho(x)^2\dd x\d t \le M_a(T)-M_{a}(-T).
\end{align*}
Finally, by a standard momentum-type estimate (see \cite[Lemma A.10]{tao2006nonlinear}) and the conservation laws, we arrive at the estimate 
\begin{align*}
    \n{M(t)} \le C \Nrm{\bphi_t}{L^2_x}^2\Nrm{\bphi_t}{H^\frac12_x}^2 \le C
\end{align*}
for all $t\ge 0$.
Hence, this yields the following result.
\begin{prop}
    Let $\bphi_t$ be a global solution to \eqref{eq:GP_system_vector_form} with $H^1$ initial data.  Then there exists $C>0$, dependent only on $\|\bphi_0\|_{H^1_x}$, such that we have the space-time estimate 
    \begin{align}
        \|\bphi_t\|_{L^4_{t, x}(\R\times\R^3)}\le C.
    \end{align}
\end{prop}

\hspace{1.5em}

\subsection*{Acknowledgments}

The authors thank Alessandro Olgiati for helpful discussions. 

\subsection*{Fundings}
J. Chong is partially supported by the National Key R\&D Program of China, Project Number 2024YFA1015500. 
J. Lee is partially supported by
the Swiss National Science Foundation through the NCCR SwissMAP,
the SNSF Eccellenza project PCEFP2\_181153, 
by the Swiss State Secretariat for Research and Innovation through the project P.530.1016 (AEQUA), and
Basic Science Research Program through the National Research Foundation of Korea (NRF) funded by the Ministry of Education (RS-2024-00411072).
Z. Sun is partially supported by the Austrian Science Fund (FWF), grant DOI 10.55776/P33010 and 10.55776/F65, as well as by the European Research Council (ERC) under the European Union's Horizon 2020 research and innovation programme, ERC Advanced Grant NEUROMORPH, no. 101018153.

\subsection*{Conflict of interest} The authors have no conflicts of interest to declare that are relevant to the content of this article.

\renewcommand{\bibname}{\centerline{Bibliography}}
\bibliographystyle{abbrv} 
\bibliography{bec,physics}
\end{document}